\documentclass[11pt,reqno]{amsart}
\usepackage{amsaddr}
\usepackage[left=2cm,right=2cm,top=2cm,bottom=2cm]{geometry}
\usepackage[shortlabels]{enumitem}
\usepackage{amsmath,amsthm,amssymb,bm,bbm,dsfont}
\usepackage{mathtools}\mathtoolsset{centercolon}
\mathtoolsset{showonlyrefs=true}
\usepackage[mathscr]{eucal}
\usepackage{upgreek}
\usepackage{setspace}
\usepackage{graphicx}
\usepackage{verbatim}
\usepackage{float}
\usepackage{placeins}
\usepackage{array}
\usepackage{booktabs}

\usepackage{threeparttable}
\usepackage[update,prepend]{epstopdf}
\usepackage{multirow}
\usepackage{amsfonts,amssymb,dsfont}
\usepackage[abs]{overpic}
\usepackage{xfrac}

\usepackage[usenames,dvipsnames]{xcolor}
\usepackage{colortbl}

\definecolor{dullmagenta}{rgb}{0.4,0,0.4}   
\definecolor{darkblue}{rgb}{0,0,0.4}
\usepackage[bookmarks,colorlinks,breaklinks]{hyperref}
\hypersetup{
	unicode=false,          
	pdftoolbar=true,        
	pdfmenubar=true,        
	pdffitwindow=false,     
	pdfstartview={FitH},    
	pdftitle={My title},    
	pdfauthor={Author},     
	pdfsubject={Subject},   
	pdfcreator={Creator},   
	pdfproducer={Producer}, 
	pdfkeywords={keyword1} {key2} {key3}, 
	pdfnewwindow=true,      
	colorlinks=true,        
	linkcolor=Red,          
	citecolor=ForestGreen,  
	filecolor=Magenta,      
	urlcolor=BlueViolet,    
}

\usepackage{doi}
\usepackage{caption, subcaption}
\usepackage{enumitem}

\makeatletter
\newcommand{\opnorm}{\@ifstar\@opnorms\@opnorm}
\newcommand{\@opnorms}[1]{%
	\left|\mkern-1.5mu\left|\mkern-1.5mu\left|
	#1
	\right|\mkern-1.5mu\right|\mkern-1.5mu\right|
}
\newcommand{\@opnorm}[2][]{%
	\mathopen{#1|\mkern-1.5mu#1|\mkern-1.5mu#1|}
	#2
	\mathclose{#1|\mkern-1.5mu#1|\mkern-1.5mu#1|}
}
\makeatother

\makeatletter
\ifx\@NODS\undefined%

\let\Diamond=\diamondsuit
\else%
\fi
\makeatother

\makeatletter
\ifx\@NODS\undefined%

\let\mathbb=\mathds
\else%
\fi
\makeatother

\DeclarePairedDelimiter{\ceil}{\lceil}{\rceil}

\DeclareMathOperator*{\argmax}{\arg\max}
\DeclareMathOperator*{\argmin}{\arg\min}

\DeclareMathOperator{\Tr}{Tr}

\DeclareMathOperator{\e}{\mathrm{e}}

\DeclareMathOperator{\Var}{Var}

\newcommand{\ket}[1]{| #1 \rangle}

\newcommand{\be}{{\mathbf e}}

\def\0{{\mathbf{0}}}
\def\1{{\mathbf{1}}}
\def\2{{\mathbf{2}}}
\def\3{{\mathbf{3}}}
\def\4{{\mathbf{4}}}
\def\5{{\mathbf{5}}}
\def\6{{\mathbf{6}}}

\def\7{{\mathbf{7}}}
\def\8{{\mathbf{8}}}
\def\9{{\mathbf{9}}}

\def\be{\begin{equation}}
\def\ee{\end{equation}}
\def\bea{\begin{eqnarray}}
\def\eea{\end{eqnarray}}

\def\eps{\varepsilon}

\theoremstyle{plain}
\newtheorem{theo}{Theorem} 
\newtheorem{prop}[theo]{Proposition} 
\newtheorem{lemm}[theo]{Lemma} 
\newtheorem{coro}[theo]{Corollary} 
\newtheorem*{prop2}{Proposition~\ref{prop:I2}}
\newtheorem*{prop3}{Proposition~\ref{prop:saddle}}
\newtheorem*{prop4}{Proposition~\ref{prop:Esp}}

\theoremstyle{definition}
\newtheorem{defn}[theo]{Definition} 

\theoremstyle{remark}
\newtheorem{remark}{Remark}[section]

\begin{document}
\let\origmaketitle\maketitle
\def\maketitle{
	\begingroup
	\def\uppercasenonmath##1{} 
	\let\MakeUppercase\relax 
	\origmaketitle
	\endgroup
}
	
\title{\bfseries \Large{Quantum Sphere-Packing Bounds with Polynomial Prefactors}}
	
\author{ {Hao-Chung Cheng$^{1,2,3}$, Min-Hsiu Hsieh$1$, and Marco  Tomamichel$^{1,4}$}}
\address{\small  	
$^{1}$Centre for Quantum Software and Information (UTS:Q$\ket{SI}$),\\
Faculty of Engineering and Information Technology, University of Technology Sydney, Australia\\
$^{2}$Graduate Institute Communication Engineering, National Taiwan University, Taiwan (R.O.C.)\\
$^{3}$Department of Applied Mathematics and Theoretical Physics, University of Cambridge, United Kingdom\\
$^{4}$School  of  Physics,  The University  of  Sydney, Australia}
\email{\href{mailto:HaoChung.Ch@gmail.com}{HaoChung.Ch@gmail.com}}
\email{\href{mailto:Min-Hsiu.Hsieh@uts.edu.au}{Min-Hsiu.Hsieh@uts.edu.au}}
\email{\href{mailto:marco.tomamichel@uts.edu.au}{marco.tomamichel@uts.edu.au}}
	
	
\begin{abstract}
We study lower bounds on the optimal error probability in classical coding over classical-quantum channels at rates below the capacity, commonly termed quantum sphere-packing bounds.
Winter and Dalai have derived such bounds for classical-quantum channels; however, the exponents in their bounds only coincide when the channel is classical.  In this paper, we show  that these two exponents admit a variational representation and are related by the Golden-Thompson inequality, reaffirming that Dalai's expression is stronger in general classical-quantum channels. Second, we establish a finite blocklength sphere-packing bound for classical-quantum channels, which significantly improves Dalai's prefactor from the order of subexponential to polynomial. Furthermore, the gap between the obtained error exponent for constant composition codes and the best known classical random coding exponent vanishes in the order of $o(\log n / n)$, indicating our sphere-packing bound is almost exact in the high rate regime. Finally, for a special class of symmetric classical-quantum channels, we can completely characterize its optimal error probability without the constant composition code assumption. The main technical contributions are two converse Hoeffding bounds for quantum hypothesis testing and the saddle-point properties of error exponent functions.  
\end{abstract}
	
	\maketitle
	
	\section{Introduction} \label{sec:intro}
	
	Shannon's noisy coding theorem \cite{Sha48} states that a message in an appropriately coded form can be {reliably} transmitted through a discrete memoryless channel $\mathscr{W}$, provided the coding rate $R$ is below the channel capacity $C_\mathscr{W}$. More precisely, the probability of decoding errors can be made arbitrarily small as the coding blocklength grows. Later, Shannon  himself pioneered the study of the exponential dependency of the optimal error probability $\epsilon^*(n,R)$ for a blocklength $n$ and transmission rate $R$ \cite{Sha59}. He defined the \emph{reliability function} to be, for any fixed coding rate $R< {C}_\mathscr{W}$,
	\begin{align}
	{E}(R) := \limsup_{n\to+\infty} \, -\frac1n \log \epsilon^*(n,R).
	\end{align}
	The quantity $E(R)$ then provides a measure of how rapidly the error probability approaches zero with an increase in blocklength. This characterization of the reliability function is hence called the \emph{reliability function analysis} or the \emph{error exponent analysis}.
	
For a classical channel, lower bounds for the reliability function can be established by random coding arguments \cite{Fei55,Fan61,Gal65,Gal68}. However, upper bounds require different techniques since the code-dependent bounds on the error probability need to be optimized over all codebooks. The first result---the \emph{sphere-packing bound} $E(R) \leq E_\text{sp}(R)$---was studied by Fano \cite{Fan61} and 
Shannon, Gallager, and Berlekamp \cite{SGB67}, and proved by the later paper. The \emph{sphere-packing exponent} $E_\text{sp}(R)$ is defined as
	\begin{align} \label{eq:sp}
	E_\text{sp}(R) :=  \sup_{s\geq 0}\left\{\max_{P} E_0(s,P) - sR\right\},
	\end{align}
where $P$ is maximized over all probability distributions on the input alphabet, and $E_0(s,P)$ is the \emph{auxiliary function} or \emph{Gallager's function} \cite{Gal65}.  Unlike Shannon-Gallager-Berlekamp's technique which relates channel coding to binary hypothesis testing, Haroutunian \cite{Har68,HHH07} employed a combinatorial method and obtained an upper bound for the reliability function in terms of the following expression
	\begin{align} \label{eq:sp2}
	\widetilde{E}_\text{sp}(R) :=
	\max_{P} \min_{\mathscr{V} } \left\{
	{D}\left(\mathscr{V}\|\mathscr{W}|P\right): I(P,\mathscr{V}) \leq R \right\},
	\end{align}
where $\mathscr{V}$ is minimized over all channels with the same output alphabet as $\mathscr{W}$, ${D}(\mathscr{V}\|\mathscr{W}|P)$ is the conditional relative entropy between the {dummy channel} $\mathscr{V}$ and the {true channel} $\mathscr{W}$, and $I(P,\mathscr{V})$ is the mutual information of the channel $\mathscr{V}$ (the detailed definitions are given in Section \ref{sec:notation}). It was then realized that the two quantities in Eqs.~\eqref{eq:sp} and \eqref{eq:sp2} are equivalent: they are related by convex program duality \cite{Har68, Bla74,BTC88,CK11}. Therefore, these two expressions, Eqs.~\eqref{eq:sp} or \eqref{eq:sp2}, are both called sphere-packing exponents.

Error exponent analysis in classical-quantum (c-q) channels is more challenging because of the noncommutative nature of quantum mechanics. Burnashev and Holevo \cite{BH98} introduced a quantum version of the auxiliary function \cite{Hol00,HM16} and initialized the study of reliability functions in c-q channels.	Winter \cite{Win99} derived a sphere-packing bound for c-q channels in the form of $\widetilde{E}_\text{sp}(R)$ in Eq.~\eqref{eq:sp2}, generalizing  Haroutunian's idea \cite{Har68}. Dalai \cite{Dal13} employed Shannon-Gallager-Berlekamp's approach \cite{SGB67} to establish a sphere-packing bound with Gallager's exponent in Eq.~\eqref{eq:sp}. In the follow-up work \cite{DW14b},  Dalai and Winter pointed out that these two exponents are not equal in c-q channels.
	In this work, we explicitly demonstrate a relationship between the two quantities. Precisely, we show that they individually admit a variational representation (Theorem \ref{theo:dual_sp} in Section \ref{sec:relation}):
	\begin{align}
	E_\textnormal{sp}(R) 
	&= \max_{P} \sup_{0<\alpha\leq 1} \min_{\sigma}  \left\{  \frac{1-\alpha}{\alpha} \left(  \sum_{x}P(x) D_{\alpha} \left(W_x\|\sigma\right) - R \right) \right\}; \label{eq:dual_sp001} \\
	\widetilde{E}_\textnormal{sp}(R) 
	&= \max_{P} \sup_{0<\alpha\leq 1} \min_{\sigma}  \left\{  \frac{1-\alpha}{\alpha} \left(  \sum_{x}P(x)  D_{\alpha}^\flat \left(W_x\|\sigma\right) - R \right) \right\}, \label{eq:dual_sp002}
	\end{align}
	where $\sigma$ is minimized over all density operators on some Hilbert space $\mathcal{H}$; 
	$W_x$ is the channel output state on $\mathcal{H}$;
	$D_\alpha$ is the (Petz) $\alpha$-R\'enyi divergence \cite{Pet86}; and $D_\alpha^\flat$ is the   \emph{log-Euclidean} $\alpha$-R\'enyi divergence. 
	
Since $D_\alpha \leq D_\alpha^\flat$ for all $\alpha\in(0,1]$, as a simple consequence of the Golden-Thompson inequality \cite{Gol65,Tho65}, the exponent $E_\text{sp}(R)$ in Eq.~\eqref{eq:dual_sp001} is stronger than $\widetilde{E}_\text{sp}(R)$ in~Eq.~\eqref{eq:dual_sp002}, i.e.~
	\begin{align}
	E(R) \leq E_\textnormal{sp} ( R ) \leq \widetilde{E}_\text{sp}(R).
	\end{align}
	These two exponents coincide for all $R$ only when all the channel output states commute\footnote{For the coding rates above channel capacity, these two exponents are both zero ($\alpha$ attains 1 in Eqs.~\eqref{eq:dual_sp001} and \eqref{eq:dual_sp002}). We exclude this trivial case and only consider the rate being strictly below capacity.} (i.e.~for classical channels). Thus, we call $E_\textnormal{sp} ( R)$ and $\widetilde{E}_\text{sp}(R)$ the \emph{strong sphere-packing exponent} and the \emph{weak sphere-packing exponent}, respectively. The lower bounds for the optimal error probability in terms of these two quantities are called the strong sphere-packing bound
	\begin{equation}
	\epsilon^*\left(n, R\right)  \geq f(n) \exp\left\{
	-n\left[ E_\textnormal{sp}(R - g(n))  \right]
	\right\}, \label{eq:ssp}
	\end{equation}
	and the weak sphere-packing bound 
	\begin{equation}
	\epsilon^*\left(n, R\right) \geq f(n) \exp\left\{
	-n\left[ \widetilde{E}_\textnormal{sp}(R - g(n))  \right]
	\right\}, \label{eq:wsp}
	\end{equation}
	where $f(n)$ is the prefactor of the bound, and $g(n)$ is a rate back-off term.
	We note that $g(n) = 0$ in our main result, and hence we only study $f(n)$ in the following discussion.
	
The strong sphere-packing bound obtained by Dalai \cite{Dal13} had a prefactor $f(n) = \mathrm{e}^{-O(\sqrt{n})} $, which is loose for small blocklength $n$ or in the situation where the transmission rate is close to channel capacity. 
Furthermore, such bound only holds asymptotically, i.e.~when $n$ tends to infinity.
The main contribution of this paper is to establish a sphere-packing bound with a better prefactor $f(n) = O(n^{-t})$ for some $t>1/2$, which not only holds for finite blocklength $n$ but also notably improves Dalai's bound \cite{Dal13} from the order of subexponential to polynomial  (Corollary \ref{coro:refined}). When restricting to constant composition codes, we can be more explicit about the obtained prefactor, namely, $f(n) = n^{-\frac12\left(1+\left|E_\text{sp}'(R)\right|+o(1)\right)}$ (Theorem \ref{theo:refined})\footnote{The notion $E_\text{sp}'(R)$ means the left-derivative of $E_\text{sp}(R)$. Note that $E_\text{sp}(R)$ is not necessarily differentiable.}. Moreover, this sphere-packing bound and the best known random coding upper bound \cite{AW12, SMF14, Sca14, Hon15} in the classical case coincide up to the third-order term (see the discussion in Section \ref{sec:main})).  Hence, our result yields a tight asymptotics of the sphere-packing bound for constant composition codes. Our second contribution is to show that, for a class of symmetric c-q channels, the prefactor $f(n) = O(n^{-\frac12\left(1+\left|E_\text{sp}'(R)\right|\right)})$, holds for general codes. In other words, 
we are able to obtain a tight sphere-packing bound for general codes, by exploiting a symmetric property of the channel.

Our main ingredients are a tight concentration inequality in strong large deviation theory \cite{BR60}, \cite[Theorem 3.7.4]{DZ98}, \cite[Section III.D]{AW14b} (Appendix \ref{app:tight}), an one-shot converse via quantum hypothesis testing \cite{Fan61, SGB67, Bla74}, and a uniform continuity property (Proposition~\ref{prop:UC} in Appendix~\ref{app:UC}). The strategy of the proof consists of three steps: (i) formulate the error probability of a certain codebook to a hypothesis testing problem; (ii) prove lower (or called the {converse}) bounds on type-I error in quantum hypothesis testing; and (iii) relate the error with the strong sphere-packing exponent. In Section \ref{ssec:converse_HT}, we first prove a one-shot converse bound to relate the channel coding problem into a binary quantum hypothesis (Proposition~\ref{lemm:hypothesis}).
We then employ Bahadur-Ranga Rao's inequality \cite{BR60} to establish converse Hoeffding bounds for quantum hypothesis testing (Theorem~\ref{theo:sharp_Hoeffding}).
Next, we apply a uniform continuity property (Proposition~\ref{prop:UC} in Appendix~\ref{app:UC}) and the established sharp Hoeffding bound (Theorem~\ref{theo:sharp_Hoeffding}) to prove two finite blocklength converse bounds on the optimal type-I error for a fixed-composition codeword.
The first bound is a Chebyshev-type bound with a subexponential prefactor (Proposition~\ref{prop:Chebyshev}), while the second bound is a sharp converse bound with a polynomial prefactor (Proposition~\ref{prop:sharp}).  Finally, we combine these two results to obtain a refined strong sphere-packing bound with a polynomial prefactor and finite blocklength.
	
Table \ref{table:comparison} collects major proof approaches of classical sphere-packing bounds, Eqs.~\eqref{eq:ssp} and \eqref{eq:wsp}, and discusses their generalizations to c-q channels. We remark that the established  finite blocklength bounds and the polynomial prefactor are crucial for the analysis of coding performance in the medium error probability regime (more commonly known as moderate deviation analysis) \cite{AW14b,CH17, CCT+16b}, classical data compression with quantum side information \cite{CHDH-2018, CHDH2-2018}, and joint source-channel coding with quantum side information \cite{CHDH3-2018}. 
	
The remaining part of the paper is organized as follows. Section~\ref{sec:notation} introduces the notation and necessary preliminaries. The relationship between the weak and strong sphere-packing exponents is proved in Section~\ref{sec:relation}. In Section~\ref{sec:main}, we prove a refined sphere-packing bound for c-q channels. We consider a symmetric c-q channel and establish an exact sphere-packing bound in Section~\ref{sec:symm}. Lastly, we conclude this paper in Section~\ref{sec:conclusion}.

\begin{table}[th!]
\centering
\resizebox{1\columnwidth}{!}{
\begin{tabular}{@{}>{\columncolor[gray]{0.90}}lcccccc@{}} %
		\toprule
				
& \multirow{2}{*}{Asymptotics} &  Composition & Pre-factor & Rate back-off & Classical-Quantum & \multirow{2}{*}{Expression} \\ 
\multirow{-2}{*}{Bounds\textbackslash Settings}
&   & dependent & $f(n)$ & $g(n)$ & channels &\\

\midrule
\midrule		

$\quad\,\,\,\,$Shannon-Gallager-  & \multirow{2}{*}{Asymptotic} & \multirow{2}{*}{Yes} & \multirow{2}{*}{$\mathrm{e}^{-O(\sqrt{n})}$} & \multirow{2}{*}{$O\left(\frac{\log n}{n}\right)$} & \multirow{2}{*}{Dalai \cite{Dal13}}& \multirow{2}{*}{$E_\text{sp}(R)$} \\
\multirow{-2}{*}{(a)} Berlekamp \cite{SGB67} & & & & & & \\
\midrule			

$\quad\,\,\,\,$Haroutunian \cite{Har68}  & \multirow{3}{*}{Asymptotic}  & \multirow{3}{*}{Yes} & \multirow{3}{*}{$\mathrm{e}^{-o(n)}$} & \multirow{3}{*}{$o(1)$} & \multirow{3}{*}{Winter \cite{Win99}} & \multirow{3}{*}{ $\widetilde{E}_\text{sp}(R)$ }\\
$\quad\,\,$ Omura \cite{Omu75} &  & & & & &\\
\multirow{-3}{*}{(b)} $  $Csis\'ar-Korner \cite{CK11} &  & & & & &\\
\hline
				
(c) Blahut \cite{Bla74}  & Asymptotic & Yes\protect\footnotemark & $\mathrm{e}^{-O(\sqrt{n})}$ & $O\left(n^{-\frac12}\right)$ & Eqs.~\eqref{eq:sketch3} \& \eqref{eq:caseI} & $E_\text{sp}(R)$\\
\hline		
				
& \multirow{2}{*}{Finite Blocklength} & \multirow{2}{*}{Yes} & \multirow{2}{*}{$n^{-\frac12\left(1+\left|E_\text{sp}'(R)\right|+o(1)\right)}  $} & \multirow{2}{*}{$0$ } & \multirow{2}{*}{Theorem \ref{theo:refined}} & \multirow{2}{*}{$E_\text{sp}(R)$} \\\multirow{-2}{*}{(d) Altu\u{g}-Wagner \cite{AW14}}&  & & & & &\\	
\hline		
								
(e) Elkayam-Feder \cite{EF16}  & Asymptotic & Yes & $O\left(n^{-t}\right)$ & $O\left(\frac{\log n}{n}\right)$ & Unknown & Unknown\\	
\hline				
				
$\quad\,\,\,\,$Agustin-Nakibo\u{g}lu  & \multirow{2}{*}{Finite Blocklength} & \multirow{2}{*}{No} & \multirow{2}{*}{$O\left(n^{-t}\right)$} & \multirow{2}{*}{$0$} & \multirow{2}{*}{Unknown}& \multirow{2}{*}{Unknown} \\
\multirow{-2}{*}{(f)}  $\,$\cite{Aug69,Aug78,Nak16a,Nak16b, Nak18a, Nak18b} & & & & & & \\			
\bottomrule				
\end{tabular}
}
\vspace{0.5em}
\caption[Caption without FN]{Different sphere-packing bounds are compared by (i) whether the bounds hold for finite blocklength $n$ or hold asymptotically as $n\to+\infty$; 	(ii) whether or not they are dependent on the constant composition codes;	(iii) \& (iv) 
the asymptotics $f(n)$ and $g(n)$; (v) the corresponding generalizations to classical-quantum channel coding; The parameter $t$ in rows (e) and (f) is some value in the range $t> 1/2$; and (vi) whether their error exponent expressions for c-q channels are expressed by $E_\text{sp}(R)$ given in Eq.~\eqref{eq:sp} or by $\widetilde{E}_\text{sp}(R)$ defined in Eq.~\eqref{eq:sp2}. 
}	\label{table:comparison}
		
\end{table}		

\footnotetext{Blahut in Ref.~\cite{Bla74} claimed that the sphere-packing bound can be derived without using constant composition argument. However, there is a non-trivial gap. We refer readers to the discussion by Nakibo\u{g}lu in Ref.~\cite[Appendix A]{Nak18b}}

	\section{Notation and Preliminaries} \label{sec:notation}
	
	Throughout this paper, we consider a finite-dimensional Hilbert space $\mathcal{H}$.  The set of density operators (i.e.~positive semi-definite operators with unit trace) and the set of full-rank density operators on $\mathcal{H}$ are defined as $\mathcal{S(H)}$ and $\mathcal{S}_{>0}(\mathcal{H})$, respectively. For $\rho,\sigma\in\mathcal{S(H)}$, we write $\rho\ll \sigma$ if the support of $\rho$ is contained in the support of $\sigma$. The identity operator on $\mathcal{H}$ is denoted by $\mathds{1}_\mathcal{H}$. If there is no possibility of confusion, we will skip the subscript $\mathcal{H}$.  We use $\Tr\left[\,\cdot\, \right]$ to denote the trace.  Let $\mathbb{N}$, $\mathbb{R}$, $\mathbb{R}_{\geq 0}$, and $\mathbb{R}_{> 0}$ denote the set of integers, real numbers, non-negative real numbers, and positive real numbers, respectively.
	Define $[n] := \{1,2,\ldots, n\}$ for $n\in\mathbb{N}$.
	
	For a positive semi-definite operator $A$ whose spectral decomposition is $A = \sum_{i} a_i P_i$, where $(a_i)_i$ and $(P_i)_i$ are the eigenvalues and eigenprojections of $A$, its power is defined as: $A^p := \sum_{i:a_i\neq 0} a_i^p P_i$.
	In particular, $A^0$ denotes the projection onto $\texttt{supp}(A)$, where we use $\texttt{supp}(A)$ to denote the support of the operator $A$.
	Further, $A\perp B$ means $\texttt{supp}(A) \cap \texttt{supp}(B) = \emptyset$, and $A\ll B$ indicates $\texttt{supp}(A) \subseteq \texttt{supp}(B)$.
	We denote by $\log$ the natural logarithm. 
	We use $f \vee g$ (resp.~$a\wedge b$) to denote the pointwise maximum (resp.~minimum) between two functions $f$ and $g$.

	\subsection{Information Quantities and Error-Exponent Functions} \label{ssec:Info}
	Given a pair of positive semi-definite operators $\rho,\sigma\in\mathcal{S(H)}$, we define {quantum relative entropy} \cite{Ume62,HP91} and {relative variance} \cite{TH13,Li14,TV15}, respectively as
	\begin{align} 
	{D}(\rho\|\sigma) &:=  \Tr \left[ \rho \left( {\log} \rho - {\log} \sigma \right) \right]; \label{eq:relative} \\
	{V}(\rho\|\sigma) &:=  \Tr \left[ \rho \left( {\log} \rho - {\log} \sigma \right)^2 \right] - {D}(\rho\|\sigma)^2, \label{eq:info_var}
	\end{align}
	when $\rho\ll\sigma$, and $+\infty$ otherwise.

	For two positive definite operators $\rho,\sigma > 0$ on $\mathcal{H}$, 
	and every $\alpha\in (0,1)$, we define the following two families of quantum R\'enyi divergences \cite{Pet86,ON00,MO14b}:
	\begin{align}
	&D_\alpha(\rho\|\sigma) := \frac{1}{\alpha-1} \log Q_\alpha(\rho\|\sigma) , \quad
	Q_\alpha(\rho\|\sigma) := \Tr \left[ \rho^\alpha \sigma^{1-\alpha} \right]; \label{eq:Petz}\\
	&D^\flat_\alpha(\rho\|\sigma) := \frac{1}{\alpha-1} \log Q^\flat_\alpha(\rho\|\sigma), \quad
	Q^\flat_\alpha(\rho\|\sigma) := \Tr \left[\mathrm{e}^{\alpha \log \rho + (1-\alpha) \log \sigma}\right]. \label{eq:chaotic}
	\end{align}
	We term the above quantities as the \emph{(Petz) $\alpha$-R\'enyi divergence}, and the \emph{log-Euclidean $\alpha$-R\'enyi divergence}, respectively.
	The log-Euclidean R\'enyi divergence arises from the \emph{log-Euclidean operator mean} (also called the \emph{chaotic mean}): $A\Diamond_\alpha B  := \exp\left( (1-\alpha)\log A + \alpha \log B \right)$ for $0\leq \alpha \leq 1$. 
	For general positive semi-definite operators $\rho,\sigma\geq 0$, the above definitions can be extended as
	\begin{align} \label{eq:delta_cont}
	Q_\alpha(\rho\|\sigma):= \lim_{\delta \downarrow 0} Q_\alpha(\rho+\delta\mathds{1}\|\sigma+\delta\mathds{1})
	\quad \text{and} \quad
	Q_\alpha^\flat(\rho\|\sigma):= \lim_{\delta \downarrow 0} Q_\alpha^\flat(\rho+\delta\mathds{1}\|\sigma+\delta\mathds{1}).
	\end{align}
	From the Golden-Thompson inequality \cite{Gol65,Tho65}:
	\begin{align}
	\Tr\left[ \e^{A+B} \right] \leq \Tr\left[ \e^A \e^B \right], \quad
	\forall A,B\geq 0,
	\end{align}
	these two quantities are related by 
\begin{align} \label{eq:compare}
Q^\flat_\alpha(\rho\|\sigma)  \leq Q_\alpha(\rho\|\sigma), \; \forall \alpha\in(0,1).
\end{align}  
For $\alpha=1$ and $\alpha=0$, we define (see e.g.~\cite[Lemma 3.5]{MO14b}):
\begin{align}
&D_1(\rho\|\sigma) := \lim_{\alpha\uparrow 1} D_\alpha (\rho\|\sigma) = {D}(\rho\|\sigma), \quad
D_1^\flat(\rho\|\sigma) := \lim_{\alpha\uparrow 1} D_\alpha^\flat (\rho\|\sigma) = {D}(\rho\|\sigma);\\
&D_0(\rho\|\sigma) := \lim_{\alpha\downarrow 0} D_\alpha (\rho\|\sigma), \quad D_0^\flat(\rho\|\sigma) := \lim_{\alpha\downarrow 0} D_\alpha^\flat (\rho\|\sigma).
	\end{align}
	
We will need the properties of the R\'enyi divergence for the next section.
\begin{lemm} \label{lemma:chaotic}
The following hold:
\begin{enumerate}[(a)]
	\item\label{Da_mono_cont}	For every $\rho,\sigma \in \mathcal{S(H)}$, the map $\alpha \mapsto D_\alpha\left( \rho \|\sigma\right)$  is continuous and monotone increasing on $[0,1]$.
	
	\item\label{Da_second_mono} Let $\rho \in \mathcal{S(H)}$, positive semi-definite operators $\sigma_1$ and  $\sigma_2$ on $\mathcal{H}$, and $\alpha\in[0,1]$.
	If $\sigma_1\leq \sigma_2$, then $ D_\alpha(\rho\|\sigma_1)\geq D_\alpha(\rho\|\sigma_2)$.
	Moreover, if $\sigma_1 = \gamma \sigma_2$ for some $\gamma>0$, then $D_\alpha(\rho\|\sigma_1) = D_\alpha(\rho\|\sigma_2) - \log \gamma$.
	
	\item\label{Da_second_convex} For every $\rho \in \mathcal{S(H)}$ and $\alpha\in[0,1]$, the map $\sigma \mapsto D_\alpha(\rho\|\sigma)$ is  convex and lower semi-continuous on $\mathcal{S(H)}$.
\end{enumerate}
\end{lemm}
We note that item~\ref{Da_mono_cont} was proved in \cite[Lemma~3.12, Corollary~3.15]{MO14b}; item~\ref{Da_second_mono} was proved in \cite[Lemma~3.24]{MO14b}; item~\ref{Da_second_convex} was shown in \cite{Ando79, Lie73, Pet86} \cite[Theorem 3.16]{MO14b}\footnote{It was shown in \cite[Corollary 3.27]{MO14b} that the map $\sigma \mapsto D_\alpha(\rho\|\sigma)$ is lower semi-continuous on $\mathcal{S(H)}$ for all $\alpha\in(0,1)$. The argument can be extended to the range $\alpha\in[0,1]$ by the same method in \cite[Lemma 3.26, Corollary 3.27]{MO14b}.}.

Let $\mathcal{X} = \{1,2,\ldots, |\mathcal{X}| \}$ be a finite alphabet, and  let $\mathscr{P}(\mathcal{X})$ be the set of probability distributions on $\mathcal{X}$. A classical-quantum (c-q) channel $\mathscr{W}$ maps elements of the finite set $\mathcal{X}$ to density operators in $\mathcal{S}(\mathcal{H})$, i.e.~$\mathscr{W}:x\mapsto W_x$. For a c-q channel $\mathscr{W}: \mathcal{X}\to \mathcal{S(H)}$ and $P\in\mathscr{P}(\mathcal{X})$, it is convenient to denote the corresponding c-q state:
\begin{align}
P\circ \mathscr{W} := \sum_{x\in\mathcal{X}} P(x) |x\rangle\langle x| \otimes W_x.
\end{align}
We also express the input distribution $P\in\mathscr{P}(\mathcal{X})$ as a diagonal matrix with respect to the computational basis $\{|x\rangle\}_{x\in\mathcal{X}}$, i.e.~$P = \sum_{x\in\mathcal{X}} P(x) |x\rangle\langle x|$.
Denote the {conditional relative entropy} of two c-q channels $\mathscr{V}, \mathscr{W}: \mathcal{X}\to \mathcal{S(H)}$ with a prior distribution $P\in\mathscr{P}(\mathcal{X})$ by
	\begin{align}
	{D}\left( \mathscr{V} \| \mathscr{W} | P \right) :=  \sum_{x\in\mathcal{X}} P(x)  {D}\left( V_x \| W_x \right).
	\end{align}
	Similarly, we define the following conditional entropic quantities for $\mathscr{W}: \mathcal{X}\to \mathcal{S(H)}$, $\sigma\in\mathcal{S(H)}$ and $P\in\mathscr{P}(\mathcal{X})$:
	\begin{align}
	{D}\left( \mathscr{W} \| \sigma | P \right) &:=  \sum_{x\in\mathcal{X}} P(x) {D}\left( W_x \| \sigma \right), \label{eq:D|P}\\
	{D}_\alpha \left( \mathscr{W} \| \sigma | P \right) &:=  \sum_{x\in\mathcal{X}} P(x) {D}_\alpha\left( W_x \| \sigma \right), \label{eq:Petz_P} \\
	{D}_\alpha^\flat \left( \mathscr{W} \| \sigma | P \right) &:=  \sum_{x\in\mathcal{X}} P(x) {D}_\alpha^\flat\left( W_x \| \sigma \right).
	\end{align}
	The \emph{mutual information} of the prior distribution $P\in\mathscr{P}(\mathcal{X})$ and the c-q channel $\mathscr{W}: \mathcal{X}\to \mathcal{S(H)}$ is defined as 
	\begin{align}
	I(P,\mathscr{W}) := \inf_{\sigma\in\mathcal{S(H)}} {D} \left( \mathscr{W} \| \sigma | P \right)
		= {D} \left( \mathscr{W} \| P\mathscr{W} | P \right), \label{eq:mutual2}
	\end{align}
	where $P\mathscr{W} := \sum_{x\in\mathcal{X}} P(x) W_x$ and the second equality can be found in Ref.~\cite{SW01}. The (classical) \emph{capacity} of the channel $\mathscr{W}: \mathcal{X}\to \mathcal{S(H)}$ is denoted by \cite{SW97, Hol98}:
	\begin{align} \label{eq:capacity}
	C_\mathscr{W} := \max_{P\in\mathscr{P}(\mathcal{X})} I(P,\mathscr{W}).
	\end{align}
	
We define two related information quantities: for every $\alpha\in[0,1]$,
\begin{align}
I_\alpha^{(1)}(P,\mathscr{W}) &:= \inf_{\sigma\in\mathcal{S(H)}} D_\alpha \left( P\circ \mathscr{W} \| P\otimes \sigma \right); \\
I_\alpha^{(2)}(P,\mathscr{W}) &:= \inf_{\sigma\in\mathcal{S(H)}} D_\alpha \left( \mathscr{W} \|  \sigma | P \right). \label{eq:I2}
\end{align}
The term $I_\alpha^{(1)}(P,\mathscr{W})$ is called \emph{the order $\alpha$ R\'enyi information} \cite{HT14, WWY14, Nak16a} or the \emph{generalized Holevo quantity}. The second term $I_\alpha^{(2)}(P,\mathscr{W})$ can be viewed as a variant of the $\alpha$-R\'enyi mutual information, and we call it the \emph{the order $\alpha$ Augustin information} \cite{Nak18a,Nak18b, CLH18}.	It can be verified that these two functions are related by Jensen's inequality:
\begin{align}
I_\alpha^{(1)}(P,\mathscr{W}) \leq I_\alpha^{(2)}(P,\mathscr{W}). \label{eq:Jen}
\end{align}
For the case of $\alpha = 1$, they both equal conventional mutual information, i.e. $I_1^{(1)}(P,\mathscr{W}) = I_1^{(2)}(P,\mathscr{W}) = I(P,\mathscr{W})$. Mosonyi and Ogawa \cite[Proposition 4.2]{MO14b} showed that for all $\alpha\in[0,1]$,
\begin{align} \label{eq:radius}
C_{\alpha,\mathscr{W}} := \sup_{ P \in \mathscr{P}(\mathcal{X})} I_\alpha^{(1)}(P,\mathscr{W})
= \sup_{ P \in \mathscr{P}(\mathcal{X})} I_\alpha^{(2)}(P,\mathscr{W}),
\end{align}
and it is termed the \emph{R\'enyi radius} or the \emph{R\'enyi capacity} of order $\alpha$. 
It is not hard to verify that $C_{\alpha,\mathscr{W}}$ equals the usual channel capacity $C_{\mathscr{W}}$ as $\alpha = 1$.
Moreover, Proposition \ref{prop:I2} below and the compactness of $\mathscr{P}(\mathcal{X})$ show that the suprema in Eq.~\eqref{eq:radius} can be replaced with maxima. The following proposition presents important properties of $\alpha$-Augustin mutual information and radius. The proof is given in Appendix \ref{app:I2}.
	
\begin{prop}[Properties of order $\alpha$ Augustin Information and Radius] \label{prop:I2}
Given any classical-quantum channel $\mathscr{W}:\mathcal{X}\to \mathcal{S(H)}$ with $|\mathcal{X}|< \infty$, the following hold:
\begin{enumerate}[(a)]
		\item\label{I2-mono} For every $P\in\mathscr{P}(\mathcal{X})$, $\alpha \mapsto I_\alpha^{(2)}(P,\mathscr{W})$ is monotone increasing on $[0,1]$, and $I_\alpha^{(2)}(P,\mathscr{W}) \leq \log |\mathcal{X}|$ for all $\alpha\in[0,1]$.
		
		\item\label{I2-Augustin_mean} For every $(\alpha,P)\in(0,1]\times \mathscr{P}(\mathcal{X})$, there exists a unique $\sigma_{\alpha,P} \in \mathcal{S(H)}$, termed Augustin mean, such that
		\begin{align}
		I_\alpha^{(2)}(P,\mathscr{W}) =  D_\alpha\left(\mathscr{W}\| \sigma_{\alpha,P} | P\right),
		\end{align}
		and
		\begin{align} \label{eq:fixed-point}
		\mathsf{T}_{\alpha,P}(\sigma) = \sigma \text{ and }
		\sigma \gg P\mathscr{W}
		\text{ if and only if } \sigma = \sigma_{\alpha,P},
		\end{align}	
		where the map $\mathsf{T}_{\alpha,P}:\mathcal{S}_{P,\mathscr{W}}(\mathcal{H}) \to \mathcal{S(H)}$ is defined as
		\begin{align}
		\mathsf{T}_{\alpha,P}(\sigma) = \sum_{x\in\mathcal{X}} P(x) \frac{\sigma^{\frac{1-\alpha}{2}} W_x^\alpha \sigma^{\frac{1-\alpha}{2}} }{\Tr\left[W_x^\alpha \sigma^{1-\alpha}\right]}.
		\end{align}		

		\item\label{I2-conc_P} For every $\alpha\in[0,1]$, the map $P\mapsto I_\alpha^{(2)}(P,\mathscr{W})$ is concave on $\mathscr{P}(\mathcal{X})$.
		
		\item\label{I2-conc_alpha} For every $P\in\mathscr{P}(\mathcal{X})$, $\alpha \mapsto \frac{1-\alpha}{\alpha} I_\alpha^{(2)}(P,\mathscr{W})$ is concave on $(0,1]$.
		
		\item\label{I2-cont_alpha} For every $P\in\mathscr{P}(\mathcal{X})$, $\alpha \mapsto I_\alpha^{(2)}(P,\mathscr{W})$ is continuous on $[0,1]$.
		
		\item\label{I2-cont_equi} The family of functions $\{ I_\alpha^{(2)}(P,\mathscr{W})  \}_{ \alpha\in[0,1]  }$ is uniformly equicontinuous in $P\in\mathscr{P}(\mathcal{X})$.
		Moreover,
		The map $(\alpha, P) \mapsto I_\alpha^{(2)}(P,\mathscr{W}) \text{ is jointly continuous on } [0,1]\times \mathscr{P}(\mathcal{X})$. 
		
		\item\label{I2-cont_mean} The map $(\alpha, P) \mapsto \sigma_{\alpha,P} \text{ is jointly continuous on } (0,1]\times \mathscr{P}(\mathcal{X})$. 
		
		\item\label{I2-cont_C}
		The map $\alpha \mapsto C_{\alpha,\mathscr{W}}$ is continuous and monotone increasing on $[0,1]$. 
\end{enumerate}
\end{prop}

The \emph{strong sphere-packing exponent} \cite{Dal13} of a c-q channel $\mathscr{W} : \mathcal{X}\to \mathcal{S(H)}$ and a rate $R\geq 0$ is defined by
\begin{align} \label{eq:sp_1}
E_\text{sp}(R) :=  \max_{P\in\mathscr{P}(\mathcal{X})} E_\text{sp}(R,P),
\end{align}
where
\begin{align} \label{eq:sp_1P}
E_\text{sp}(R,P) := \sup_{s\geq 0}\left\{ E_0(s,P) - sR\right\},
\end{align}
and $E_0$ is the \emph{auxiliary function} of the c-q channel $\mathscr{W}$ (see \cite{BH98, Hol00, HM16}):
\begin{align} \label{eq:E0}
E_0(s,P) :=  -\log \Tr \left[
\left( \sum_{x\in\mathcal{X}} P(x) \cdot W_x^{1/(1+s)}\right)^{1+s}
\right]
\end{align}
for all $P\in\mathscr{P}(\mathcal{X})$ and $s\geq0$.

The \emph{weak sphere-packing exponent} \cite{Win99} is defined as
\begin{align} \label{eq:sp_2}
\widetilde{E}_\text{sp}(R) :=
\max_{P\in\mathscr{P}(\mathcal{X})} \widetilde{E}_\text{sp}(R,P),
\end{align}
where
\begin{align} \label{eq:sp_2P}
\widetilde{E}_\text{sp}(R,P) :=
\min_{\mathscr{V}: \mathcal{X}\to \mathcal{S(H)}} \left\{
{D}\left(\mathscr{V}\|\mathscr{W}|P\right): I(P,\mathscr{V}) \leq R \right\}.
\end{align}
	
We also need the following definitions: for any $R\geq 0$ and $P\in\mathscr{P}(\mathcal{X})$,
\begin{align}
E_\text{sp}^{(1)} (R,P) &:=  \sup_{0<\alpha\leq 1}  \frac{1-\alpha}{\alpha}   \left(  I_\alpha^{(1)}(P,\mathscr{W}) -R \right); \label{eq:Esp1}\\
E_\text{sp}^{(2)} (R,P) &:= \sup_{0<\alpha\leq 1}   \frac{1-\alpha}{\alpha}   \left( I_\alpha^{(2)}(P,\mathscr{W}) -R \right), \label{eq:Esp2} 
\end{align}
	
Eq.~\eqref{eq:Jen} implies that (see also Theorem \ref{theo:dual_sp}) $E_\text{sp}^{(1)} (R,P)\leq E_\text{sp}^{(2)} (R,P)$.	By quantum Sibson's identity \cite{SW12}\footnote{
For joint state $\rho_{AB} \in\mathcal{S}(AB)$ with marginal states denoted by $\rho_A$ and $\rho_B$, quantum Sibson's identity \cite{SW12} states that the minimizer of $\min_{\sigma\in\mathcal{S(H)}} D_\alpha(\rho_{AB}\|\rho_A\otimes \sigma_B)$ is
	$ \left( \Tr_A\left[ \rho_{AB}^\alpha \right] \right)^{ 1/\alpha} / \Tr\left[  \left( \Tr_A\left[ \rho_{AB}^\alpha \right] \right)^{1/\alpha} \right] $.
For the case of classical-quantum channels, the minimizer of $\min_{\sigma\in\mathcal{H}} D_\alpha(P\circ\mathscr{W}\|P\otimes \sigma) $ is then $ \left( \sum_{x\in\mathcal{X}} P(x) W_x^{\alpha}  \right)^{1/\alpha} / \Tr\left[ \left( \sum_{x\in\mathcal{X}} P(x) W_x^{\alpha} \right)^{1/\alpha} \right] $.
}, 	one finds
\begin{align}
E_\text{sp}^{(1)} (R,P) = E_\text{sp} (R,P). \label{eq:Esp_Sibson}
\end{align}
Proposition \ref{prop:I2} and Eq.~\eqref{eq:radius} imply that the two quantities given in Eqs.~\eqref{eq:Esp1} and \eqref{eq:Esp2} are equal to the strong sphere-packing exponent by maximizing over the input distributions:
\begin{align}
E_\text{sp}(R) =  \max_{ P \in \mathscr{P}(\mathcal{X})} E_\text{sp}^{(1)} (R,P) = \max_{ P \in \mathscr{P}(\mathcal{X})} E_\text{sp}^{(2)} (R,P). \label{eq:Esp_max}
\end{align}
In Proposition \ref{prop:Esp} below, one has $E_\text{sp}(R) = +\infty$ for $R<C_{0,\mathscr{W}}$, and $E_\text{sp}(R) = 0$ as $R>C_\mathscr{W}$. Throughout this paper, we further assume that the considered c-q channel $\mathscr{W}$ satisfies $ C_{0,\mathscr{W}}  < C_\mathscr{W}.$ 
We note that $C_{0,\mathscr{W}}$ gives an upper bound to the \emph{zero-error capacity} of the channel. For details, we refer the readers to Ref.~\cite{Dal13}.
As we will show in Section \ref{sec:main}, the quantity $E_\text{sp}^{(2)} (R,P)$ plays a significant role in the connection between hypothesis testing and channel coding. Moreover, Proposition \ref{prop:saddle} below shows that the  the optimizer in Eqs.~\eqref{eq:I2}  and \eqref{eq:Esp2} forms a saddle-point. The proof closely follows Altu{\u{g}} and  Wagner \cite[Proposition 1]{AW14}, and is given in Appendix \ref{app:saddle}.

\begin{prop}[Saddle-Point] \label{prop:saddle}
Consider a classical-quantum channel $\mathscr{W} : \mathcal{X}\to \mathcal{S(H)}$, any $R\in(C_{0,\mathscr{W}}, C_\mathscr{W}   )$, and $P\in \mathscr{P}(\mathcal{X})$. Let\footnote{For $\alpha\in(0,1]$, $\mathcal{S}_{P,\mathscr{W}}(\mathcal{H})$ is the \emph{effective domain} of $F_{R,P}(\alpha,\cdot)$ \cite{Roc70}. Namely, $F_{R,P}(\alpha, \sigma)$ is finite for all  $\sigma \in \mathcal{S}_{P,\mathscr{W}}(\mathcal{H})$. Proposition~\ref{prop:saddle} then aims to find the saddle-points within its domain $(0,1]\times \mathcal{S}_{P,\mathscr{W}}(\mathcal{H})$.
Otherwise, the saddle-point makes no sense.}
\begin{align} 
\mathcal{S}_{P,\mathscr{W}}(\mathcal{H}) &:= \left\{ \sigma \in \mathcal{S(H)}: \forall x \in \textnormal{\texttt{supp}}(P), \, W_x \not\perp  \sigma   \right\}.
\end{align}
Define
\begin{align} \label{eq:F}
F_{R,P} (\alpha, \sigma) := 
\begin{dcases}
\frac{1-\alpha}{\alpha} \left(   D_\alpha\left( \mathscr{W} \| \sigma | P  \right) -R \right), & \alpha \in (0,1) \\
0, & \alpha = 1
\end{dcases},
\end{align}
on $(0,1]\times \mathcal{S}(\mathcal{H})$, and denote by
\begin{align} \label{eq:PR}
\mathscr{P}_R(\mathcal{X}) := \left\{ P\in\mathscr{P}(X) : \sup_{0<\alpha\leq 1} \inf_{\sigma \in \mathcal{S(H)}}  F_{R,P} (\alpha, \sigma) \in \mathbb{R}_{>0}    \right\}.
\end{align}
The following holds
\begin{enumerate}[(a)]
\item\label{saddle-a} For any $P\in\mathscr{P}(\mathcal{X})$, $F_{R,P}(\cdot,\cdot)$ has a saddle-point on $(0,1]\times \mathcal{S}_{P,\mathscr{W}}(\mathcal{H})$ with  the saddle-value:
\begin{align}
\min_{\sigma \in \mathcal{S(H)}}\sup_{0<\alpha\leq 1}   F_{R,P} (\alpha,\sigma) = \sup_{0<\alpha\leq 1} \min_{\sigma \in \mathcal{S(H)}}  F_{R,P} (\alpha,\sigma)= E_\textnormal{sp}^{(2)}(R,P).
\end{align}
			
\item\label{saddle-b} Fix $P\in\mathscr{P}_R(\mathcal{X})$. Any saddle-point $(\alpha_{R,P}^\star, \sigma_{R,P}^\star)$ of $F_{R,P}(\cdot,\cdot)$ satisfies $\alpha_{R,P}^\star\in(0,1)$ and 
\begin{align}
\sigma_{R,P}^\star 
\gg W_x, \quad \forall x \in \textnormal{\texttt{supp}}(P).
\end{align}
			
\item\label{saddle-c} For $P\in\mathscr{P}_R(\mathcal{X})$, the saddle-point is unique.

\item\label{saddle-d}  For any $\underline{R} \in (C_{0,\mathscr{W}}, R]$, both $\alpha_{r,P}^\star $ and $\sigma_{r,P}^\star$ are jointly continuous functions of $(r,P)$ on $[\underline{R},R]\times \mathscr{P}(\mathcal{X})$.
			
\end{enumerate}
\end{prop}
	
The following proposition discusses the continuity and differentiability of the error-exponent functions. The proof is shown in Appendix \ref{app:Esp}. 
	
\begin{prop}[Properties of Error-Exponent Functions] \label{prop:Esp}
Consider a classical-quantum channel $\mathscr{W} : \mathcal{X} \to \mathcal{S(H)}$ with $C_{0,\mathscr{W}} < C_\mathscr{W}$.  We have
\begin{enumerate}[(a)]
			
\item\label{Esp-a}  
Given every $P\in\mathscr{P}(\mathcal{X})$, $E_\textnormal{sp}^{(2)}(\cdot,P)$ is convex and non-increasing on $[0,+\infty]$, and continuous  on $\left[ I_0^{(2)}(P,\mathscr{W}) ,  +\infty \right]$. For every $R>C_{0,\mathscr{W}}$, $E_\textnormal{sp}^{(2)}(R,\cdot)$ is continuous on $\mathscr{P}(\mathcal{X})$. Further,
\begin{align}
&E_\textnormal{sp}^{(2)} (R,P) =
\begin{dcases}
+\infty, & R< I_0^{(2)}(P,\mathscr{W}) \\
0, & R\geq I_1^{(2)}(P,\mathscr{W})	
\end{dcases}.
\end{align}
			
\item\label{Esp-b}   $E_\textnormal{sp}(\cdot)$ is convex and non-increasing  on $ [0,+\infty]$,  and continuous on $ [ C_{0,\mathscr{W}} ,  +\infty ]$. Further,
\begin{align}
&E_\textnormal{sp} (R) =
\begin{dcases}
+\infty, & R< C_{0,\mathscr{W}} \\
0, & R\geq C_{1,\mathscr{W}}
\end{dcases}.
\end{align}
			
\item\label{Esp-c} Consider any $R\in(C_{0,\mathscr{W}}, C_\mathscr{W})$ and $P\in\mathscr{P}_R(\mathcal{X})$ (see Eq.~\eqref{eq:PR}). The function $E_\textnormal{sp}^{(2)}(\cdot, P)$ is differentiable with
\begin{align} \label{eq:s2}
s_{R,P}^{\star} := \frac{1-\alpha_{R,P}^\star}{\alpha_{R,P}^\star} = - \left.\frac{\partial E_\textnormal{sp}^{(2)}(r,P)}{\partial r}\right|_{r=R} \in \mathbb{R}_{>0},
\end{align}
where $ \alpha_{R,P}^{\star} $ is the optimizer in Eq.~\eqref{eq:Esp2}.
Moreover, 
\begin{align}
I_{\alpha_{R,P}^\star} (P,\mathcal{W}) > R.
\end{align}
			
\end{enumerate}
\end{prop}

Given any $R\in(C_{0,\mathscr{W}}, C_\mathscr{W})$ and $P\in\mathscr{P}_R(\mathcal{X})$, we define $E_\text{sp}'(R)$ as the \emph{left derivative} of $E_\text{sp}(R)$:
\begin{align}
E_\text{sp}'(R)  := \lim_{\delta \downarrow 0} \frac{ E_\text{sp}(R-\delta) - E_\text{sp}(R) }{ -\delta}.
\end{align}
Since $E_\text{sp}(R)$ is continuous in $R$ by Proposition~\ref{prop:Esp}-\ref{Esp-b}, the above definition is well-defined.
Moreover, Eq.~\eqref{eq:Esp_max} and  Proposition~\ref{prop:Esp}-\ref{Esp-c} imply that the absolute value of $E_\text{sp}'(R)$ means the maximum absolute value of the slope of the sphere-packing exponent at $R$ by
\begin{align} \label{eq:s_star}
\left| E_\text{sp}'(R) \right| = 
\max_{P: E_\text{sp}^{(2)}(R,P) = E_\text{sp}(R)   } s^\star_{R,P}.
\end{align}
Note that the term $\left| E_\text{sp}'(R) \right|$ in Eq.~\eqref{eq:s_star} is well-defined and finite by item \ref{saddle-d} in Proposition \ref{prop:saddle}.

\subsection{Quantum Hypothesis Testing and Channel Coding} \label{ssec:binay}
	Consider a binary hypothesis whose null and alternative hypotheses are $\rho\in\mathcal{S(H)}$ and $\sigma\in\mathcal{S(H)}$, respectively. The \emph{type-I error} and \emph{type-II error} of the hypothesis testing, for an operator $0\leq Q\leq \mathds{1}$, are defined as:
	\begin{align}
	\alpha\left(Q;\rho\right)&:= \Tr\left[ (\mathds{1}-Q) \rho \right],\\
	\beta\left(Q;\sigma\right)&:= \Tr\left[ Q \sigma \right].
	\end{align}
	There is a trade-off relation between these two errors. Thus we can define the minimum Type-I error when the type-II error is below $\mu\in(0,1)$ as
	\begin{align} \label{eq:alpha}
	\widehat{\alpha}_{\mu}\left(\rho\|\sigma\right)
	:= \min_{0\leq Q\leq \mathds{1} } \big\{ \alpha\left(Q;\rho\right) : \beta\left(Q;\sigma\right) \leq \mu  \big\}.
	\end{align}
	
	We define an error-exponent function \cite{Hay06, NH07, ANS+08} for two sequences of states
	\begin{align}
	&\mathsf{H}_0: \rho^n = \rho_1\otimes \rho_2\otimes \cdots \otimes \rho_n,\\
	&\mathsf{H}_1: \sigma^n = \sigma_1\otimes \sigma_2\otimes \cdots \otimes \sigma_n,
	\end{align}
by 
\begin{align} \label{eq:exponent_HT}
\phi_n\left( r| \rho^n \| \sigma^n  \right)
:= \sup_{\alpha\in(0,1]} \left\{ \frac{1-\alpha}{\alpha} \left( \frac1n D_\alpha\left(\rho^n\|\sigma^n\right) - r \right)  \right\},\quad r\geq 0.
\end{align}
It is known that \cite[Lemma 4]{ANS+08}
\begin{align} \label{eq:phi}
\phi_n\left( r| \rho^n \| \sigma^n  \right) = +\infty, \quad \forall r\in \left[0, - \frac1n D_0\left(\rho^n\|\sigma^n\right) \right).
\end{align}
		
Let $\mathcal{M}$ be a finite alphabetical set with size $M=|\mathcal{M}|$. An ($n$-block) \emph{encoder} is a map $f_n:\mathcal{M}\to \mathcal{X}^n$ that encodes each message $m\in\mathcal{M}$ to a codeword $\mathbf{x}^n(m) :=   x_1(m) x_2(m) \ldots x_n(m) \in\mathcal{X}^n$. The codeword $\mathbf{x}^n(m)$ is then mapped to a state
\begin{align}
W_{\mathbf{x}^n(m)}^{\otimes n} = W_{x_1(m)} \otimes W_{x_2(m)} \otimes \cdots \otimes W_{x_n(m)} \in \mathcal{S}(\mathcal{H}^{\otimes n}).
\end{align}
The \emph{decoder} is described by a positive operator-valued measurement (POVM) $\Pi_n = \{\Pi_{n,1},\ldots, \Pi_{n,M} \}$ on $\mathcal{H}^{\otimes n}$, where $\Pi_{n,i} \geq 0$ and $\sum_{i=1}^{M} \Pi_{n,i} = \mathds{1}$. The pair $(f_n, \Pi_n) =: \mathcal{C}_n$ is called an $(n,R)$-\emph{code} with \emph{rate} $R = \frac1n \log |\mathcal{C}_n| = \frac1n \log M$.  The error probability of  sending a message $m$ with the code $ \mathcal{C}_n$ is $\epsilon_m(\mathcal{C}_n) :=  1- \Tr\left(\Pi_{n,m} W_{\mathbf{x}^n(m)}^{\otimes n}\right)$. We use $\epsilon_\text{max}(\mathcal{C}_n) = \max_{m\in\mathcal{M}} \epsilon_m(\mathcal{C}_n) $ and $\bar{\epsilon}(\mathcal{C}_n) = \frac1M \sum_{m\in\mathcal{M}} \epsilon_m(\mathcal{C}_n)$ to denote the \emph{maximal} error probability and the \emph{average} error probability, respectively.  	
Denote by $\epsilon^*\left(n, R \right)$ the {smallest} average probability of error among all the coding strategies with a blocklength $n$ and coding rate $R$, i.e.
\begin{align}
\epsilon^*(n,R) := \inf \{ \bar{\epsilon}(\mathcal{C}_n) : \mathcal{C}_n \text{ is an }(n,R)\text{-code}\}.
\end{align}
The reliability function of the channel $\mathscr{W}$ and the coding rate $R$ is defined by\footnote{Throughout this paper, we skip the dependence of the channel $\mathscr{W}$ in the reliability function and error-exponent functions.}
\begin{align} \label{eq:reliability}
E(R) := \limsup_{n\to +\infty} -\frac1n \log \epsilon^*\left(n,R\right).
\end{align}
Winter \cite{Win99} and Dalai \cite{Dal13} showed that the reliability function of a c-q channel can be upper bounded by $E(R)\leq \widetilde{E}_\text{sp}(R)$ and $E(R)\leq E_\text{sp}(R)$, respectively.
Given a sequence ${\mathbf{x}}^n \in \mathcal{X}^n$, we denote by 
\begin{align} \label{eq:empirical}
P_{\mathbf{x}^n} (x) := \frac1n \sum_{i=1}^n \mathbf{1}\left\{ x = x_i \right\}, \quad \forall x\in\mathcal{X}
\end{align}
the empirical distribution of $\mathbf{x}^n$, where $x_i$ is the $i$-th position of $\mathbf{x}^n$; and $\mathbf{1}\{A\} = 1$ if the event $A$ is true; otherwise $\mathbf{1}\{A\} = 0$.
A constant composition code with a composition $P_{\mathbf{x}^n}$ refers to a codebook whose codewords all have the same distribution $P_{\mathbf{x}^n}$.

\subsection{Nussbaum-Szko{\l}a Distributions} \label{ssec:NS}
Assume the dimension of the Hilbert space $\mathcal{H}$ is $d$. Given density operators $\rho,\sigma \in\mathcal{S(H)}$ with spectral decompositions
\begin{align}
\rho = \sum_{i\in[d]} \lambda_i |e_i\rangle\langle e_i |, \quad \text{and} \quad
\sigma = \sum_{j\in[d]} \mu_j |f_j\rangle\langle f_j |,
\end{align}
we define the \emph{Nussbaum-Szko{\l}a distributions} \cite{NS09} $p^{\rho,\sigma}, q^{\rho,\sigma}$ as
\begin{align} \label{eq:NS_idx}
p^{\rho,\sigma}(i,j) := \lambda_i |\langle e_i | f_j \rangle |^2, \quad
q^{\rho,\sigma}(i,j) := \mu_j |\langle e_i | f_j \rangle |^2.
\end{align}
The distributions $p^{\rho,\sigma},q^{\rho,\sigma}$ have the same mathematical properties as the density operators $\rho,\sigma$ in some cases, and thus are useful in the sequel. First, one can verify that \cite{NS09,TH13},
\begin{align} \label{eq:NS1}
D_\alpha\left(\rho\|\sigma \right) = D_\alpha \left( p^{\rho,\sigma} \| q^{\rho,\sigma} \right), \quad \forall \alpha\in[0,1].
\end{align}
Second, for product states $\rho_1\otimes \rho_2$ and $\sigma_1\otimes\sigma_2$, we have
\begin{align} \label{eq:NS2}
p^{\rho_1\otimes \rho_2, \sigma_1\otimes \sigma_2} = p^{\rho_1,\sigma_1} \otimes p^{\rho_2,\sigma_2}, \quad \text{and} \quad
q^{\rho_1\otimes \rho_2, \sigma_1\otimes \sigma_2} = q^{\rho_1,\sigma_1} \otimes q^{\rho_2,\sigma_2}.
\end{align}
Third, $\rho\ll\sigma$ if and only if $p^{\rho,\sigma} \ll q^{\rho,\sigma}$. Moreover, we will use $\omega$ to represent the pair of indices $(i,j)$ in Eq.~\eqref{eq:NS_idx}, and view the distributions $p^{\rho,\sigma},q^{\rho,\sigma}$ as diagonal matrices, e.g.~$\Tr\left[ p^{\rho,\sigma} \right] = \sum_{\omega\in [d]\times [d]} p^{\rho,\sigma}(\omega) $.

\section{Relation between the Strong and Weak Sphere-Packing Exponents} \label{sec:relation}
	
This section derives alternative formulations of the strong and weak sphere-packing exponents of Eqs.~\eqref{eq:sp}-\eqref{eq:sp2}, and provides a relation between these two exponents. As we will show later, the derived formulations are essentially optimization problems in the primal domain, while the expressions in Eqs.~\eqref{eq:sp} and \eqref{eq:sp2} are corresponding dual representations.
	
We first consider the following convex optimization problem and then exploit it to establish variational formulations of the sphere-packing exponents. Let $\rho, \tau \in \mathcal{S(H)}$ be two density operators. Consider the following convex optimization problem:
\begin{align} \label{eq:Primal}
\begin{split}
\mathrm{(P)} \quad & e(r) := \inf_{\sigma \in\mathcal{S}(\mathcal{H})} {D}\left( \sigma\|\rho\right),\\
&\text{subject to} \quad
{D}\left( \sigma\|\tau\right) \leq r.
\end{split}
\end{align}
The above primal problem is interpreted as finding the optimal operator $\sigma^\star$ that achieves the minimum  relative entropy $e(r)$ to $\rho$, within $r$-radius to $\tau$. The following result shows the dual representation of problem $\mathrm{(P)}$ via Lagrangian duality.
	
	\begin{lemm}[{\cite[Section 3.7]{Hay06}, \cite{Nag06}, \cite[Theorem~3.6]{MO14b}}] \label{lemm:dual}
		The dual problem of $\mathrm{(P)}$ is given by
		\begin{align}
		\mathrm{(D)} \quad \sup_{ s\geq 0 } \left\{ - (1+s) \log Q^\flat_{\frac1{1+s}}(\rho\|\tau)  -sr 
		\right\}.
		\end{align}
	\end{lemm}
	\begin{proof}
By the method of Lagrange multipliers, the primal problem in Eq.~\eqref{eq:Primal} can be rewritten as
		\begin{align}
		&\sup_{s\geq 0} \inf_{\sigma \in \mathcal{S(H)} } \left\{  
		{D}(\sigma \| \rho) + s \left({D}(\sigma\|\tau) - r \right)
		\right\} \\
		&= \sup_{s\geq 0}  \left\{   (1+s)\inf_{\sigma \in \mathcal{S(H)} } \left\{
		\frac{1}{1+s}{D}(\sigma \| \rho) + \frac{s}{1+s} {D}(\sigma\|\tau) \right\} - sr
		\right\} \\
&= \sup_{ s\geq 0 } \left\{ - (1+s) \log Q^\flat_{\frac1{1+s}}(\rho\|\tau)  -sr 
\right\}, \label{eq:Lagrange}
\end{align}
where the last equality follows from \cite[Theorem~3.6]{MO14b}.
\end{proof}

\begin{theo}[Variational Representations of the Sphere-Packing Exponents] \label{theo:dual_sp}
Let $\mathscr{W}:\mathcal{X}\to \mathcal{S(H)}$ be a classical-quantum channel. For any $R> C_{0,\mathscr{W}}$, we have
\begin{align}
\widetilde{E}_\textnormal{sp}(R,P) 
&=  \sup_{0<\alpha\leq 1} \min_{\sigma\in\mathcal{S(H)}}  \left\{  \frac{1-\alpha}{\alpha} \left(    D_{\alpha}^\flat \left(\mathscr{W}\|\sigma|P\right)  - R \right) \right\}, \quad\text{and}  \label{eq:dual_sp02} \\	
E_\textnormal{sp}(R,P) 
&\leq \sup_{0<\alpha\leq 1} \min_{\sigma\in\mathcal{S(H)}}  \left\{  \frac{1-\alpha}{\alpha} \left(  D_{\alpha} \left(\mathscr{W}\|\sigma|P\right) -R \right) \right\}, \label{eq:dual_sp01} 
\end{align}
where $\widetilde{E}_\textnormal{sp}(R,P)$ and $E_\textnormal{sp}(R,P)$ are defined in Eqs.~\eqref{eq:sp_2P} and \eqref{eq:sp_1P}, respectively.
		
Moreover, equality in Eq.~\eqref{eq:dual_sp01} is attained when maximizing over all prior distributions, i.e.,
		\begin{align}
		E_\textnormal{sp}(R) = \max_{P\in\mathscr{P}(\mathcal{X})} E_\textnormal{sp}(R,P)
		&= \max_{P\in\mathscr{P}(\mathcal{X})} \sup_{0<\alpha\leq 1} \min_{\sigma\in\mathcal{S(H)}}  \left\{  \frac{1-\alpha}{\alpha} \left(  D_{\alpha} \left(\mathscr{W}\|\sigma|P\right) - R \right) \right\}. \label{eq:dual_sp03} 
		\end{align}
	\end{theo}
	\begin{proof}
We start with the proof of Eq.~\eqref{eq:dual_sp02}. Recall Eq.~\eqref{eq:mutual2}:
\begin{align}
 I(P,\mathscr{V})
 = \min_{\sigma\in\mathcal{S(H)}} {D}\left( \mathscr{V}\|\sigma|P\right).
\end{align}
We find
\begin{align}
\widetilde{E}_\text{sp}(R,P)&=  \min_{\mathscr{V}: \mathcal{X}\to \mathcal{S(H)}} \left\{ {D}\left(\mathscr{V}\|\mathscr{W}|P\right): I(P,\mathscr{V}) \leq R \right\} \\
&=  \min_{\mathscr{V}: \mathcal{X}\to \mathcal{S(H)}}  \left\{ {D}\left(\mathscr{V}\|\mathscr{W}|P\right): \min_{\sigma\in\mathcal{S(H)}} {D}\left( \mathscr{V}\|\sigma|P\right) \leq R \right\} \\
&=  \sup_{s\geq 0} \min_{\mathscr{V}: \mathcal{X}\to \mathcal{S(H)}}  \left\{{D}\left(\mathscr{V}\|\mathscr{W}|P\right) + s\left( \min_{\sigma\in\mathcal{S(H)}} {D}\left( \mathscr{V}\|\sigma|P\right) - R \right)\right\} \label{eq:dual_sp1}\\
&= \sup_{s\geq 0} \min_{\sigma\in\mathcal{S(H)}} \min_{\mathscr{V}: \mathcal{X}\to \mathcal{S(H)}}  \left\{ -sR +\sum_{x\in\mathcal{X}} P(x) \left[ {D}\left(V_x \|W_x \right) + s \cdot {D}\left( V_x \| \sigma\right) \right] \right\} \\
&= \sup_{s\geq 0} \min_{\sigma\in\mathcal{S(H)}}  \left\{ \sum_{x\in\mathcal{X}} P(x) \min_{V_x \in \mathcal{S(H)} }  \left[ {D}\left(V_x \|W_x \right) + s \cdot {D}\left( V_x \| \sigma\right) - sR \right] \right\} \label{eq:dual_sp2} \\
&= \sup_{s\geq 0}\min_{\sigma\in\mathcal{S(H)}}  \left\{ \sum_{x\in\mathcal{X}} P(x) 
\left[ - (1+s) \log Q_{\frac{1}{1+s}}^\flat (W_x\|\sigma) - s R
\right]
\right\}. \label{eq:dual_sp00}
\end{align}
In the third equality we introduced the constraint into the objective function via the Lagrange multiplier $s\geq 0$; The fifth equality follows from the linearity of the convex combination; and the last line  is due to Lemma~\ref{lemm:dual}.
Further, we use the substitution $\alpha = 1/(1+s)$ and recall the definition of the log-Euclidean $\alpha$-R\'enyi divergence given in Eq.~\eqref{eq:chaotic} 
to obtain:
\begin{align}
\widetilde{E}_\text{sp}(R,P)
&=  \sup_{0<\alpha\leq 1} \min_{\sigma\in\mathcal{S(H)}}  \left\{  \frac{1-\alpha}{\alpha} \left(    D_{\alpha}^\flat \left(\mathscr{W}\|\sigma|P\right) -R \right) \right\}. \label{eq:dual_sp5}
\end{align}
Hence, we prove the first claim in Eq.~\eqref{eq:dual_sp02}.
		
Next, we will prove Eq.~\eqref{eq:dual_sp01}.  From Jensen's inequality and the concavity of the logarithm, the right-hand side of Eq.~\eqref{eq:dual_sp01} implies that
\begin{align}
&\sup_{0<\alpha\leq 1} \min_{\sigma\in\mathcal{S(H)}}  \left\{  \frac{1-\alpha}{\alpha} \left(  \sum_{x\in\mathcal{X}}P(x) D_{\alpha} \left(W_x\|\sigma\right) -R \right) \right\}  \label{eq:dual_sp10}\\
&= \sup_{0< \alpha\leq 1} \min_{\sigma\in\mathcal{S(H)}} \left\{  -\frac1{\alpha} \sum_{x\in\mathcal{X}}P(x) \log\Tr\left[ W_x^\alpha \sigma^{1-\alpha} \right] -\frac{1-\alpha}{\alpha}R \right\} \label{eq:dual_sp8}\\
&\geq  \sup_{0< \alpha\leq 1} \min_{\sigma\in\mathcal{S(H)}} \left\{  -\frac1{\alpha} \log\Tr\left[ \sum_{x\in\mathcal{X}}P(x) \left[ W_x^\alpha \sigma^{1-\alpha}\right] \right] - \frac{1-\alpha}{\alpha} R \right\}	\label{eq:dual_sp9}\\
&= E_\text{sp}(R,P), \label{eq:dual_sp11}
\end{align}
where the last equality follows from Eq.~\eqref{eq:Esp_Sibson}.
		
Finally, we invoke the following identity proved by Mosonyi and Ogawa \cite[Proposition 4.2]{MO14b}:
\begin{align}
\max_{P\in\mathscr{P}(\mathcal{X})} \min_{\sigma\in\mathcal{S(H)}}    D_{\alpha} \left(\mathscr{W}\|\sigma|P\right) 
&= \max_{P\in\mathscr{P}(\mathcal{X})} \min_{\sigma\in\mathcal{S(H)}} D_\alpha \left( P\circ \mathscr{W} \| P \otimes \sigma \right). \label{eq:g_Hol}
\end{align}
Note that the above relation also holds for $D_\alpha^\flat$.
Then, combining Eqs.~\eqref{eq:sp_1P}, \eqref{eq:Esp_Sibson} and \eqref{eq:sp_1}, Eq.~\eqref{eq:dual_sp03} holds as follows:
\begin{align}
E_\text{sp}(R) &= \max_{P\in\mathscr{P}(\mathcal{X})} E_\text{sp}(R,P) \\
&=  \max_{P\in\mathscr{P}(\mathcal{X})} E_\text{sp}^{(1)}(R,P) \\
&= \max_{P\in\mathscr{P}(\mathcal{X})} \min_{\sigma\in\mathcal{S(H)}} D_\alpha \left( P\circ \mathscr{W} \| P \otimes \sigma \right) \\
&=  \max_{P\in\mathscr{P}(\mathcal{X})} \min_{\sigma\in\mathcal{S(H)}}    D_{\alpha} \left(\mathscr{W}\|\sigma|P\right).
\end{align}

		
\end{proof}
	
The following corollary is a simple consequence of the variational representations of the sphere-packing exponents in Theorem \ref{theo:dual_sp} and Eq.~\eqref{eq:compare} .
\begin{coro}
For any classical-quantum channel $\mathscr{W}:\mathcal{X}\to \mathcal{S(H)}$, $R> C_{0,\mathscr{W}}$, and $P\in\mathscr{P}(\mathcal{X})$, it holds that
\begin{align}
E_\textnormal{sp} (R,P) \leq \widetilde{E}_\textnormal{sp} (R,P).
\end{align}
\end{coro}

\section{Finite Blocklength Sphere-Packing Bound} \label{sec:main}
The main result in the section is a finite blocklength strong sphere-packing bound for c-q channels with a polynomial prefactor (Theorem \ref{theo:refined}), improving upon a subexponential prefactor obtained in \cite{Dal13}. To establish this result, 
the key is to employ a hypothesis testing reduction and  a sharp concentration inequality ~\cite{BR60, AW14}.   
Our proof consists of three major steps: (i) reducing the channel coding problem to binary hypothesis testing (Proposition~\ref{lemm:hypothesis}); (ii) bounding its type-I error from below (Propositions \ref{prop:Chebyshev} and \ref{prop:sharp}); (iii) employing Theorem \ref{theo:dual_sp} to relate the derived bound to the strong sphere-packing exponent. 
We discuss Propositions~\ref{lemm:hypothesis}, \ref{prop:Chebyshev}, and \ref{prop:sharp} in Section~\ref{ssec:converse_HT}, and prove the main result Theorem \ref{theo:refined} in Section~\ref{ssec:Proof}.

\begin{theo}[Refined Strong Sphere-Packing Bound of Constant Composition Codes] \label{theo:refined}
Consider a classical-quantum channel $\mathscr{W}: \mathcal{X}\to \mathcal{S(H)}$ and  $R \in(C_{0,\mathscr{W}}, C_\mathscr{W}) $.  For every $\gamma > 0$, there exist an $N_0\in\mathbb{N}$ and a constant $A>0$ such that for all constant composition codes $\mathcal{C}_n$ of length $n\geq N_0$ with message size $|\mathcal{C}_n|\geq \exp\{nR\}$, we have
\begin{align} \label{eq:goal}
\bar{\epsilon}\left( \mathcal{C}_n \right) \geq  \frac{A}{n^{\frac12\left(1 + |E_\textnormal{sp}'(R)| + \gamma  \right)} }  \exp\left\{ -n E_\textnormal{sp}(R)\right\}.
\end{align}
\end{theo}
	
The following corollary generalizes the refined sphere-packing bound for constant composition codes  to arbitrary codes via a standard argument \cite[p.~95]{SGB67}. 

\begin{coro}[Refined Strong Sphere-Packing Bound for General Codes] \label{coro:refined}
Consider a classical-quantum channel $\mathscr{W}: \mathcal{X}\to \mathcal{S(H)}$ and  $R\in(C_{0,\mathscr{W}}, C_\mathscr{W}) $.  There exist some $t> 1/2$ and $N_0\in\mathbb{N}$ such that for all codes of length $n\geq N_0$, we have
\begin{align} \label{eq:goal_general}
\epsilon^*\left(n, R\right) \geq  n^{-t} \exp\left\{-n E_\textnormal{sp}(R)\right\}.
\end{align}
\end{coro}
\noindent Proofs for Theorem \ref{theo:refined} and Corollary \ref{coro:refined} are provided in Section \ref{ssec:Proof}.
\medskip

Theorem~\ref{theo:refined} yields
\begin{align} \label{eq:converse}
\log \frac{1}{\bar{\epsilon}(\mathcal{C}_n)} \leq n E_\text{sp}(R) + \frac12\left(1+ \left|E_\text{sp}'(R)\right| \right) \log n + o(\log n),
\end{align}
where the term $\frac12\left(1+ \left|E_\text{sp}'(R)\right| \right)$ can be viewed as a second-order term (see the discussions in \cite[Section 4.4]{Tan14}). On the other hand, for the case of classical \emph{non-singular} channels\footnote{For classical \emph{singular} channels, one has 
$\log \frac{1}{\bar{\epsilon}(\mathcal{C}_n)} \geq n E_\text{r}(R) + \frac12 \log n + \Omega(1)$ \cite{Sca14}. Further, it was conjectured that \cite{AW14c} that  $\log \frac{1}{\bar{\epsilon}(\mathcal{C}_n)} \leq n E_\text{sp}(R) + \frac12 \log n + o(\log n),$ for all asymmetric classical singular channels and constant composition codes. However, such a result remains open.}, it was shown that \cite[Theorem 3.6]{Sca14}, for all constant composition codes $\mathcal{C}_n$ and rate $R\in(R_\text{crit}, C_\mathscr{W})$,
\begin{align} \label{eq:achieve}
\log \frac{1}{\bar{\epsilon}(\mathcal{C}_n)} \geq n E_\text{r}(R) + \frac12\left(1+ \left|E_\text{r}'(R)\right| \right) \log n + \Omega(1),
\end{align}
where $E_\text{r}(R)$ is the \emph{random coding exponent}, and $R_\text{crit}$ is the critical rate such that $E_\text{r}(R) = E_\text{sp}(R)$ for all $R\geq R_\text{crit}$ \cite[p.~160]{Gal68}, \cite{HM16}.	Hence our result, Theorem~\ref{theo:refined}, matches the achievability up to the logarithmic order. We note that whether the third order $o(\log n)$ in Eq.~(\ref{eq:converse}) can be improved to $O(1)$ is still unknown even for the classical case.

\subsection{Converse Bounds for Quantum Hypothesis Testing} \label{ssec:converse_HT}
This section contains the hypothesis testing reduction method (Proposition~\ref{lemm:hypothesis}) and two converse bounds for the optimal type-I error of a composition (Propositions~\ref{prop:Chebyshev} and \ref{prop:sharp}).
First of all, Proposition~\ref{lemm:hypothesis} allows us to lower bound the optimal error probability of a code to the optimal type-I error in quantum hypothesis testing.
To further lower bound the type-I error, we prove a converse Hoeffding bound for a binary hypothesis testing of product states in Theorem~\ref{theo:sharp_Hoeffding}.
Choosing the hypotheses in Theorem~\ref{theo:sharp_Hoeffding} as  a channel output state of a codeword against an $n$-copy state and employing a uniform continuity (Proposition~\ref{prop:UC} in Appendix~\ref{app:UC}), we establish a finite blocklength Chebyshev-type converse bound (Proposition~\ref{prop:Chebyshev}), and a sharp converse bound (Proposition~\ref{prop:sharp}), respectively. In Section~\ref{ssec:Proof}, we will combine these two bounds to prove the desired sphere-packing bound (Theorem~\ref{theo:refined}).

We remark that the converse bounds (Propositions~\ref{prop:Chebyshev} and \ref{prop:sharp}) for fixed composition of codewords hold for all sufficiently large $n$. Moreover, we call such bounds \emph{finite blocklength bounds} because the blocklength $n$ only depends on the channel $\mathscr{W}$, coding rate $R$, and $|\mathcal{X}|$. The independence of the composition of codes allows us to establish finite blocklength sphere-packing bounds. The key ingredient here is to prove a non-trivial uniform continuity property in the large deviation regime (Appendix~\ref{app:UC}).

In the following, we show Proposition~\ref{lemm:hypothesis} and its proof.
We note that Proposition~\ref{lemm:hypothesis} below is similar to the meta-converse in Ref.~\cite{PPV10}. However, the idea dates back to Fano \cite{Fan61}, Shannon-Gallager-Berlekamp \cite{SGB67}, and Blahut \cite{Bla74,Bla87}.
\begin{prop} \label{lemm:hypothesis}
For any classical-quantum channel $\mathscr{W}: \mathcal{X}\to \mathcal{S(H)}$ and any code $\mathcal{C}$ with message size $M$, it follows that
\begin{align}
\epsilon_{\max}\left(\mathcal{C}\right) \geq \max_{  \sigma \in \mathcal{S(H)} } \min_{x\in\mathcal{C}}\widehat{\alpha}_{\frac1M} \left( W_x \| \sigma \right). 
\end{align}		
\end{prop}
\begin{proof}
Let $x(m)$ be the codeword encoding the message $m \in \{1,\ldots,M\}$. Define a binary hypothesis testing problem:
\begin{align}
&\mathsf{H}_0: W_x, \\ 
&\mathsf{H}_1: \sigma,
\end{align}
where $\sigma \in \mathcal{S}\left(\mathcal{H}\right)$ can be viewed as a dummy channel output.	Since $\sum_{m=1}^M \beta\left( \Pi_{m}; \sigma \right) = 1$ for any POVM $\Pi = \{\Pi_{1},\ldots,\Pi_{M}  \}$, and $\beta\left( \Pi_{m} ; \sigma \right) \geq 0$ for every $m \in\mathcal{M}$, there must exist a message $m \in\mathcal{M}$ for any  code $\mathcal{C}$ such that $\beta\left( \Pi_{m} ; \sigma \right)\leq \frac1M $. Fix such $x := x \left(m\right)$. 	Then 
\begin{align} \label{eq:sketch7}
\epsilon_{\max}\left(\mathcal{C}\right) \geq \epsilon_{m}\left( \mathcal{C} \right) = \alpha\left(\Pi_{m};  W_{x}  \right) \geq \widehat{\alpha}_{\frac1M} \left( W_{x} \| \sigma \right) \geq \min_{x\in\mathcal{C}} \widehat{\alpha}_{\frac1M} \left( W_x \| \sigma \right).
\end{align}
Since the above inequality \eqref{eq:sketch7} holds for every $\sigma \in \mathcal{S}\left(\mathcal{H}\right)$, it follows that
\begin{align}
\epsilon_{\max}\left(\mathcal{C}\right) \geq \max_{  \sigma \in \mathcal{S(H)} } \min_{x\in\mathcal{C}} \widehat{\alpha}_{\frac1M} \left( W_x \| \sigma \right).
\end{align}	
\end{proof}

Before showing the following converse Hoeffding bound, Theorem~\ref{theo:sharp_Hoeffding}, we first introduce some notation.
Let 
\begin{align} \label{eq:product}
&\mathsf{H}_0: \rho^n = \rho_1\otimes \cdots \otimes \rho_n;\\
&\mathsf{H}_1: \sigma^n = \sigma_1\otimes \cdots \otimes \sigma_n,
\end{align}
where $\rho_x,\sigma_x\in\mathcal{S(H)}$ for $x\in[n]$. Further, denote by $(p_i,q_i)$ be the Nussbaum-Szko{\l}a distribution of $(\rho_i,\sigma_i)$ given in Section~\ref{ssec:NS}.
For $\alpha\in[0,1]$, define
\begin{align}
B_\alpha(\rho^n\|\sigma^n) &:= \frac1n \sum_{x\in[n]} \mathbb{E}_{ v_{x,\alpha}} \left[  \log \frac{p_x}{q_x} \right]; \\
V_\alpha(\rho^n\|\sigma^n) &:= \frac1n \sum_{x\in[n]} \mathbb{E}_{ v_{x,\alpha}} \left[  \left|\log \frac{p_x}{q_x} - \mathbb{E}_{ v_{x,\alpha}} \left[  \log \frac{p_x}{q_x} \right]\right|^2 \right];
\\
T_\alpha(\rho^n\|\sigma^n) &:= \frac1n \sum_{x\in[n]} \mathbb{E}_{ v_{x,\alpha}} \left[  \left|\log \frac{p_x}{q_x} - \mathbb{E}_{ v_{x,\alpha}} \left[  \log \frac{p_x}{q_x} \right] \right|^3 \right],
\end{align}
where $(p_x, q_x)$ is the Nussbaum-Szko{\l}a distribution of $(\rho_x, \sigma_x)$ for $x\in[n]$, and the \emph{tilted distribution} is 
\begin{align}
v_{x,\alpha}(i,j) := \frac{ p_x^\alpha(i,j) q_x^{1-\alpha}(i,j) }{ \sum_{\imath,\jmath} p_x^\alpha(\imath,\jmath) q_x^{1-\alpha}(\imath,\jmath)  }, \quad \alpha\in[0,1].
\end{align}
With the above notation, we have the following converse bound.
\begin{theo}[Sharp Converse Hoeffding Bounds for Quantum Hypothesis Testing] \label{theo:sharp_Hoeffding}
	Consider a binary hypothesis testing: $\mathsf{H}_0: \rho^n = \bigotimes_{i=1}^n \rho_i $ and $\mathsf{H}_1: \sigma^n = \bigotimes_{i=1}^n \sigma_i$ given in Eq.~\eqref{eq:product} with $\rho^n \ll \sigma^n$.
	Let $r\in\mathbb{R}$ be such that there exists an $ \alpha^\star \in (0,1]$ such that
	\begin{align} \label{eq:alpha_star}
	\phi_n \left( r | \rho^n\|\sigma^n \right) = 
	\frac{ 1 - \alpha^\star }{ \alpha^\star} \left( \frac{1}{n} D_{\alpha^\star}\left( \rho^n \| \sigma^n \right) - r \right).
	\end{align}

	Then, we have the following: (i) for any test $Q_n$,
	either
	\begin{align}
	\alpha\left(Q^n; \rho^n \right) \geq  \frac18 \exp\left\{-n \phi_n( r |\rho^n\|\sigma^n)  - \alpha^\star \sqrt{2n V_{\alpha^\star}(\rho^n\|\sigma^n) } \right\}, \label{eq:lower_alpha_ch}
	\end{align}
	or 
	\begin{align}
	\beta\left(Q^n; \sigma^n\right) \geq \frac18 \exp\left\{-n r  - (1-\alpha^\star) \sqrt{2n V_{\alpha^\star}(\rho^n\|\sigma^n) } \right\} \label{eq:lower_beta_ch}
	\end{align}
	holds;
	(ii) if $\alpha^\star \in (0,1)$, then for any test $Q_n$,
	either
	\begin{align}
	\alpha\left(Q^n; \rho^n \right) \geq  \mathrm{e}^{-n \phi_n( r |\rho^n\|\sigma^n)}
	\frac{  \mathrm{e}^{- K_n(\alpha^\star)} }{ 2\sqrt{2n\pi   V_{\alpha^\star}(\rho^n\|\sigma^n) } } \left( 1 - \frac{ 1 + (1+K_n(\alpha^\star)^2) }{ 2 \sqrt{ n V_{\alpha^\star}(\rho^n\|\sigma^n)  } } \right), \label{eq:sharp0}
	\end{align}
	or
	\begin{align}
	\beta\left(Q^n; \sigma^n\right) \geq \mathrm{e}^{-n r }
	\frac{  \mathrm{e}^{- K_n(\alpha^\star )} }{ 2\sqrt{2n\pi  V_{\alpha^\star}(\rho^n\|\sigma^n) } } \left( 1 - \frac{ 1 + (1+K_n(\alpha^\star)^2) }{ 2 \sqrt{ n V_{\alpha^\star}(\rho^n\|\sigma^n)  } } \right) \label{eq:sharp01}
	\end{align}
	holds. Here, $K_n (\alpha) := \frac{ 15\sqrt{2\pi} T_\alpha(\rho^n\|\sigma^n)}{  V_\alpha(\rho^n\|\sigma^n)  } \in\mathbb{R}_{>0}$ and $V_{\alpha^\star}(\rho^n\|\sigma^n) \in\mathbb{R}_{>0}$.
\end{theo}
\begin{proof}
	The first claim directly follows from Dailai's result in \cite[Theorem 4]{Dal13}.
	Before proceeding, we need to introduce some notation.
	Let $\tilde{p}^n := \bigotimes_{i=1}^n \tilde{p}_{i}$ and $\tilde{q}^n := \bigotimes_{i=1}^n \tilde{q}_{i}$, where $(\tilde{p}_{i}, \tilde{q}_{i})$ are the Nussbaum-Szko{\l}a distributions \cite{NS09} of $(\rho_{i}, \sigma_i)$ for $i\in [n]$.
	Since $D_\alpha(\rho_{i}\|\sigma_{i}) = D_\alpha(\tilde{p}_i\|\tilde{q}_i)$, for all $\alpha\in(0,1)$, we shorthand 
	\begin{align}
	\phi_n(r) := \phi_n\left( r|\rho^n\| \sigma^n  \right)  = \phi_n(r|\tilde{p}^n\|\tilde{q}^n) \in \mathbb{R}_{>0}. \label{eq:sharp220}
	\end{align} 
	
	Applying Nagaoka's argument \cite{Nag06}, for any $0\leq Q_n\leq \mathds{1}$ with $\delta = \exp\{nr - n\phi_n(r)\}$, we have
	\begin{align}  \label{eq:sharp15}
	\alpha\left(Q_n; \rho^n \right) + \delta \beta\left(Q_n; \sigma^n \right) \geq  \frac12 \left( \alpha\left(\tilde{\mathscr{U}}; \tilde{p}^n \right) + \mathrm{e}^{nr-n\phi_n(r)} \beta\left( \tilde{\mathscr{U}};\tilde{q}^n\right) \right),
	\end{align}
	where $\alpha\left( \tilde{\mathscr{U}};  \tilde{p}^n \right) := \sum_{\omega \notin \tilde{\mathscr{U}} }  \tilde{p}^n(\omega)$, $\beta\left( \tilde{\mathscr{U}}; \tilde{q}^n\right) := \sum_{\omega\in \tilde{\mathscr{U}} }  \tilde{q}^n(\omega)$,
	and  
	\begin{align} \label{eq:tilde_U}
	\tilde{\mathscr{U}} &:= \left\{  \omega:  \tilde{p}^n(\omega)\mathrm{e}^{ n{\phi}_n\left( {r}\right)} >  \tilde{q}^n(\omega) \mathrm{e}^{ n{r}}  \right\}.
	\end{align}	
	Now, we further define the non-normalized distributions ${p}^n := \bigotimes_{i=1}^n {p}_{i}$ and ${q}^n := \bigotimes_{i=1}^n q_{i}$, where $p_i := \tilde{p}_i q_i^0$, $q_i := \tilde{q}_i p_i^0$, for every $i\in[n]$.
	Namely, we restrict $(p^n, q^n)$ to be in the joint support of $\tilde{p}^n$ and $\tilde{q}^n$.
	Letting 
	\begin{align} \label{eq:U}
	{\mathscr{U}} &:= \left\{  \omega:  {p}^n(\omega)\mathrm{e}^{ n{\phi}_n\left( {r}\right)} >  {q}^n(\omega) \mathrm{e}^{ n{r}}  \right\},
	\end{align}	
	it is not hard to see that
	\begin{align}
	\alpha\left( {\mathscr{U}};  {p}^n \right) &= \alpha\left( \tilde{\mathscr{U}};  \tilde{p}^n \right); \\
	\beta\left( {\mathscr{U}}; {q}^n\right) &= \beta\left( \tilde{\mathscr{U}}; \tilde{q}^n\right) \\
	\phi_n(r) &= \phi_n(r|p^n\|q^n).
	\end{align}
	Hence, we focus on the pair $(p^n, q^n)$ and the decision region $\mathscr{U}$ onwards.
	
	Let
	\begin{align}
	\begin{split}
	\Lambda_{0,n}(\alpha) := \frac1n \sum_{i\in[n]} \Lambda_{0,i}(\alpha), \quad &\Lambda_{0,i}(\alpha):=\log \mathbb{E}_{ {p}_i } \left[ \mathrm{e}^{(1-\alpha) \log \frac{ {q}_i }{ {p}_i } } \right]; \\
	\Lambda_{1,n}(\alpha) := \frac1n \sum_{i\in[n]} \Lambda_{0,i}(\alpha), \quad
	&\Lambda_{1,i}(\alpha) := \log \mathbb{E}_{ {q}_i } \left[ \mathrm{e}^{(1-\alpha) \log \frac{ {p}_i }{ {q}_i } } \right].
	\end{split}	
	\end{align}
	Since $p^n$ and $q^n$ have the same support, both $\Lambda_{0,n}(\alpha)$ and $\Lambda_{1,n}(\alpha)$ are smooth functions in $\alpha\in \mathbb{R}$. One can the calculate their derivatives as follows:
	\begin{align}
	\Lambda_{0,n}'(\alpha)  = \frac{1}{n} \sum_{i\in[n]} \mathbb{E}_{ v_{i,\alpha} } \left[  \log \frac{ {p}_i }{ {q}_i }  \right]; \quad &\Lambda_{1,n}'(\alpha)  = \frac{1}{n} \sum_{i\in[n]} \mathbb{E}_{ v_{i,1-\alpha} } \left[  \log \frac{ {q}_i }{ {p}_i }  \right] \\
	\Lambda_{0,n}''(\alpha)  = \frac{1}{n}  \sum_{i\in[n]} \text{Var}_{ v_{i,\alpha} } \left[  \log \frac{ {p}_i }{ {q}_i }  \right]; \quad &\Lambda_{1,n}''(\alpha)  = \frac{1}{n}  \sum_{i\in[n]} \text{Var}_{ v_{i,1-\alpha} } \left[  \log \frac{ {q}_i }{ {p}_i }  \right], \\
	\begin{split}
	T_{0,n }(\alpha) :=  \frac1n  \sum_{i\in[n]}\mathbb{E}_{v_{i,\alpha}} &\left[ \left| \log \frac{ {p}_i}{ {q}_i} - \Lambda'_{0,n}(\alpha) \right|^3 \right]; \\ 
	T_{1,n }(\alpha) :=  \frac1n  \sum_{i\in[n]} \mathbb{E}_{v_{i,1-\alpha}} &\left[ \left| \log \frac{ {q}_i}{ {p}_i} - \Lambda'_{1,n}(\alpha) \right|^3 \right],
	\end{split}
	\end{align}
	where we denote the tilted distribution by
	\begin{align}
	v_{i,\alpha}(\omega) := \frac{ p_i^\alpha(\omega) q_i^{1-\alpha}(\omega) }{ \sum_{\bar{\omega}} p_i^\alpha( \bar{\omega}) q_i^{1-\alpha} (\bar{\omega})  }, \quad \alpha\in[0,1].
	\end{align}
	Further, it is not hard to verify that for all $\alpha\in[0,1]$.
	\begin{align} 
	\begin{split}\label{eq:sharp_sym}
	&\Lambda_{0,n}(\alpha) = \Lambda_{1,n}(1-\alpha); \quad 
	\Lambda_{0,n}'(\alpha) = -\Lambda_{1,n}'(1-\alpha); \\
	&\Lambda_{0,n}''(\alpha) = \Lambda_{1,n}''(1-\alpha);\quad
	T_{0,n}(\alpha) = T_{1,n}(1-\alpha).	
	\end{split}
	\end{align}
	Next, we define the \emph{Lengendre-Fenchel transform}:
	\begin{align}
	&\Lambda_{j,n}^*(z) := \sup_{\alpha \in\mathbb{R}} \left\{ (1-\alpha) z - {\Lambda}_{j,n}(\alpha)  \right\}, \quad j\in\{0,1\}. \label{eq:FL}
	\end{align}
	The quantities $\Lambda_{j,n }^*(z)$ would appear in the lower bounds of $\alpha\left( Q^n;  \rho^n \right)$ and $\beta\left( Q^n; \sigma^n \right)$ as shown later.
	
	Now, we are ready to show the first claim. Ref.~\cite[Theorem~4]{Dal13} states that for any test $Q^n$, either
	\begin{align}
	\alpha\left(Q^n; \rho^n \right) \geq  \frac18 \exp\left\{-n \left[ \alpha^\star \Lambda_{0,n}'\left(\alpha^\star\right) - \Lambda_{0,n}\left( \alpha^\star \right)  \right] 
	- \alpha^\star \sqrt{ 2 n V_{\alpha^\star}\left( \rho^n\|\sigma^n \right) }
	\right\},
	\end{align}
	or 
	\begin{align}
	\beta\left(Q^n; \sigma^n\right) \geq \frac18 \exp\left\{-n \left[ -(1-\alpha^\star) \Lambda_{0,n}'\left(\alpha^\star\right) + \Lambda_{0,n}\left( \alpha^\star \right)  \right] 
	- (1-\alpha^\star) \sqrt{ 2 n V_{\alpha^\star}\left( \rho^n\|\sigma^n \right) }
	\right\},
	\end{align}
	holds. By Eqs.~\eqref{eq:sharp_sym}, and \eqref{eq:regularity9}, \eqref{eq:regularity2} in Appendix~\ref{app:LF}, we have 
	\begin{align}
	\phi_n(r) &= \alpha^\star \Lambda_{0,n}'\left(\alpha^\star\right) - \Lambda_{0,n}\left( \alpha^\star \right); \\
	r &=  -(1-\alpha^\star) \Lambda_{0,n}'\left(\alpha^\star\right) + \Lambda_{0,n}\left( \alpha^\star \right),
	\end{align}
	which proves the first claim.
	
	To show the second claim, we will employ Bahadur-Ranga Rao's concentration inequality, Theorem \ref{theo:Rao}, in Appendix \ref{app:tight}, to further lower bound $\alpha\left( {\mathscr{U}};  {p}^n \right)$ and $\beta\left( {\mathscr{U}}; {q}^n\right)$. 
	Letting $Z_i = \log  \frac{ {q}_i }{ {p}_i }$ with probability measure $\mu_i =  {p}_i$, and $z =  {r} -  {\phi}_n( {r}) $ in Theorem \ref{theo:Rao}, the Bahadur-Randga Rao's inequality gives
	\begin{align}
	\alpha\left( {\mathscr{U}};  {p}^n \right) &:= \sum_{\omega\notin {\mathscr{U}}}  {p}^n(\omega)  \\
	&= \Pr\left\{ \frac1n\sum_{i=1}^n Z_i \geq   {\phi}_n( {r}) - r \right\} \\	
	&\geq  \exp\left\{  -n \Lambda^*_{0,n } \left(  {\phi}_n( {r}) -  {r} \right)  \right\}  
	\frac{  \mathrm{e}^{- K_n(\alpha^\star)} }{ \sqrt{2\pi \Lambda_n''(\alpha^\star) } } \left( 1 - \frac{ 1 + (1+K_n(\alpha^\star)^2) }{ 2 \sqrt{ \Lambda_{0,n}''(\alpha^\star)  } } \right), 
	\end{align}
	where 
	\begin{align}
	K_n(\alpha) := 15\sqrt{2\pi}\frac{  T_\alpha(\rho^n\|\sigma^n)}{  V_\alpha(\rho^n\|\sigma^n)  } = 15\sqrt{2} \frac{ T_{0,n}(\alpha) }{  \Lambda_{0,n}''(\alpha) }.
	\end{align}
	Moreover, Lemma \ref{lemm:regularity} in Appendix \ref{app:LF} relates the Legendre-Fenchel transform $\Lambda_{j,P_{\mathbf{x}^n}}^* (z)$ to the desired error exponent function $\phi_n(r)$: 	\begin{align}
	&\Lambda_{0,n} '' (\alpha^\star) >0; \label{eq:reg1}\\
	&\Lambda^*_{0,n}  \left(  {\phi}_n({r}) - {r} \right) =  {\phi}_n(r); \label{eq:reg2}\\
	&\Lambda^*_{1,n}  \left( r -  {\phi}_n({r}) \right) = r. \label{eq:reg3}
	\end{align}
	Hence, we have 
	\begin{align}
	\alpha\left( {\mathscr{U}};  {p}^n \right) \geq \exp\left\{  -n   {\phi}_n( {r})   \right\}  
	\frac{  \mathrm{e}^{- K_n(\alpha^\star)} }{ \sqrt{2\pi \Lambda_{0,n}''(\alpha^\star) } } \left( 1 - \frac{ 1 + (1+K_n(\alpha^\star)^2) }{ 2 \sqrt{ \Lambda_{0,n}''(\alpha^\star)  } } \right). \label{eq:sharp12mh}
	\end{align}
	
	Similarly, applying Theorem \ref{theo:Rao} with $Z_i = \log  \frac{ {p}_i }{ {q}_i }$, $\mu_i =  {q}_i$, and $z =  {\phi}_n( {r}) -  {r}$ yields
	\begin{align}
	\beta\left( {\mathscr{U}}; {q}^n\right) &:= \sum_{\omega\in {\mathscr{U}}}  {q}^n(\omega) \\
	&= \Pr\left\{ \frac1n\sum_{i=1}^n Z_i \geq  {\phi}_n( {r}) -  {r}  \right\} \\
	&\geq   \exp\left\{  -n \Lambda^*_{1,n } \left(  r - {\phi}_n( {r})  \right)  \right\}  
		\frac{  \mathrm{e}^{- K_n(1-\alpha^\star)} }{ \sqrt{2\pi \Lambda_{1,n}''(1-\alpha^\star) } } \left( 1 - \frac{ 1 + (1+K_n(1-\alpha^\star)^2) }{ 2 \sqrt{ \Lambda_{1,n}''(1-\alpha^\star)  } } \right) \\
	&=  \exp\left\{  -n r\right\}  
	\frac{  \mathrm{e}^{- K_n(1-\alpha^\star)} }{ \sqrt{2\pi \Lambda_{1,n}''(1-\alpha^\star) } } \left( 1 - \frac{ 1 + (1+K_n(1-\alpha^\star)^2) }{ 2 \sqrt{ \Lambda_{1,n}''(1-\alpha^\star)  } } \right) \\ 
	&=  \exp\left\{  -n r\right\}  
	\frac{  \mathrm{e}^{- K_n(\alpha^\star)} }{ \sqrt{2\pi \Lambda_{0,n}''(\alpha^\star) } } \left( 1 - \frac{ 1 + (1+K_n(\alpha^\star)^2) }{ 2 \sqrt{ \Lambda_{0,n}''(\alpha^\star)  } } \right),
	\label{eq:sharp13mh}.
	\end{align}
	where the last equality follows from Eq.~\eqref{eq:sharp_sym}.
	Hence, by Eqs.~\eqref{eq:sharp15}, \eqref{eq:sharp12mh}, and \eqref{eq:sharp13mh}, we conclude our claim.
\end{proof}
	
In the following Proposition, we will prove a finite blocklength  bound of a composition via Eqs.~\eqref{eq:lower_alpha_ch} and \eqref{eq:lower_beta_ch} in the above Theorem~\ref{theo:sharp_Hoeffding}. 
The difficulty of deriving a finite blocklength result is that one needs to obtain some universal coefficients independent of all possible compositions. Our core technique here is a uniform continuity property, Proposition~\ref{prop:UC}, which will be presented in Appendix~\ref{app:UC}.

The following result is essentially a Chebyshev-type bound with prefactor $\exp\{O(\sqrt{n})\}$. We will employ it to lower bound the error of ``bad sequences" that yield a inferior error exponent in Section \ref{ssec:Proof}.

\begin{prop}[Chebyshev-Type Bound for a Fixed Composition]\label{prop:Chebyshev}
	Let $\mathscr{W}:\mathcal{X}\to\mathcal{S(H)}$ be a classical-quantum channel. Fix $R\in(C_{0,\mathscr{W}}, C_\mathscr{W})$. 
	Consider a sequence $\mathbf{x}^n \in \mathcal{X}^n$ 
	Then, for every $c>0$, 
	there exist a state $\sigma^\star\in\mathcal{S(H)}$, an integer $N_0\in\mathbb{N}$, independent of the sequences $\mathbf{x}^n$ and $\sigma$, such that for all $n\geq N_0$ we have
	\begin{align} \label{eq:che-type}
	\widehat{\alpha}_{ c\exp\{ -nR \}  } \left( W_{\mathbf{x}^n}^{\otimes n}\| (\sigma^\star)^{\otimes n} \right) \geq \exp\left\{ - A\sqrt{n} -n E_\textnormal{sp}^{(2)}\left(  R , P_{\mathbf{x}^n}\right) 	
	\right\},
	\end{align}	
	where  $A\in\mathbb{R}_{>0}$ is a finite positive constant depending on $R$ and $\mathscr{W}$.
\end{prop}
\begin{proof}
	Fix an arbitrary $\underline{R} \in \left( C_{0,\mathscr{W}}, R \right)$. Let $\gamma_n := \frac{a\sqrt{n}}{2n} + \frac{\log 8 - \log c}{n}$ and $R_n := R - \gamma_n$ for some $a\in\mathbb{R}$. The choice of $a$ and the rate back-off term $\gamma_n$ will become evident later. Let $N_1 \in\mathbb{N}$ such that $R_n \in [ \underline{R} , R ]$ for all $n\geq N_1$. Subsequently, we choose such $n\geq N_1$ onwards.
	
	We choose the optimal output state as
	\begin{align}
	\sigma^\star = \arg \min_{\sigma \in\mathcal{S(H)}} \sup_{0<\alpha\leq 1} \frac{1-\alpha}{\alpha}\left( D_\alpha\left(\mathscr{W}\|\sigma|P_{\mathbf{x}^n} \right) - R_n \right).
	\end{align}
	Let ${p}^n := \bigotimes_{i=1}^n p_{x_i}$ and ${q}^n := \bigotimes_{i=1}^n q_{x_i}$, where $(p_{x_i},q_{x_i})$ are Nussbaum-Szko{\l}a distributions \cite{NS09} of $(W_{x_i},\sigma_{R,P}^\star)$ for every $i\in [n]$. Since $D_\alpha( W_{x_i}\|\sigma^\star) = D_\alpha(p_{x_i}\|q_{x_i})$, for $\alpha\in(0,1]$, again we shorthand for all $R_n\in[\underline{R},R]$,
	\begin{align}
	\phi_n(R_n) := \phi_n\left( R_n|W_{\mathbf{x}^n}^{\otimes n}\| (\sigma^\star)^{\otimes n}  \right)  = \phi_n(R_n|p^n\|q^n) = E_\textnormal{sp}^{(2)}\left( R_n , P_{\mathbf{x}^n} \right), \label{eq:sharp22_ch}
	\end{align} 
	where the last equality in Eq.~\eqref{eq:sharp22_ch} follows from the saddle-point property, item \ref{saddle-a} in Proposition~\ref{prop:saddle}. Moreover, item \ref{saddle-b} in Proposition \ref{prop:saddle} implies that the state $\sigma^\star$ dominants all the states: $\sigma^\star \gg W_x$, for all $x\in \texttt{supp}(P_{\mathbf{x}^n})$, Hence, we have $p^n\ll q^n$. This guarantees that $V_\alpha\left( W_x\| \sigma^\star \right)$ is finite for all $\alpha\in[0,1]$ and all $x\in \texttt{supp}(P_{\mathbf{x}^n})$.
	
	Theorem~\ref{theo:sharp_Hoeffding} implies that for any test $Q^n$ 
	either
	\begin{align}
	\alpha\left(Q^n; W_{\mathbf{x}^n}^{\otimes n} \right) \geq  \frac18 \exp\left\{-n \phi_n( R_n )  - \alpha_{R_n,P_{\mathbf{x}^n}}^\star \sqrt{2n V_{\alpha_{R_n,P_{\mathbf{x}^n}}^\star}( P_{\mathbf{x}^n}, \mathscr{W} ) } \right\}, \label{eq:ch1}
	\end{align}
	or 
	\begin{align}
	\beta\left(Q^n; (\sigma^\star)^{\otimes n}\right) \geq \frac18 \exp\left\{-n R_n  - (1-\alpha_{R_n,P_{\mathbf{x}^n}}^\star) \sqrt{2n V_{\alpha_{R_n,P}^\star}( P_{\mathbf{x}^n}, \mathscr{W} ) } \right\},
	\end{align}
	where $\alpha_{r,P_{\mathbf{x}^n} }^\star \in(0,1)$ satisfies, for all $r\in[\underline{R},R]$,
	\begin{align}
	\phi_n(r) = \frac{ 1- \alpha_{r,P_{\mathbf{x}^n} }^\star }{ \alpha_{r,P_{\mathbf{x}^n} }^\star } \left( I_{\alpha_{r,P_{\mathbf{x}^n} }^\star}^{(2)}\left(P_{\mathbf{x}^n}, \mathscr{W} \right) - r \right).
	\end{align}
	and we define
	\begin{align}
	V_\alpha(P,\mathscr{W}) &:= \sum_{x\in\mathcal{X}} P(x) \mathbb{E}_{ v_{x,\alpha}} \left[  \left|\log \frac{p_x}{q_x} - \mathbb{E}_{ v_{ x,\alpha }} \left[  \log \frac{p_x}{q_x} \right]\right|^2 \right], \\
	v_{x,\alpha}(\omega) &:= \frac{ p_{x}^\alpha(\omega) q_{x}^{1-\alpha}(\omega) }{ \sum_{\bar{\omega}} p_{x}^\alpha( \bar{\omega}) q_{x}^{1-\alpha} (\bar{\omega})  }.
	\end{align}	
	Now, we will introduce a constant independent of $\mathbf{x}^n$ to obtain a finite blocklength lower bound for $\widehat{\alpha}_{c \exp\{-nR\}}(W_{\mathbf{x}^n}^{\otimes n}\|(\sigma^\star)^{\otimes n})$.
	Let
	\begin{align} \label{eq:Vmax_ch}
	V_{\max}(R) := \max_{ (r,P)\in [\underline{R},R]\times \mathscr{P}(\mathcal{X})} V_{\alpha_{r,P}^\star} (P , \mathscr{W}).
	\end{align}
	Recall that $V_\alpha\left( P_{\mathbf{x}^n}, \mathscr{W} \right)$ is finite for all $\alpha\in[0,1]$. Proposition~\ref{prop:UC} in Appendix~\ref{app:UC} implies that $V_{\alpha_{r,P}^\star}(P,\mathscr{W})$ is joint continuous on $[\underline{R},R]\times \mathscr{P}(\mathcal{X})$.
	Further, since $\mathscr{P}(\mathcal{X})$ is compact, the quantity in Eq.~\eqref{eq:Vmax_ch} is well-defined and finite.
	Therefore,
	\begin{align}
	\beta\left(Q^n; (\sigma^\star)^{\otimes n}\right) &\geq \frac18 \exp\left\{-n R_n  - (1-\alpha_{R_n,P_{\mathbf{x}^n}}^\star) \sqrt{2n V_{\alpha_{R_n,P_{\mathbf{x}^n}}^\star}( P_{\mathbf{x}^n}, \mathscr{W} ) } \right\} \\
	&\geq \frac18 \exp\left\{-n R_n  - \sqrt{2n V_{\max}(R) } \right\} \\
	&= c\exp\left\{ -nR \right\}, \label{eq:ch2}
	\end{align}
	where we choose $a:=  \sqrt{2 V_{\max}(R) } $ in the rate back-off term $\gamma_n := \frac{a\sqrt{n}}{2n} + \frac{\log 8 - \log c}{n}$.
	
	Next, Eqs.~\eqref{eq:ch1} and \eqref{eq:ch2} yield
	\begin{align}
	\widehat{\alpha}_{c \exp\{-nR\}}(W_{\mathbf{x}^n}^{\otimes n}\|(\sigma^\star)^{\otimes n}) &\geq \frac18 \exp\left\{-n \phi_n( R_n)  - \alpha_{R_n,P_{\mathbf{x}^n}}^\star \sqrt{2n V_{\alpha_{R_n,P_{\mathbf{x}^n}}^\star}( P_{\mathbf{x}^n}, \mathscr{W} ) } \right\} \\
	&\geq \frac18 \exp\left\{ - \sqrt{2n V_{\max}(R) -n E_\text{sp}^{(2)}( R - \gamma_n, P_{\mathbf{x}^n} )   } \right\}. \label{eq:Cheby1}
	\end{align}
	Further, the convexity and the monotone decreases of $r\mapsto E_\text{sp}^{(2)}(r,P)$ given in Proposition~\ref{prop:Esp}-\ref{Esp-a}
	shows that 
	\begin{align}
	E_\text{sp}^{(2)}( R - \gamma_n, P_{\mathbf{x}^n} ) &\leq E_\text{sp}^{(2)} (R, P_{\mathbf{x}^n} ) - \gamma_n \left.\frac{\partial E_\text{sp}^{(2)}( r, P_{\mathbf{x}^n} )}{\partial r}\right|_{r=R} , \\
	&\leq E_\text{sp}^{(2)} (R, P_{\mathbf{x}^n} ) - \gamma_n  \left.\frac{\partial E_\text{sp}^{(2)}( r, P_{\mathbf{x}^n} )}{\partial r}\right|_{r=\underline{R}}.
	\label{eq:Cheby2}
	\end{align}
	Next, we denote
	\begin{align}
	\Upsilon := \max_{P\in\mathscr{P}(\mathcal{X})} \left|  \left.\frac{\partial E_\text{sp}^{(2)}( r, P )}{\partial r}\right|_{r=\underline{R}}  \right|. \label{eq:Cheby3}
	\end{align}
	Observe that $\Upsilon \in \mathbb{R}_{\geq0}$ due to $\underline{R}> C_{0,\mathscr{W}}$ and item \ref{saddle-d} of Proposition \ref{prop:saddle}.
	Then, Eqs.~\eqref{eq:Cheby1}, \eqref{eq:Cheby2}, and \eqref{eq:Cheby3} lead to
	\begin{align}
	\widehat{\alpha}_{c\exp\{-n R\}} \left( W_{\mathbf{x}^n}^{\otimes n}\|(\sigma^\star)^{\otimes n} \right) \geq   \exp\left\{ -\log 8 -\sqrt{2n V_{\max}(R)} - \gamma_n \Upsilon - n E_\text{sp}^{(2)} \left( R, P_{\mathbf{x}^n}   \right) \right\}.
	\end{align}
	Since $\gamma_n = O(1/\sqrt{n})$, for any $A>\sqrt{2V_{\max}(R)}$, there exists a sufficiently large $N_2\in\mathbb{N}$ such that for all $n\geq N_2$.
	\begin{align}
	\log 8 + \sqrt{2n V_{\max}(R)} + \gamma_n \Upsilon \leq A\sqrt{n}.
	\end{align}
	By letting $N_0 := \max\{N_1, N_2\}$ completes the proof.
\end{proof}

The following Proposition \ref{prop:sharp} is a sharp converse bound with polynomial prefactors obtained from Eqs.~\eqref{eq:sharp0} and \eqref{eq:sharp01} in Theorem~\ref{theo:sharp_Hoeffding}, which in turn were proved by employing Bahadur-Ranga Rao's inequality (see Appendix \ref{app:tight}).  Similar to Proposition~\ref{prop:Chebyshev} presented before, we will employ the uniform continuity, Proposition~\ref{prop:UC}, given in Appendix~\ref{app:UC} to prove Proposition~\ref{prop:sharp}.
In Section \ref{ssec:Proof}, we will exploit this result to bound the error of ``good sequences" with a polynomial  prefactor.
	
\begin{prop}[Sharp Converse Bound for a Fixed Composition]\label{prop:sharp}
	Let $\mathscr{W}:\mathcal{X}\to\mathcal{S(H)}$ be a classical-quantum channel Fix $R\in(C_{0,\mathscr{W}}, C_\mathscr{W})$. 
	Consider a sequence $\mathbf{x}^n \in \mathcal{X}^n$ satisfying 	\begin{align} \label{eq:sharp_cond2}
	E_\textnormal{sp}^{(2)}\left( R , P_{\mathbf{x}^n} \right) \in [\nu, + \infty)
	\end{align}
	for some constant $\nu>0$.  
	Then, for every $c>0$, 
	there exist a state $\sigma^\star\in\mathcal{S(H)}$, an integer $N_0\in\mathbb{N}$, independent of the sequences $\mathbf{x}^n$ and $\sigma$, such that for all $n\geq N_0$ we have
	\begin{align} \label{eq:refine0}
	\widehat{\alpha}_{ c\exp\{ -nR \}  } \left( W_{\mathbf{x}^n}^{\otimes n}\| (\sigma^\star)^{\otimes n} \right) \geq \frac{A}{n^{\frac12\left( 1 + s^\star_{ R, P_{\mathbf{x}^n} }  \right)}} \exp\left\{ -n E_\textnormal{sp}^{(2)}\left(  R , P_{\mathbf{x}^n}\right)  \right\},
	\end{align}	
	where $s^\star_{R,P} := - \left.\frac{\partial E_\textnormal{sp}^{(2)} (r,P)  }{\partial r}\right|_{r=R}$, and $A\in\mathbb{R}_{>0}$ is a finite positive constant depending on $R, \nu$ and $\mathscr{W}$.	
	

\end{prop}

\begin{proof}
	Fix an arbitrary $\underline{R} \in \left( C_{0,\mathscr{W}}, R \right)$. Let $\gamma_n := \frac{\log n}{2n} + \frac{x}{n}$ and $R_n := R - \gamma_n$ for some $x\in\mathbb{R}$. The choice of $x$ and the rate back-off term $\gamma_n$ will become evident later. Let $N_1 \in\mathbb{N}$ such that $R_n \in [ \underline{R} , R ]$ for all $n\geq N_1$. Subsequently, we choose such $n\geq N_1$ onwards.
	
	We choose the optimal output state as
	\begin{align}
	\sigma^\star = \arg \min_{\sigma \in\mathcal{S(H)}} \sup_{0<\alpha\leq 1} \frac{1-\alpha}{\alpha}\left( D_\alpha\left(\mathscr{W}\|\sigma|P_{\mathbf{x}^n} \right) - R_n \right).
	\end{align}
	as in the proof of Proposition~\ref{prop:Chebyshev}.
	Let ${p}^n := \bigotimes_{i=1}^n p_{x_i}$ and ${q}^n := \bigotimes_{i=1}^n q_{x_i}$, where $(p_{x_i},q_{x_i})$ are Nussbaum-Szko{\l}a distributions \cite{NS09} of $(W_{x_i},\sigma^\star)$ for every $i\in [n]$. Since $D_\alpha( W_{x_i}\|\sigma^\star) = D_\alpha(p_{x_i}\|q_{x_i})$, for $\alpha\in(0,1]$, again we shorthand for all $R_n\in[\underline{R},R]$,
	\begin{align}
	\phi_n(R_n) := \phi_n\left( R_n| W_{\mathbf{x}^n}^{\otimes n} \| (\sigma^\star)^{\otimes n}  \right)  = \phi_n(R_n|p^n\|q^n) = E_\textnormal{sp}^{(2)}\left( R_n , P_{\mathbf{x}^n} \right), \label{eq:sharp22}
	\end{align} 
	where the last equality in Eq.~\eqref{eq:sharp22} follows from the saddle-point property, item \ref{saddle-a} in Proposition~\ref{prop:saddle}. Moreover, item \ref{saddle-b} in Proposition \ref{prop:saddle} implies that the state $\sigma^\star$ dominants all the states: $\sigma^\star \gg W_x$, for all $x\in \texttt{supp}(P_{\mathbf{x}^n})$, Hence, we have $p^n\ll q^n$. Without loss of generality, we set zero all elements of $q_{x_i}$ that do not lie in the support of $p_{x_i}$, i.e.~ $q_{x_i}(\omega) = 0$, $\omega \not\in \texttt{supp}(p_{x_i})$, $i\in[n]$, because those elements do not contribute in $\phi_n(R_n)$.
	
	Next, we define
	\begin{align}
	V_\alpha(P,\mathscr{W}) &:= \sum_{x\in\mathcal{X}} P(x) \mathbb{E}_{ v_{x,\alpha}} \left[  \left|\log \frac{p_x}{q_x} - \mathbb{E}_{ v_{ x,\alpha }} \left[  \log \frac{p_x}{q_x} \right]\right|^2 \right]; 
	\\
	T_\alpha(P,\mathscr{W}) &:= \sum_{x\in\mathcal{X}} P(x) \mathbb{E}_{ v_{x,\alpha}} \left[  \left|\log \frac{p_x}{q_x} - \mathbb{E}_{ v_{x,\alpha}} \left[  \log \frac{p_x}{q_x} \right] \right|^3 \right], \\
	v_{x,\alpha}(\omega) &:= \frac{ p_{x}^\alpha(\omega) q_{x}^{1-\alpha}(\omega) }{ \sum_{\bar{\omega}} p_{x}^\alpha( \bar{\omega}) q_{x}^{1-\alpha} (\bar{\omega})  },
	\end{align}	
	Applying Theorem~\ref{theo:sharp_Hoeffding}, we have for any test $Q^n$,
	\begin{align}
	&\alpha\left(Q^n; W_{\mathbf{x}^n}^{\otimes n} \right) \geq  \mathrm{e}^{-n \phi_n( R_n )}
	\frac{  \mathrm{e}^{- K_n(\alpha_{R_n,P_{\mathbf{x}^n}}^\star)} }{ 2\sqrt{2n\pi   V_{\alpha_{R_n,P_{\mathbf{x}^n}}^\star}( P_{\mathbf{x}^n},\mathscr{W} ) } } \left( 1 - \frac{ 1 + (1+K_n(\alpha_{R_n,P_{\mathbf{x}^n}}^\star)^2) }{ 2 \sqrt{ n V_{ \alpha_{R_n,P_{\mathbf{x}^n}}^\star }( P_{\mathbf{x}^n},\mathscr{W} )  } } \right) \label{eq:lower_alpha} \\
	&\beta\left(Q^n; (\sigma^\star)^{\otimes n}  \right) \geq \mathrm{e}^{-n R_n }
	\frac{  \mathrm{e}^{- K_n(\alpha_{R_n,P_{\mathbf{x}^n}}^\star)} }{ 2\sqrt{2n\pi   V_{\alpha_{R_n,P_{\mathbf{x}^n}}^\star}( P_{\mathbf{x}^n},\mathscr{W} ) } } \left( 1 - \frac{ 1 + (1+K_n(\alpha_{R_n,P_{\mathbf{x}^n}}^\star)^2) }{ 2 \sqrt{ n V_{\alpha_{R_n,P_{\mathbf{x}^n}}^\star}( P_{\mathbf{x}^n},\mathscr{W} )  } } \right), \label{eq:lower_beta}
	\end{align}
	where $K_n(\alpha) := 15\sqrt{2} \frac{ T_\alpha(P_{\mathbf{x}^n},\mathscr{W}) }{ V_\alpha(P_{\mathbf{x}^n},\mathscr{W})} $, and $\alpha_{r,P_{\mathbf{x}^n} }^\star \in(0,1)$ satisfies, for all $r\in[\underline{R},R]$,
	\begin{align}
	\phi_n(r) = \frac{ 1- \alpha_{r,P_{\mathbf{x}^n} }^\star }{ \alpha_{r,P_{\mathbf{x}^n} }^\star } \left( I_{\alpha_{r,P_{\mathbf{x}^n} }^\star}^{(2)}\left(P_{\mathbf{x}^n}, \mathscr{W} \right) - r \right).
	\end{align}
	
	In the following, we will remove the dependency of $R_n$ and $P_{\mathbf{x}^n}$ in $V_{(\cdot)}(\cdot)$ and $K_n(\cdot)$.
	Define  the following quantities:
	\begin{align}
	V_{\max} (R,\nu) &:= \max_{(r,P) \in [\underline{R},R]\times \mathscr{P}_{R,\nu} (\mathcal{X}) }  V_{\alpha_{r,P}^\star}\left( P, \mathscr{W} \right); \label{eq:Vmax} \\
	V_{\min} (R,\nu) &:= \min_{(r,P) \in [\underline{R},R]\times \mathscr{P}_{R,\nu} (\mathcal{X})  }  V_{\alpha_{r,P}^\star}\left( P, \mathscr{W} \right);\label{eq:Vmin} \\
	K_{\max} (R,\nu) &:= 15\sqrt{2\pi} \max_{(r,P) \in [\underline{R},R]\times \mathscr{P}_{R,\nu} (\mathcal{X})  } \frac{ T_{\alpha_{r,P}^\star}(P, \mathscr{W}) }{ V_{\alpha_{r,P}^\star} (P,\mathscr{W})}, \label{eq:Kmax}
	\end{align}
	where 
	\begin{align} \label{eq:PRnu}
	\mathscr{P}_{R,\nu}(\mathcal{X}) := \left\{ P \in\mathscr{P}(\mathcal{X}): \nu\leq E_\text{sp}^{(2)}(R,P_{\mathbf{x}^n}) \leq E_\text{sp}(R) < +\infty  \right\}
	\end{align}
	is a compact set owing to the continuity of $r\mapsto E_\text{sp}^{(2)}(r,P)$ given in Proposition~\ref{prop:Esp}. 
	Also, Proposition~\ref{prop:UC} in Appendix~\ref{app:UC} shows that the objective functions in Eqs.~\eqref{eq:Vmax}, \eqref{eq:Vmin}, and \eqref{eq:Kmax} are continuous functions on $ \mathscr{P}(\mathcal{X})$, which guarantees  the maximization and minimization in the above definitions are well-defined and finite.
	Further, the quantity $V_{\min}(R,\nu)$ is bounded away from zero because of the positivity given in Theorem~\ref{theo:sharp_Hoeffding}.

	Now, we are ready to derive the lower bounds for $\widehat{\alpha}_{c\exp\{-nR\}}\left( W_{\mathbf{x}^n}^{\otimes n}\|(\sigma^\star)^{\otimes n} \right)$. Let $N_2\in\mathbb{N}$ be sufficiently large such that for all $n\geq N_2$,
	\begin{align} \label{eq_largen}
	\sqrt{n} \geq \frac{ 1+ \left( 1 + K_{\max} (R,\nu) \right)^2 }{ \sqrt{ V_{\min} (R,\nu)  } }.
	\end{align}
	Then, Eqs.~\eqref{eq:lower_alpha} and \eqref{eq:lower_beta} give
	\begin{align}
	\alpha\left( Q^n;  W_{\mathbf{x}^n}^{\otimes n} \right) 
	&\geq \frac{A(R,\nu)}{\sqrt{n}} \exp\left\{  -n \phi_n(R_n) \right\};  \\
	\beta\left( Q^n;  (\sigma^\star)^{\otimes n} \right) &\geq \frac{A(R,\nu)}{\sqrt{n}} \exp\left\{  -n R_n \right\} ,
	\end{align}
	where
	\begin{equation}
	A (R,\nu) := \frac{ \mathrm{e}^{-K_{\max} (R,\nu)} }{ 4\sqrt{ 2\pi V_{\max} (R,\nu) }    }.
	\end{equation}
	Choosing $x = - \log A(R,\nu) + \log c$ in the rate back-off term $\gamma_n =  \frac{\log n}{2n} + \frac{x}{n}$, we have
	\begin{align}
	\beta\left( Q^n;  (\sigma^\star)^{\otimes n} \right) &\geq c\exp\{-nR\}. \label{eq:lower_beta2}
	\end{align}
	Combining Eqs.~\eqref{eq:lower_alpha} and \eqref{eq:lower_beta2} then yields
	\begin{align}
	\widehat{\alpha}_{ c\exp\{ -n R \}  } \left( W_{\mathbf{x}^n}^{\otimes n}\| (\sigma^\star)^{\otimes n} \right) &\geq \frac{A(R,\nu)}{\sqrt{n}} \exp\left\{ -n \phi_n\left( R_n \right)  \right\} = \frac{A(R,\nu)}{\sqrt{n}} \exp\left\{ -n E_\text{sp}^{(2)}\left(  R - \gamma_n , P_{\mathbf{x}^n}\right)  \right\}. \label{eq:sharp20} 
	\end{align}	
	
	It remains to remove the rate back-off term $\gamma_n$ in Eq.~\eqref{eq:sharp20}. By Taylor's theorem, one has
	\begin{align} \label{eq:sharp21}
	E_\text{sp}^{(2)} \left( R - \gamma_n, P_{\mathbf{x}^n} \right)= E_\text{sp}^{(2)} \left( R, P_{\mathbf{x}^n} \right) - \gamma_n  \left.\frac{\partial E_\text{sp}^{(2)} \left( r, P_{\mathbf{x}^n}\right) }{\partial r} \right|_{r=R} + \frac{\gamma_n^2}{2}  \left.\frac{\partial^2 E_\text{sp}^{(2)} \left( r, P_{\mathbf{x}^n}\right) }{\partial r^2} \right|_{r=\bar{R}},
	\end{align}
	for some $\bar{R} \in (\underline{R}, R)$. Recalling item \ref{regularity-d} in Lemma \ref{lemm:regularity}, one can show that
	\begin{align} \label{eq:sharp23}
	\begin{split}
	&-   \left.\frac{\partial E_\text{sp}^{(2)} \left( r, P_{\mathbf{x}^n}\right) }{\partial r} \right|_{r=R}  = s_{R, P_{\mathbf{x}^n}   }^\star = \frac{ 1 - \alpha_{R, P_{\mathbf{x}^n}   }^\star }{ \alpha_{R, P_{\mathbf{x}^n}   }^\star } \in\mathbb{R}_{>0}, \\
	&\left.\frac{\partial^2 E_\text{sp}^{(2)} \left( r, P_{\mathbf{x}^n}\right) }{\partial r^2} \right|_{r=\bar{R}}= \frac{ (1+\bar{s})^3 }{  V_{ \alpha^\star_{\bar{R},P_{\mathbf{x}^n } } }(P_{\mathbf{x}^n},\mathscr{W}) } \leq \frac{  (1+s_0)^3 }{  V_{\min}( {R}, \nu  )   } =: \Upsilon \in \mathbb{R}_{>0},
	\end{split}
	\end{align}
	where 
	\begin{align}\bar{s} := - \left. \frac{\partial E_\text{sp}^{(2)}(r, P_{\mathbf{x}^n})}{ \partial r} \right|_{r=\bar{R}}  \leq - \left.\frac{\partial E_\text{sp}^{(2)}(r, P_{\mathbf{x}^n})}{ \partial r} \right|_{r=\underline{R}} =: s_0  	\in\mathbb{R}_{>0}
	\end{align}
	by the monotone decreases of $r\mapsto E_\text{sp}^{(2)}(r,P)$.
	Then, Eqs.~\eqref{eq:sharp20}, \eqref{eq:sharp21} and \eqref{eq:sharp23} lead to
	\begin{align}
	\widehat{\alpha}_{ c\exp\{ -n R \}  } \left( W_{\mathbf{x}^n}^{\otimes n}\| (\sigma^\star)^{\otimes n} \right) &\geq \frac{A(R,\nu)}{\sqrt{n}}  \exp\left\{ -n E_\text{sp}^{(2)}\left(  R, P_{\mathbf{x}^n}\right) - n \left[ \gamma_n   \left( s^\star_{R, P_{\mathbf{x}^n} } +\frac{\gamma_n}{2}\Upsilon \right)  \right] \right\} \\
	&= \frac{ A(R,\nu) }{ n^{\frac12\left( 1 + s^\star_{ R, P_{\mathbf{x}^n}}       \right)    }    } \exp\left\{ -n E_\textnormal{sp}^{(2)}\left(  R , P_{\mathbf{x}^n}\right) - \ell_n \right\}, \label{eq:sharp24}
	\end{align}
	where we denote by 
	\begin{align}
	\ell_n :=  -\left( s^\star_{R, P_{\mathbf{x}^n}} + \frac{\gamma_n}{2} \Upsilon \right){\log A(R,\nu)} + \frac{\gamma_n\Upsilon}{4}\log n.
	\end{align}
	Since $s^\star_{ R, P_{\mathbf{x}^n} } \in \mathbb{R}_{>0}$ and $\gamma_n \log n = o(1)$, we choose a constant $L\in\mathbb{R}_{>0}$ and $N_3\in\mathbb{N}$ such that
	\begin{align} \label{eq:sharp25}
	\ell_n \leq L, \quad \forall N\geq N_3.
	\end{align} 
	Hence, Eqs.~\eqref{eq:sharp24} and \eqref{eq:sharp25} lead to
	\begin{align} \label{eq:sharp26}
	\widehat{\alpha}_{ c\exp\{ -n R \}  } \left( W_{\mathbf{x}^n}^{\otimes n}\| (\sigma^\star)^{\otimes n} \right) 
	&= \frac{ A(R,\nu)\exp\{-L\} }{ n^{\frac12\left( 1 + s^\star_{ R, P_{\mathbf{x}^n}}     \right)      }    } \exp\left\{ -n E_\textnormal{sp}^{(2)}\left(  R , P_{\mathbf{x}^n}\right)  \right\}.
	\end{align}
	By letting $N_0 := \max\left\{N_1, N_2, N_3  \right\}$ and $A' := A(R,\nu) \exp\{-L\}$, we conclude the proof.	
\end{proof}
	
\subsection{Proofs of Theorem \ref{theo:refined} and Corollary \ref{coro:refined}} \label{ssec:Proof}
	
We are ready to prove our main result---the refined strong sphere-packing bound in Theorem \ref{theo:refined} for constant composition codes and Corollary \ref{coro:refined} for general codes.
	
\begin{proof}[Proof of Theorem \ref{theo:refined}]
Fix any rate $ C_{0,\mathscr{W}} < R <C_\mathscr{W}$. First note that by Ref.~\cite[Proposition 10]{HM16}, we find
\begin{align} \label{eq:EspR>0}
E_\text{sp}(R) \in \mathbb{R}_{>0}.
\end{align}
By Proposition~\ref{lemm:hypothesis} and the standard expurgation method (see e.g.~\cite[p.~96]{SGB67}, \cite[Theorem 20]{Bla74}, \cite[p.~395]{Bla87}), it holds for every constant composition code $\mathcal{C}_n$ with a common composition $P_{\mathbf{x}^n}$ that
\begin{align}
\overline{\epsilon} \left( \mathcal{C}_n \right) \geq \frac12 {\epsilon}_{\max} \left(  \mathcal{C}_n' \right)
&\geq \max_{  \sigma \in \mathcal{S(H)} } \frac12 \widehat{\alpha}_{ {1}/{|\mathcal{C}_n'|} } \left( W_{\mathbf{x}^n}^{\otimes n} \| \sigma^{\otimes n} \right)  \\
&\geq \max_{  \sigma \in \mathcal{S(H)} } \frac12 \widehat{\alpha}_{2\exp\{-nR\}} \left( W_{\mathbf{x}^n}^{\otimes n} \| \sigma^{\otimes n} \right)  \label{eq:sketch5} \\
&\geq \frac12 \widehat{\alpha}_{ 2\exp\{-nR\}} \left( W_{\mathbf{x}^n}^{\otimes n} \| (\sigma^\star)^{\otimes n} \right), \label{eq:sketch3}
\end{align}
where $\mathcal{C}_n'$ is an expurgated code with message size $|\mathcal{C}_n'| = \ceil{|\mathcal{C}_n|/2} \geq \frac12\exp\{nR\}$.
Inequality \eqref{eq:sketch5} holds because the map $\mu \mapsto \widehat{\alpha}_\mu$ is monotone decreasing.
In the last line \eqref{eq:sketch3} we denote by
$\sigma^\star$
a channel output state that depends on the coding rate $R$ and the composition $P_{\mathbf{x}^n}$, and $\sigma^\star$ will be chosen later.

In the following, we deal with sequences of inputs that will yield different lower bounds. Fix an arbitrary $\delta \in (0, E_\text{sp}(R))$. Let $\nu := E_\text{sp}(R) - \delta > 0$, and recall the definition in Eq.~\eqref{eq:PRnu}:
\begin{align}
\mathscr{P}_{R,\nu}(\mathcal{X}) := \left\{ P_{\mathbf{x}^n} \in\mathscr{P}(\mathcal{X}): \nu\leq E_\text{sp}^{(2)}(R,P_{\mathbf{x}^n}) \leq E_\text{sp}(R) < +\infty  \right\}.
\end{align}
The set $\mathscr{P}_{R,\nu}(\mathcal{X})$ ensures that the error exponents of the input sequences $\mathbf{x}^n$ with composition $P_{\mathbf{x}^n} \in \mathscr{P}_{R,\nu}(\mathcal{X})$ are close to the sphere-packing exponent $E_\text{sp}(R)$. 

For sequences $\mathbf{x}^n$ with $ P_{\mathbf{x}^n} \notin \mathscr{P}_{R,\nu} (\mathcal{X})$, we infer that
\begin{align}
E_\text{sp}(R) - E_\text{sp}^{(2)} \left(R, P_{\mathbf{x}^n}\right)  = \delta > 0.
\end{align}
We then apply the Chebyshev-type bound, Proposition \ref{prop:Chebyshev}, with $c=2$ to obtain, for some $\kappa\in\mathbb{R}_{>0}$ and $\forall P_{\mathbf{x}^n} \notin \mathscr{P}_{R,\nu}(\mathcal{X})$,
\begin{align} 
\widehat{\alpha}_{2\exp\{-nR\}} \left( W_{\mathbf{x}^n}^{\otimes n} \| (\sigma^\star)^{\otimes n}  \right) &\geq   \exp\left\{ -\kappa  \sqrt{n} -n E_\text{sp}^{(2)}\left(R, P_{\mathbf{x}^n}  \right)    \right\},  \label{eq:caseI1} \\
&\geq  \exp\left\{ -\kappa  \sqrt{n} -n \left[ E_\text{sp}\left(R \right) - \delta\right]     \right\}, \label{eq:caseI}
\end{align}	
for all sufficiently large $n$, say $n\geq N_1\in\mathbb{N}$. The equality in Eq.~\eqref{eq:caseI1} follows from the saddle-point property, item \ref{saddle-a} in Proposition \ref{prop:saddle}, and the constants $\kappa_1$, $\kappa_2$ are positive and finite constants.

Next, we consider sequences $\mathbf{x}^n$ with $P_{\mathbf{x}^n} \in \mathscr{P}_{R,\nu}(\mathcal{X})$. Since such sequences  satisfy  Eq.~\eqref{eq:sharp_cond2}, we apply the sharp lower bound, Proposition \ref{prop:sharp}, with $c=2$ to obtain, $\forall P_{\mathbf{x}^n} \in \mathscr{P}_{R,\nu}(\mathcal{X})$,
\begin{align}
\widehat{\alpha}_{2\exp\{-nR\}} \left( W_{\mathbf{x}^n}^{\otimes n} \| (\sigma^\star)^{\otimes n} \right) &\geq \frac{2A}{n^{\frac12\left( 1 + s^\star_{R, P_{\mathbf{x}^n}  }  \right)}} \exp\left\{ -n E_\text{sp}^{(2)}\left(R , P_{\mathbf{x}^n} \right)   \right\},  \label{eq:caseII1}
\end{align}
for all sufficiently large $n$, say $n\geq N_2\in\mathbb{N}$, and some $A\in\mathbb{R}_{>0}$. In the following, we will relate the term $s^\star_{R, P_{\mathbf{x}^n}}$ in Eq.~\eqref{eq:caseII1} to $\left|E_\text{sp}'(R) \right|$. 
The idea follows similar from \cite[Eqs.~(111)--(114)]{AW14}. Let 
\begin{align}
&\mathscr{P}_{R}^\star(\mathcal{X}) := \left\{  P\in\mathscr{P}(\mathcal{X}): 
E_\text{sp}^{(2)} (R,P) = E_\text{sp}(R)   \right\}, \\
&\mathscr{P}_{\theta} (\mathcal{X}):= \left\{ P\in\mathscr{P}_{R,\nu}(\mathcal{X}) : \min_{ Q \in \mathscr{P}_{R}^\star(\mathcal{X})   }  \left\| P - Q   \right\|_1   \geq \theta   \right\}.
\end{align}
Since $s^\star_{R, (\cdot)}$ is uniformly continuous on the compact set $P\in\mathscr{P}_{R,\nu}(\mathcal{X})$ (see item \ref{saddle-d} of Proposition \ref{prop:saddle}), one has
		\begin{align} \label{eq:Tay2}
		\forall \gamma\in\mathbb{R}_{>0}, \; 
		\exists f(\gamma) \in \mathbb{R}_{>0}, \; \text{such that} \;
		\forall P,Q\in\mathscr{P}_{R,\nu}(\mathcal{X}), \;
		\left\|  P-Q \right\|_1 < f(\gamma)
		\Rightarrow
		\left| s_{R,P}^\star - s_{R,Q}^\star \right| < \gamma.
		\end{align}
By choosing $\gamma \in \mathbb{R}_{>0}$ that satisfies Eq.~\eqref{eq:Tay2}, it follows that
\begin{align} \label{eq:caseII2}
s^\star_{R, P_{\mathbf{x}^n}} \leq \left|E_\text{sp}'(R) \right| + \gamma, \quad \forall P_{\mathbf{x}^n}\in \mathscr{P}_{R,\nu}(\mathcal{X}) \backslash \mathscr{P}_{f(\gamma)}(\mathcal{X}). 
\end{align}
Hence, Eqs.~\eqref{eq:caseII1} and \eqref{eq:caseII2} further lead to, $\forall P_{\mathbf{x}^n}\in \mathscr{P}_{R,\nu}(\mathcal{X}) \backslash \mathscr{P}_{f(\gamma)}(\mathcal{X})$,
\begin{align}
\widehat{\alpha}_{2\exp\{-nR\}} \left( W_{\mathbf{x}^n}^{\otimes n} \| (\sigma^\star)^{\otimes n} \right)  &\geq \frac{ 2A }{n^{\frac12\left( 1 + |E_\text{sp}'(R)|  + \gamma    \right)}} \exp\left\{ -n E_\text{sp}\left(R  \right)   \right\}. \label{eq:caseII3}
\end{align}
For the case $P_{\mathbf{x}^n} \in\mathscr{P}_{R,\nu}(\mathcal{X}) \cap \mathscr{P}_{f(\gamma)} (\mathcal{X})$, we have  
\begin{align} \label{eq:caseII4}
E_\text{sp}(R) - \max_{ P \in  \mathscr{P}_{f(\gamma)}(\mathcal{X})    } E_\text{sp}^{(2)} (R,P_{\mathbf{x}^n}) =: \delta' > 0.
\end{align}
Then, Eqs.~\eqref{eq:caseII1} and \eqref{eq:caseII4} give, $\forall P_{\mathbf{x}^n}\in \mathscr{P}_{R,\nu}(\mathcal{X}) \cap \mathscr{P}_{f(\gamma)}(\mathcal{X})$,
\begin{align}
\widehat{\alpha}_{2\exp\{-nR\}} \left( W_{\mathbf{x}^n}^{\otimes n} \| (\sigma^\star)^{\otimes n} \right) 
&\geq \frac{ 2A }{n^{\frac12\left( 1 + s^\star_{R, P_{\mathbf{x}^n} }    \right)}} \exp\left\{ -n \left[ E_\text{sp}\left(R  \right) - \delta' \right]   \right\}. \label{eq:caseII5}
\end{align}
		
Finally, by comparing the bounds in Eqs.~\eqref{eq:caseI}, \eqref{eq:caseII3} and \eqref{eq:caseII5}, the first-order leading term in the right-hand side of Eq.~\eqref{eq:caseII3} decays faster than that of Eqs.~\eqref{eq:caseI} and \eqref{eq:caseII5}. Thus, for sufficiently large $n$, say $n\geq N_3\in\mathbb{N}$, we combine the bounds to obtain, for all compositions $P_{\mathbf{x}^n} \in \mathscr{P}(\mathcal{X})$,
\begin{align}
\widehat{\alpha}_{2\exp\{-nR\}} \left( W_{\mathbf{x}^n}^{\otimes n} \| (\sigma^\star)^{\otimes n} \right) &\geq \frac{ 2A }{n^{\frac12\left( 1 + |E_\text{sp}'(R)| + \gamma    \right)}} \exp\left\{ -n E_\text{sp}\left(R  \right)   \right\}. \label{eq:sketch_lb}
\end{align}
By combining Eqs.~\eqref{eq:sketch3}, \eqref{eq:sketch_lb}, we conclude our result: for any $\gamma>0$ and every $n$-blocklength constant composition code $\mathcal{C}_n$,
\begin{align}
\bar{\epsilon}\left( \mathcal{C}_n  \right) \geq \frac{ A }{n^{\frac12\left( 1 + |E_\text{sp}'(R)|  + \gamma   \right)}} \exp\left\{ -n E_\text{sp}\left(R  \right)   \right\},
\end{align}
for all sufficiently large $n\geq N_0:= \max\left\{N_1,N_2,N_3  \right\}$.
\end{proof}

\begin{proof}[Proof of Corollary \ref{coro:refined}]
For an $n$-blocklength code, there are at most $\binom{n+|\mathcal{X}|-1}{|\mathcal{X}|-1} <  n^{|\mathcal{X}|}$ different compositions. Hence, for any code with $M = \exp\{ nR \}$ codewords, there exists some codewords $M'$ of the same composition such that $M' \geq M/n^{|\mathcal{X}|}$. Denote by $\mathcal{C}_n'$ such constant composition codes with composition $P_{\mathbf{x}^n}$.	
		
Fix an arbitrary $\underline{R} \in (C_{0,\mathscr{W}}, R)$, and choose $N_1$ be an integer such that $R - \frac{|\mathcal{X}|}{n} \log n \geq \underline{R}$ for all $n\geq N_1$. Consider such $n\geq N_1$ onwards. By following the similar steps in  Theorem \ref{theo:refined}, we obtain
\begin{align}
\epsilon^*\left(n, R\right) \geq \bar{\epsilon}\left(\mathcal{C}_n' \right) &\geq  \frac{ A }{n^{\frac12\left( 1 + s^\star_{R, P_\mathbf{x}^n }    \right)}}  \exp\left\{	-n E_\textnormal{sp}^{(2)}\left(R - \frac{|\mathcal{X}|}{n} \log n, P_{\mathbf{x}^n} \right)\right\}, \label{eq:coro1}
\end{align}
for all sufficiently large $n$, say $n\geq N_2 \in \mathbb{N}$, and some $s^\star_{R,P_{\mathbf{x}^n}} \in \mathbb{R}_{>0}$.	Let
\begin{align}
\Upsilon := \max_{ P \in \mathscr{P}(\mathcal{X}): E_\text{sp}^{(2)}(\bar{R},P) = E_\text{sp}(\bar{R})  } \left| \left.\frac{ \partial  E_\text{sp}^{(2)}(r, P) }{\partial r}\right|_{r = \underline{R}} \right|.
\end{align}
Then, item \ref{Esp-a} in Proposition \ref{prop:Esp} implies that
\begin{align}
E_\textnormal{sp}^{(2)}\left(R - \frac{|\mathcal{X}|}{n} \log n , P_{\mathbf{x}^n} \right) &\leq E_\text{sp}^{(2)}(R, P_{\mathbf{x}^n} ) + \Upsilon\cdot\frac{|\mathcal{X}|}{n} \log n \\
&\leq  E_\text{sp}(R ) + \Upsilon\cdot\frac{|\mathcal{X}|}{n} \log n, \quad \forall n\geq N_2 \label{eq:coro2}
\end{align}
Combining Eqs.~\eqref{eq:coro1} and \eqref{eq:coro2} gives
\begin{align}
\epsilon^*\left(n, R\right) \geq \frac{ A }{{n}^{\frac12\left( 1 + s^\star_{R, P_{\mathbf{x}^n} }    \right) + \Upsilon|\mathcal{X}|} }  \exp\left\{-n E_\textnormal{sp}(R)\right\}, \quad \forall n\geq \max\{N_1, N_2\}.
\end{align}
By choosing $t \in \mathbb{R}_{>0}$ such that $ n^{-t} \leq A n^{-\frac12\left( 1 + s^\star_{R, P_{\mathbf{x}^n}} \right) - \Upsilon|\mathcal{X}|} $, and letting $N_0 := \max\{N_1, N_2\}$, we conclude our claim.
\end{proof}
	
\section{Symmetric Classical-Quantum Channels} \label{sec:symm}
In this section, we consider a symmetric c-q channels. By using the symmetric property of the channels, we show that the uniform distribution, denoted by $U_{\mathcal{X}}$, achieves the maximum of $E_\textnormal{sp}^{(1)}(R,\cdot)$ and $E^{(2)}_\textnormal{sp}(R,\cdot)$. Then, by choosing the optimal output state $\sigma_{R}^\star = \sigma^\star_{R,U_\mathcal{X}}$, every input sequence in the codebook is a good codeword and attains the sphere-packing exponent $E_\text{sp}(R)$. Hence, we can remove the assumption of constant composition codes and apply Theorem \ref{theo:refined} in Section \ref{sec:main} to obtain the optimal prefactor for the sphere-packing bound (Theorem \ref{theo:symm}).

A c-q channel $\mathscr{W}:\mathcal{X}\to \mathcal{S(H)}$ is \emph{symmetric} if it satisfies
\begin{align} \label{eq:sym}
W_x := V^{x-1} W_1 ( V^{\dagger} )^{{x-1}}, \quad \forall x\in\mathcal{X},
\end{align}
where $W_1\in\mathcal{S(H)}$ is an arbitrary density operator, and $V$ is a unitary operator on $\mathcal{S(H)}$ that satisfies $V^{\dagger} V = V V^{\dagger} = V^{|\mathcal{X}|} = \mathds{1}_\mathcal{H}$.
		
\begin{theo}[Exact Sphere-packing Bound for Symmetric Classical-Quantum Channels] \label{theo:symm}
For any rate $R\in(C_{0,\mathscr{W}}, C_\mathscr{W})$, there exist $A>0$ and $N_0\in\mathbb{N}$ such that for all codes $\mathcal{C}_n$ of length $n\geq N_0$ with message size $|\mathcal{C}_n|\geq \exp\{nR\}$, we have
\begin{align} \label{eq:symm}
{\epsilon}_{\max} \left(\mathcal{C}_n \right) \geq  \frac{A}{ n^{\frac12 \left( 1+\left| E_\textnormal{sp}'(R)\right| \right) } }   \exp\left\{-n E_\textnormal{sp}(R)
\right\}.
\end{align}
\end{theo}
	
\begin{proof}
The proof consists of the following steps.
First, we show that the distribution ${U}_{\mathcal{X}}$ satisfies $E_\text{sp}^{(1)}(R,{U}_{\mathcal{X}}) = E_\text{sp}^{(2)}(R,{U}_{\mathcal{X}}) = E_\text{sp}(R)$.
Second, we show that $E_\text{sp}^{(2)}(R,P) = E_\text{sp}(R)$ for all $P\in\mathscr{P}(\mathcal{X})$, which means that any codeword attains the sphere-packing exponent.
Finally, we follow Theorem \ref{theo:refined} to complete the proof.
	
Fix any $R\in(C_{0,\mathscr{W}}, C_\mathscr{W})$. From the definition of the symmetric channels in Eq.~\eqref{eq:sym}, it is not hard to verify that $U_{\mathcal{X}} \mathscr{W}^{\alpha} = V U_{\mathcal{X}} \mathscr{W}^\alpha V^\dagger$ for all $\alpha\in (0,1]$, where we denote by $P\mathscr{W}^\alpha := \sum_{x\in\mathcal{X}} P(x) W_x^\alpha$ for all $\alpha\in(0,1]$. Hence,  it follows that 
\begin{align}
\Tr[ W_x^{\alpha}  ( U_{\mathcal{X}} \mathscr{W}^{\alpha} )^{\frac{1-\alpha}{\alpha}} ] 
&= \Tr[ V^{x-1} W_1^{\alpha} V^{\dagger\, x-1}  ( U_{\mathcal{X}} \mathscr{W}^{\alpha} )^{\frac{1-\alpha}{\alpha}} ]  \\
&= \Tr[  W_1^{\alpha}  ( U_{\mathcal{X}} \mathscr{W}^{\alpha} )^{\frac{1-\alpha}{\alpha}} ]  \label{eq:sym1}
\end{align} 
for all $x\in\mathcal{X}$ and $\alpha\in (0,1]$. 
Summing Eq.~\eqref{eq:sym1} over all $x\in\mathcal{X}$ and dividing by $M$ yields that
\begin{align}
\Tr[ W_x^{\alpha}  ( U_{\mathcal{X}} \mathscr{W}^{\alpha} )^{\frac{1-\alpha}{\alpha}} ] 
= \Tr[ ( U_{\mathcal{X}} \mathscr{W}^{\alpha} )^{\frac{1}{\alpha}}], \label{eq:sym2}
\end{align}
for all $x\in\mathcal{X}$ and $\alpha\in (0,1]$. 
Recalling Proposition \ref{prop:Hol} below, the above equation shows that  	the distribution $U_\mathcal{X}$ indeed maximizes $E_0(s,P)$, $\forall s\in\mathbb{R}_{\geq 0}$. Then we have
\begin{align*}
E_\text{sp}^{(1)}(R, U_\mathcal{X}) = \sup_{s\geq 0} \left\{ \max_{P\in\mathscr{P}(\mathcal{X})} E_0(s,P) -sR  \right\} = E_\text{sp}(R).
\end{align*}
Further, Jensen's inequality shows that $E_\text{sp}^{(2)} (R,U_\mathcal{X}) \geq E_\text{sp}^{(1)} (R,U_\mathcal{X}) = E_\text{sp}(R)$, and thus, $E_\text{sp}^{(2)} (R,U_\mathcal{X}) = E_\text{sp}(R)$.
		
Next, let $(\alpha_R^\star, \sigma_R^\star)$ be the saddle-point of $F_{R,U_\mathcal{X}}(\cdot,\cdot)$ (see Eq.~\eqref{eq:F}). One can observe from the definition of $E_\text{sp}^{(2)}$ and Eq.~\eqref{eq:sym2} that all the quantities $D_{\alpha_R^\star  }(W_x\|\sigma_R^\star)$, $x\in\mathcal{X}$, are equal. 	By Proposition \ref{prop:I2}-\ref{I2-Augustin_mean}, we obtain 
\begin{align} \label{eq:symm1}
\sigma_R^\star =\frac{ \left(U_\mathcal{X} \mathscr{W}^{\alpha_R^\star}\right)^{1/\alpha_R^\star} }{ \Tr \left[ \left(U_\mathcal{X} \mathscr{W}^{\alpha_R^\star}\right)^{1/\alpha_R^\star} \right] },
\end{align}
which, in turn, implies that
\begin{align} \label{eq:symm2}
E_\text{sp}^{(2)} ( R, P ) =\sup_{\alpha\in(0,1]} F_{R,P}(\alpha, \sigma_R^\star) = \sup_{s\geq 0} \left\{ E_0(s,U_\mathcal{X}) - sR  \right\} = E_\text{sp}(R), \quad \forall P\in\mathscr{P}(\mathcal{X}).
\end{align}
Further, we have 
\begin{align}\label{eq:symm3}
\left|E_\text{sp}'(R)\right| = \frac{1-\alpha_R^\star}{\alpha_R^\star} = \left| \frac{ \partial E_\text{sp}^{(2)}(R,P)  }{\partial R}  \right|, \quad \forall P\in\mathscr{P}(\mathcal{X}).
\end{align} 
Since Eqs.~\eqref{eq:symm2} and \eqref{eq:symm3} indicates that every input sequence attains the sphere-packing exponent, we apply the same arguments in the proof of Theorem \ref{theo:refined} to conclude this theorem.

\begin{prop}[{\cite[Eq.~(38)]{Hol00}}] \label{prop:Hol}
	Let $s\in \mathbb{R}_{\geq 0}$ be arbitrary. The necessary and sufficient condition for the distribution $P^\star$ to maximize $E_0(s,P)$ is
	\begin{align} \label{eq:max_p0}
	\Tr\left[ W_x^{1/(1+s)} \cdot \left( \sum_{\bar{x}\in\mathcal{X}} P^\star(\bar{x}) W_{\bar{x}}^{1/(1+s)} \right)^s \right]
	\geq \Tr\left[ \left( \sum_{\bar{x}\in\mathcal{X}} P^\star(\bar{x}) W_{\bar{x}}^{1/(1+s)} \right)^{1+s} \right], \;\forall x \in \mathcal{X}
	\end{align}
	with equality if $P^\star(x)\neq 0$.
\end{prop}
\end{proof}

\section{Conclusions} \label{sec:conclusion}
	
In this paper, we provided an exposition of sphere-packing bounds in classical and quantum channel coding. Unlike classical results, there are two different quantum sphere-packing exponents, one being stronger than the other. We provided variational representations for these two exponents, and showed that they are ordered by the Golden-Thompson inequality. Our proof strategy was inspired by Blahut's approach of hypothesis testing reduction \cite{Bla74} and Altu{\u{g}}-Wagner's technique in strong large deviation theory \cite{AW14}. Specifically, the prefactor of the bound, that is akin to the converse Hoeffding bound in quantum hypothesis testing, can be improved by Bahadur-Ranga Rao's sharp concentration inequality \cite{BR60,DZ98}. Consequently, we obtained a refined strong sphere-packing bound for c-q channels and constant composition codes with a polynomial prefactor $f(n) = n^{-\frac12\left(1+|E_\text{sp}'(R)|+o(1)\right)}$. Moreover, the established result matches the best known random coding bound (i.e.~achievability) up to the logarithmic order \cite{AW14, SMF14, Sca14, Hon15}. For the case of general codes, the derived prefactor is of the polynomial order, i.e.~$f( n ) = O(n^{-t})$ for some $t> 1/2$. We are able to obtain the exact prefactor without the assumption of constant composition codes for a class of symmetric c-q channels. We note that the exact prefactor for general codes is still open even in the classical case. 	Finally, our refinement enables a moderate deviation analysis in c-q channels \cite{CH17} (see also \cite{CCT+16b}).

\section*{Acknowledgements}
M.-H.~Hsieh was supported by an ARC Future Fellowship under Grant FT140100574 and by US Army Research Office for Basic Scientific Research Grant W911NF-17-1-0401.
H.-C.~Cheng was supported by Ministry of Science and Technology Overseas Project for Post Graduate Research (Taiwan) under Grant 105-2917-I-002-028, 108-2917-I-564-042, and 104-2221-E-002-072.
The author would like to thank two reviewers and the editor for their useful comments.
The author would like to thank Bar\i\c{s} Nakibo\u{g}lu for the suggestion on the manuscript.
	
\appendix
	
\section{Lengendre-Fenchel Transform and Error-Exponent Functions} \label{app:LF}

In this section, we will see that the Lengendre-Fenchel transform is closely related to the error-exponent function of hypothesis testing and channel coding. 
Let $Z$ be a random variable, and denote by $\Lambda(\alpha) := \log \mathbb{E}[\e^{(1-\alpha)Z}]$ the cumulant-genearating function. We call the following quantity the \emph{Lengendre-Fenchel transform} of $\Lambda(\alpha)$:
\begin{align}
\Lambda^*(z) :=
\sup_{\alpha\in\mathbb{R}} \left\{z(1-\alpha) - \Lambda(\alpha)
\right\}.
\end{align}
This quantity arises in the exponent of various large deviation concentration inequalities \cite{DZ98}.
By choosing $Z$ to be the log-likelihood of two distributions as will be shown later, the  Lengendre-Fenchel transform will determine how fast the errors in a hypothesis testing exponentially decay.

Consider the following binary hypotheses:
\begin{align}
\begin{split} \label{eq:pn_qn}
&\mathsf{H}_0: p^n := p_{x_1}\otimes p_{x_2}\otimes \cdots p_{x_n}, \\
&\mathsf{H}_1: q^n := q_{x_1}\otimes q_{x_2}\otimes \cdots q_{x_n},
\end{split}
\end{align}
where $p_{x_i}, q_{x_i}$ for $x_i\in\mathcal{X}$, $i\in[n]$ are probability mass functions with $p^n\ll q^n$.  Given any $r\geq 0$, recall the definition of the error-exponent function in Eq.~(\ref{eq:exponent_HT}):
\begin{align}
\phi_n(r) = \phi_n(r|p^n\|q^n) = 
\sup_{\alpha\in(0,1]} \left\{ \frac{1-\alpha}{\alpha}\left(  \frac1n D_\alpha\left( p^n\| q^n  \right) - r \right)  \right\}. \label{eq:phi_n}
\end{align}
Without loss of generality, we set zero all elements of $q_{x_i}$ that do not lie in the support of $p_{x_i}$, i.e.~ $q_{x_i}(\omega) = 0$, $\omega \not\in \texttt{supp}(p_{x_i})$, $i\in[n]$, because those elements do not contribute in $\phi_n(r)$.



Let random variable $Z_0 = \log \frac{q^n}{p^n} $ with probability distribution $p^n$, and let $Z_1 = \log \frac{p^n}{q^n} $ with probability distribution $q^n$. We denote their cumulant generating fucntions by
\begin{align}
\begin{split}
\Lambda_{0,n}(\alpha) &:= \frac{1}{n}\log \mathbb{E}_{p^n}\left[ \mathrm{e}^{ (1-\alpha)Z_0} \right] = \frac1n \sum_{i\in[n]}  	\Lambda_{0,{i}} (\alpha), \\  
\Lambda_{1,n} (\alpha) &:=\frac{1}{n} \log \mathbb{E}_{q^n}\left[ \mathrm{e}^{ (1-\alpha) Z_1} \right] = \frac1n \sum_{i\in[i]} 	\Lambda_{1,{i}} (\alpha);
\end{split}
\label{eq:zero_derivative}	
\end{align}
where 
\begin{align}
\Lambda_{0,{i}} (\alpha) &:= \log \mathbb{E}_{ {p}_{i}} \left[ \mathrm{e}^{ (1-\alpha) \log \frac{ {q}_{i}}{ {p}_{i}} } \right],\quad 
\Lambda_{1,{i}} (\alpha) := \log \mathbb{E}_{ {q}_{i}} \left[ \mathrm{e}^{ (1-\alpha) \log \frac{ {p}_{i}}{ {q}_{i}} } \right]. 
\end{align}
Rewrite the right-hand side of Eq.~\eqref{eq:phi_n} with $\alpha = \frac{1}{1+s}$, and observe that 
\begin{align}
\sum_{x\in\mathcal{X}} P_{\mathbf{x}^n} (x) sD_{\frac{1}{1+s}}(p_x\|q_x) \label{eq:E021} 
&= - (1+s) {\Lambda}_{0,P_{\mathbf{x}^n}} \left( \frac{1}{1+s} \right) \\
&=: E_0^{(2)}(s,P_{\mathbf{x}^n}). \label{eq:regularity8}
\end{align}
Then the error-exponent function in Eq.~\eqref{eq:phi_n} can also be viewed as a Lengendre-Fenchel transform of $E_0^{(2)}(s,P_{\mathbf{x}^n})$:
\begin{align}
\phi_n(r) &= \sup_{s\geq 0} \left\{ E_0^{(2)}(s,P_{\mathbf{x}^n})  -sr \right\}.\label{eq:objective}
\end{align}

The following lemma relates $\phi_n(r)$ to $\Lambda^*_{j,n}(z)$, the Lengendre-Fenchel transform of Eq.~(\ref{eq:zero_derivative}):
\begin{align}
\Lambda_{j}^* (z) := \sup_{\alpha\in\mathbb{R}} \left\{ z(1-\alpha) - \Lambda_{j,n}(\alpha)  \right\}, \quad j\in\{0,1\}. \label{eq:FL2}
\end{align}
Such a transform plays a significant role in  concentration inequalities, convex analysis, and large deviation theory \cite{DZ98}.

\begin{lemm} \label{lemm:regularity}
	Let $p^n$ and $q^n$  be described as above. 
	Assume $r> \frac1n D_0\left( p^n\|q^n \right)$ and $\phi_n(r) >0$.  
	The following hold:
	\begin{enumerate}[(a)]			
		\item\label{regularity-a} 
		$\Lambda_{0,n} '' (\alpha) >0$ for all $\alpha\in[0,1]$.
		\item\label{regularity-b} $\Lambda^*_{0,n} \left(  {\phi}_n({r}) - {r} \right) =  {\phi}_n(r)$.
		
		\item\label{regularity-c} $\Lambda^*_{1,n} \left( r -  {\phi}_n({r}) \right) = r$.
		
		\item\label{regularity-d} 
		The optimizer $\alpha^\star  \in (0,1)$ of $\Lambda^*_{0 }(z)$ in Eq.~\eqref{eq:FL2} is unique, and satisfies ${\Lambda}'_{0,P_{\mathbf{x}^n} }(\alpha^\star) =  {\phi}_n(r) - r$.  In particular, one has $\alpha^\star = \frac{s^\star }{1+s^\star }$ and
		\begin{align}
		s^\star &= - \frac{ \partial  {\phi}_n (r) }{\partial r}; \\ 
		\frac{ \partial^2  {\phi}_n (r) }{\partial r^2} &=  -\left( \left.\frac{\partial^2 E_0^{(2)}(s,P_{\mathbf{x}^n})   }{\partial s^2}\right|_{s=s^\star} \right)^{-1} = 	\frac{\left(1+s^\star \right)^3}{ \Lambda_{0, P_{\mathbf{x}^n} }'' (\alpha^\star)} > 0.
		\end{align}		
	\end{enumerate}
\end{lemm}

Before proving Lemma~\ref{lemm:regularity}, we will need the following partial derivatives with respect to $t$:
\begin{align}
\Lambda'_{0,{i}} (\alpha) = \mathbb{E}_{ {v}_{{i},\alpha} } \left[ \log \frac{ {p}_{i}}{ {q}_{i}}  \right], &\quad \Lambda'_{1,{i}} (\alpha)  = \mathbb{E}_{ {v}_{{i},1-\alpha}} \left[ \log \frac{ {q}_{i}}{ {p}_{i}}  \right]; \label{eq:first_derivative}\\
\Lambda''_{0,{i}} (\alpha) = \Var_{ {v}_{{i},\alpha}} \left[ \log \frac{ {p}_{i}}{ {q}_{i}}  \right], &\quad \Lambda''_{1,{i}} (\alpha) = \Var_{ {v}_{{i},\alpha}} \left[ \log \frac{ {q}_{i}}{ {p}_{i}}  \right], \label{eq:second_derivatie} 
\end{align}
where we denote the \emph{tilted distributions} for every $i\in[n]$ and $t\in[0,1]$ by
\begin{align}
{v}_{{i},\alpha}(\omega) := \frac{  {p}_{i}^\alpha(\omega)  {q}_{i}^{1-\alpha}(\omega) }{ \sum_{\bar{\omega} }  {p}_{i}^{\alpha}(\bar{\omega})  {q}_{i}^{1-\alpha}(\bar{\omega}) }, \quad \bar{\omega} \in  \texttt{supp}(p_i).
\end{align}
Since $p^n$ and $q^n$ share the same support, the above derivatives are well-defined. Further, it is not hard to verify that 
\begin{align} 
\Lambda_{0,i}(\alpha) &= \Lambda_{1,i} (1-\alpha), \quad
\Lambda_{0,i}'(\alpha) = -\Lambda_{1,i}' (1-\alpha),\quad
\Lambda_{0,i}''(\alpha) = \Lambda_{1,i}'' (1-\alpha) \label{eq:sym0}.
\end{align}

This lemma closely follows Ref.~\cite[Lemma 9]{AW14}; however, the major difference is that we prove the claim using $\phi_n(r|\rho^n\|\sigma^n)$ in Eq.~\eqref{eq:exponent_HT} instead of the discrimination function: $\min\left\{{D}\left(\tau\|\rho\right): {D}\left(\tau\|\sigma\right) \leq r\right\}$ in Eq.~\eqref{eq:Primal}. This expression is crucial to obtaining the sphere-packing bound in Theorem \ref{theo:refined} in the strong from, cf.~Eq.~\eqref{eq:sp}, instead of the weak form, cf.~Eq.~\eqref{eq:sp2}.


\begin{proof}[Proof of Lemma~\ref{lemm:regularity}-\ref{regularity-a}]
	We will prove this statement by contradiction. Let $\alpha\in[0,1]$, Assuming that
	${\Lambda}''_{0,n}(\alpha) =  0$, implies  $\Lambda''_{0,{i}} (\alpha) =0$, $ \forall i\in[n]$.
	Recall from Eq.~(\ref{eq:second_derivatie})
	\begin{equation}
	0 = \Lambda''_{0,{i}} (\alpha) = \Var_{v_{{i},\alpha}} \left[ \log \frac{ {p}_{i}}{ {q}_{i}}  \right],
	\end{equation}
	which is equivalent to 
	\begin{align} \label{eq:sharp3}
	{p}_i (\omega) =  {q}_i (\omega) \cdot \mathrm{e}^{ \Lambda'_{0,i}(\alpha)},  \quad \forall \omega \in \texttt{supp}(p_i).
	\end{align}
	Summing both sides of Eq.~\eqref{eq:sharp3} over $\omega\in \texttt{supp}(p_i)$ gives
	\begin{align} \label{eq:regularity12}
	1 = \Tr\left[ p_i^0 q_i \right]  \mathrm{e}^{ \Lambda'_{0,i}(\alpha)}.
	\end{align}
	Then, Eqs.~\eqref{eq:sharp3} and \eqref{eq:regularity12} imply that
	\begin{align}
	\phi_n ( r ) &= \sup_{0<\alpha\leq 1} \frac{\alpha-1}{\alpha} \left( r - \frac1n D_\alpha\left( p^n \|q^n \right) \right) \\
	&= \sup_{0<\alpha\leq 1} \frac{\alpha-1}{\alpha} \left(r+ \frac1n \sum_{i\in[n]}\log \Tr\left[ p_i^0 q_i \right]    \right) \\
	&= 0, \label{eq:sharp4}
	\end{align}
	where Eq.~\eqref{eq:sharp4} follows since  $r> \frac1n D_0(p^n\|q^n) $ by assumption. However, this contradicts with the assumption $\phi_n(r)>0$. Hence, we conclude item \ref{regularity-a}.
\end{proof}

\begin{proof}[Proof of Lemma~\ref{lemm:regularity}-\ref{regularity-b}]	
	In the following, we use the substitution $s = \frac{1-\alpha}{\alpha}$ for convenience.
	
	Observe that $ E_0^{(2)} (s, P_{\mathbf{x}^n}) - sr$ in Eq.~\eqref{eq:objective} is strictly concave in $s\in\mathbb{R}_{\geq 0}$ since
	\begin{align} 
	\frac{ \partial^2E_0^{(2)} (s, P_{\mathbf{x}^n}) }{\partial s^2} &= -\frac{1}{(1+s)^3} {\Lambda}''_{0,P_{\mathbf{x}^n}} \left( \frac{s}{1+s} \right) < 0, \label{eq:regularity5}
	\end{align}	
	owing to Eqs.~\eqref{eq:regularity8}, \eqref{eq:second_derivatie}, and item~\ref{regularity-a}.
	Moreover, $s=0$ cannot be an optimum in Eq.~\eqref{eq:objective}; otherwise, it will violate the assumption $\phi_n(r) \geq 0$. Thus a unique maximizer $s^\star  \in \mathbb{R}_{> 0}$ exists such that
	\begin{align}
	{\phi}_n (r) &= -s^\star r + E_0^{(2)} (s^\star, P_{\mathbf{x}^n})\label{eq:regularity14} \\
	&= \frac{1}{1+s^\star} {\Lambda}'_{0,P_{\mathbf{x}^n}}\left( \frac{1}{1+s^\star} \right) - {\Lambda}_{0,P_{\mathbf{x}^n}} \left( \frac{1}{1+s^\star} \right) . \label{eq:regularity9}
	\end{align}
	where in the second equality we use Eq.~(\ref{eq:regularity8}) and 
	\begin{align}
	r &= \left.\frac{ \partial E_0^{(2)} (s, P_{\mathbf{x}^n}) }{\partial s}\right|_{s = s^\star} \label{eq:regularity7}\\
	&= - \frac{1}{1+s^\star} {\Lambda}'_{0,P_{\mathbf{x}^n}}  \left( \frac{1}{1+s^\star} \right)-{\Lambda}_{0,P_{\mathbf{x}^n}} \left( \frac{1}{1+s^\star} \right) .\label{eq:regularity2} 
	\end{align}
	Comparing Eq.~\eqref{eq:regularity9} with \eqref{eq:regularity2} gives
	\begin{align}
	{\Lambda}'_{0,P_{\mathbf{x}^n}} \left( \frac{1}{1+s^\star} \right) = {\phi}_n(r) - r, \label{eq:regularity11}
	\end{align}
	which is exactly the optimum solution to $\Lambda_{0,P_{\mathbf{x}^n} }^*(z)$ in Eq.~\eqref{eq:FL2} with 
	\begin{align}
	\alpha^\star &= \frac{1}{1+s^\star}\in(0,1), \label{eq:opt_t}\\
	z &=  {\phi}_n(r) - r.
	\end{align}
	Hence, we obtain
	\begin{align}
	\Lambda_{0,P_{\mathbf{x}^n} }^*\left(  {\phi}_n(r) - r \right)&= \frac{1}{1+s^\star} z - {\Lambda}_{0,P_{\mathbf{x}^n}}\left( \frac{1}{1+s^\star} \right) \\
	&= \frac{1}{1+s^\star}\left(  {\phi}_n(r) - r \right) - {\Lambda}_{0,P_{\mathbf{x}^n} }\left( \frac{1}{1+s^\star} \right)\\
	&= \frac{s^\star}{1+s^\star} {\Lambda}'_{0,P_{\mathbf{x}^n} } \left( \frac{1}{1+s^\star} \right) - {\Lambda}_{0,P_{\mathbf{x}^n} }\left( \frac{s^\star}{1+s^\star} \right)\\
	&=  {\phi}_n(r),
	\end{align}
	where Eqs.~\eqref{eq:regularity11} and \eqref{eq:regularity9} are used in the third and last equalities.
\end{proof}	

\begin{proof}[Proof of Lemma~\ref{lemm:regularity}-\ref{regularity-c}]
	This proof follows from similar arguments in item \ref{regularity-b} and Eq.~\eqref{eq:sym0}.
	Eqs.~\eqref{eq:regularity11} and \eqref{eq:sym0} lead to
	\begin{align} \label{eq:regularity30}
	\Lambda_{1,P_{\mathbf{x}^n}}'\left( \frac{s^\star}{1+s^\star}  \right) = r - \phi_n(r),
	\end{align}
	which satisfies the optimum solution to $\Lambda_{1,P_{\mathbf{x}^n}}(z)$ in Eq.~\eqref{eq:FL2} with $\alpha^\star = \frac{s^\star}{1+s^\star} \in (0,1)$ and $z = r - \phi_n(r)$.
	Then, 
	\begin{align}
	\Lambda_{1,P_{\mathbf{x}^n} }^*\left(  r - {\phi}_n(r)  \right)&= \alpha^\star z - {\Lambda}_{1,P_{\mathbf{x}^n}}(\alpha^\star) \\
	&= \frac{s^\star}{1+s^\star}\left(  r - {\phi}_n(r) \right) - {\Lambda}_{1,P_{\mathbf{x}^n} }\left( \frac{s^\star}{1+s^\star} \right)\\
	&= \frac{1}{1+s^\star} {\Lambda}'_{1,P_{\mathbf{x}^n} } \left( \frac{s^\star}{1+s^\star} \right) - {\Lambda}_{1,P_{\mathbf{x}^n} }\left( \frac{s^\star}{1+s^\star} \right)\\
	&=  r,
	\end{align}
	where the third equality is due to Eq.~\eqref{eq:regularity30}, and the last equality follows from Eqs.~\eqref{eq:sym0} and \eqref{eq:regularity2}.
	
\end{proof}

\begin{proof}[Proof of Lemma~\ref{lemm:regularity}-\ref{regularity-d}]
	The fact that a unique optimizer $\alpha^\star \in(0,1)$ exists such that ${\Lambda}'_{0,P_{\mathbf{x}^n}} \left( \alpha^\star \right) = {\phi}_n(r) - r$ follows directly from Eqs.~\eqref{eq:regularity11}, \eqref{eq:opt_t} and ${\Lambda}''_{0,P_{\mathbf{x}^n} }(\alpha)>0$, for $\alpha\in[0,1]$.
	
	Moreover, Eqs.~\eqref{eq:regularity14}, \eqref{eq:regularity7}, and \eqref{eq:regularity5} yield
	\begin{align}
	-\frac{ \partial \phi_n(r) }{ \partial r}  & = s^\star, \label{eq:regularity23} \\
	\frac{ \partial^2 \phi_n(r) }{ \partial r^2} &= -\frac{\partial s^\star }{\partial r} = \left.- \left( \frac{\partial^2 E_0^{(2)} (s, P_{\mathbf{x}^n}) }{\partial s^2} \right)^{-1}\right|_{s = s^\star} = \frac{(1+s^\star)^3}{\Lambda_{0,P_{\mathbf{x}^n}}\left( \frac{1}{1+s^\star} \right)},
	\end{align}
	which completes the claim in item \ref{regularity-d}.		
\end{proof}

\section{A Tight Large Deviation Inequality} \label{app:tight}
Let $\left(Z_i\right)_{i=1}^n$ be a sequence of independent, real-valued random variables with probability measures $\left(\mu_i\right)_{i=1}^n$. Let $\Lambda_i(\alpha) := \log \mathbb{E}\left[\mathrm{e}^{(1-\alpha) Z_i}\right]$ and define the Legendre-Fenchel transform of $\frac1n\sum_{i=1}^n \Lambda_i(\cdot)$ to be:
\begin{align}
\Lambda_n^*(z) := \sup_{\alpha \in \mathbb{R}} \left\{ z(1-\alpha) - \frac1n \sum_{i=1}^n \Lambda_i(\alpha) \right\}, \quad \forall z\in\mathbb{R}.
\end{align}
The so-called \emph{large deviation inequality} means estimate the probability of sum of independent random variables that deviate from the mean by a linear amount, i.e.~$\Pr\left\{ \frac1n\sum_{i=1}^n Z_i \geq z   \right\}$. 
Cram\'er's theorem \cite[Chapter 2]{DZ98} states that 
\begin{align} \label{eq:Cramer}
\Pr\left\{ \frac1n\sum_{i=1}^n Z_i \geq z   \right\} = \mathrm{e}^{-n \Lambda_n^*(z) + o(n)},
\end{align}
where $\mathrm{e}^{ o(n) }$ is some subexponential prefactor.
The upper bound of Eq.~\eqref{eq:Cramer} can be simply shown by the exponent Chebyshev inequality:
\begin{align}
\Pr\left\{ \frac1n\sum_{i=1}^n Z_i \geq z   \right\} \leq \mathrm{e}^{-n \Lambda_n^*(z)},
\end{align}
while the lower bound is more involved.
The Bahadur-Ranga Rao's concentration inequality below then provides a sharp lower bound and show that the subexponential prefactor can be improved to $O(\frac{1}{\sqrt{n}})$.

Let $z\in\mathbb{R}$ be such that $\exists \alpha^\star\in (0,1)$ and
\begin{align}
&\Lambda_n^*(z) = z(1-\alpha^\star) - \frac1n\sum_{i=1}^n \Lambda_i (\alpha^\star). \label{eq:Rao1}
\end{align}
Define the probability measure $\tilde{\mu}_i$ via 
\begin{align}
\frac{\mathrm{d}\tilde{\mu}_i}{\mathrm{d}\mu_i} (z_i) := \mathrm{e}^{ z_i (1-\alpha^\star) - \Lambda_i(\alpha^\star)},
\end{align}
and let $\bar{Z}_i:= Z_i - \mathbb{E}_{\tilde{\mu}_i}\left[Z_i\right]$. Furthermore, define  $m_{2,n} := \frac1n\sum_{i=1}^n \Var_{\tilde{\mu}_i}\left[\bar{Z}_i\right]$, $m_{3,n} := \frac1n \sum_{i=1}^n \mathbb{E}_{\tilde{\mu}_i}\left[\left|\bar{Z}_i\right|^3\right]$,	and $K_n(\alpha^\star) := \frac{15\sqrt{2\pi}m_{3,n}}{m_{2,n}}$. With these definitions, we can now state the following sharp concentration inequality for $\frac1n\sum_{i=1}^n Z_i$:
\begin{theo}[Bahadur-Ranga Rao's Concentration Inequality {\cite[Proposition 5]{AW14}}, \cite{BR60}] \label{theo:Rao} 
	Provided Eq.~\eqref{eq:Rao1} holds,
	then
	\begin{align}
	\Pr\left\{ \frac1n\sum_{i=1}^n Z_i \geq z   \right\}
	\geq \mathrm{e}^{-n\Lambda_n^*(z)} \frac{ \mathrm{e}^{-K_n(\alpha^\star)} }{ 2\sqrt{2 n \pi m_{2,n} }  }\left( 1 - \frac{ 1+ (1+K_n(\alpha^\star))^2}{ 2 \sqrt{2 n m_{2,n}}} \right).
	\end{align}
\end{theo}

	\section{Uniform Continuity} \label{app:UC}
	The goal of this section is to present various uniform continuity properties, which play a significant role in proving finite blocklength converse bounds (see Propositions~\ref{prop:Chebyshev} and \ref{prop:sharp}).
	Let $r\in(C_{0,\mathscr{W}}, C_{1,\mathscr{W}})$ throughout this section. 
	For any $P\in\mathscr{P}_r(\mathcal{X}):= \left\{ P\in\mathscr{P}(\mathcal{X}): E_\text{sp}^{(2)}(r,P) >0 \right\}$, we denote by  $(\alpha_{r,P}^\star,\sigma_{r,P}^\star) \in (0,1) \times \mathcal{S(H)}$  the unique saddle-point of $F(r,P)$ (see Proposition~\ref{prop:saddle}-\ref{saddle-c}).
	For $P\notin\mathscr{P}_r(\mathcal{X})$, note that $(1,\sigma)$  is a saddle-point of $F(r,P)$ for all $\sigma\in\mathcal{S(H)}$. We thus choose $(1,P\mathscr{W})$ to be the saddle-point of $F(r,P)$ for $P\notin\mathscr{P}_r(\mathcal{X})$, subsequently.
	Define
	\begin{align}
	B_r(P,\mathscr{W}) &:= \sum_{x\in\mathcal{X}} P(x) \mathbb{E}_{ v_{x,\alpha_{r,P}^\star}} \left[  \log \frac{p_x}{q_x} \right]; \label{eq:UC_1}\\
	V_r(P,\mathscr{W}) &:= \sum_{x\in\mathcal{X}} P(x) \mathbb{E}_{ v_{x,\alpha_{r,P}^\star}} \left[  \left|\log \frac{p_x}{q_x} - \mathbb{E}_{ v_{ x,\alpha_{r,P}^\star }} \left[  \log \frac{p_x}{q_x} \right]\right|^2 \right]; 
	\label{eq:UC_2} \\
	T_r(P,\mathscr{W}) &:= \sum_{x\in\mathcal{X}} P(x) \mathbb{E}_{ v_{x,\alpha_{r,P}^\star}} \left[  \left|\log \frac{p_x}{q_x} - \mathbb{E}_{ v_{x,\alpha_{r,P}^\star}} \left[  \log \frac{p_x}{q_x} \right] \right|^3 \right], \label{eq:UC_3}
	\end{align}
	where $(p_x, q_x)$ is the Nussbaum-Szko{\l}a distribution \cite{NS09} of $(W_x, \sigma_{r,P}^\star)$, and the \emph{tilted distribution} is 
	\begin{align}
	v_{x,\alpha}(i,j) := \frac{ p_x^\alpha(i,j) q_x^{1-\alpha}(i,j) }{ \sum_{\imath,\jmath} p_x^\alpha(\imath,\jmath) q_x^{1-\alpha}(\imath,\jmath)  }, \quad \alpha\in[0,1].
	\end{align}
	
	Inspired by Ref.~\cite[Lemma 62]{PPV10}, we show the following continuity property, which are crucial for establishing the large deviation bounds in finite blocklength regime.
	\begin{prop}[Uniform Continuity] \label{prop:UC}
		Fix $R\in(C_{0,\mathscr{W}}, C_{1,\mathscr{W}})$.
		For every $\underline{R} \in (C_{0,\mathscr{W}}, R]$, $B_r(P,\mathscr{W})$, $V_r(P,\mathscr{W})$, and $T_r(P,\mathscr{W})$ are jointly continuous functions of $(r,P)$ on $[\underline{R},R]\times \mathscr{P}(\mathcal{X})$.
	\end{prop}
	\begin{remark}
		When $R \geq  I_1^{(2)}(P,\mathscr{W}) = I(P,\mathscr{W})$, Proposition~\ref{prop:Esp}-\ref{Esp-a} implies that $(\alpha_{R,P}^\star, \sigma_{R,P}^\star) = (1, P\mathscr{W})$.
		In this case, $V_R(P,\mathscr{W})$ equals to the \emph{information variance} $V(P)$ introduced by Tomamichel and Tan \cite[Appendix B.3]{TV15}, and so does $T_R(P,\mathscr{W}) = T(P)$.
		The established Proposition~\ref{prop:UC} covers the special case of the continuities of $V(P)$ and $T(P)$ in $P$, and provides a rigorous proof for \cite[Lemma 29]{TV15}.
		We emphasize that such a continuity property is a critical step to ensure that the third-order term in the asymptotic expansion of coding rate (see e.g.~\cite{PPV10}, \cite{TV15}) independent of all codeword sequences.
	\end{remark}
	
	\begin{proof}[Proof of Proposition~\ref{prop:UC}]
		Inspecting the definitions given in Eqs.~\eqref{eq:UC_1}, \eqref{eq:UC_2}, and \eqref{eq:UC_3}.
		it is not hard to see that the quantities $B_r(P,\mathscr{W})$, $V_r(P,\mathscr{W})$, and $T_r(P,\mathscr{W})$ are sums of finitely many terms. We thus prove that each term is a continuous function in $(r,P)$.
		In the following, we first show the continuity of $B_r(P,\mathscr{W})$. The proof for $V_r(P,\mathscr{W})$ and $T_r(P,\mathscr{W})$ follow similarly.

		Fix an arbitrary $x\in\mathcal{X}$ onwards.
		Let $(R_k, P_k)_{k\in\mathbb{N}}$ be an arbitrary sequence such that $(R_k, P_k) \in [\underline{R},R]\times \mathscr{P}(\mathcal{X})$, and $\lim_{k\to+\infty} (R_k, P_k) = (R_0, P_0) \in [\underline{R},R] \times \mathscr{P}(\mathcal{X})$. 
		To ease the burden of notation, we let
		\begin{align}
		\alpha_k := \alpha_{R_k,P_k}^\star, \text{ and }
		\sigma_k := \sigma_{R_k,P_k}^\star, \quad \forall k \in \mathbb{N}.
		\end{align}
		Note that the joint continuity proved in Proposition~\ref{prop:saddle}-\ref{saddle-d} guarantees that
		\begin{align}
		\begin{split} \label{eq:UC_4}
		&\lim_{k\to+\infty} \alpha_k = \alpha_{R_0, P_0}^\star =: \alpha_0, \\
		&\lim_{k\to+\infty} \sigma_k = \sigma_{R_0,  P_0}^\star =: \sigma_0.
		\end{split}
		\end{align}
		
		Given the eigenvalue decompositions $W_x = \sum_{i} \lambda_i |e_i\rangle\langle e_i |$ and $\sigma_k = \sum_{j} \mu_j(\sigma_k) |f_j^k\rangle\langle f_j^k |$, Nussbaum-Szko{\l}a distributions are $p_x(i,j) = \lambda_i |\langle e_i|f_j^k\rangle|^2$ and $q_x(i,j) = \mu_j(\sigma_{R_k,P_k}^\star) |\langle e_i|f_j^k\rangle|^2$. Here, we write $f_j^k$ and $\mu_j (\sigma_k)$ to emphasize the dependence on $P_k$.
		To prove the continuity of $B_r(P,\mathscr{W})$, it suffices to show
		\begin{align} 
		\begin{split} \label{eq:UC1}
		&P_k(x)  \frac{ \lambda_i^{\alpha_k} \mu_j^{1-\alpha_k}(\sigma_k)  |\langle e_i|f_j^k \rangle|^2 }{ \sum_{\imath,\jmath}\lambda_{\imath}^{\alpha_k} \mu_{\jmath}^{1-\alpha_k}(\sigma_k)  |\langle e_{\imath}|f_{\jmath}^k \rangle|^2 }   \log \frac{ \lambda_i }{ \mu_j (\sigma_k) }  \\
		\longrightarrow \;&P_0(x)  \frac{ \lambda_i^{\alpha_0} \mu_j^{1-\alpha_0}(\sigma_0)  |\langle e_i|f_j^0 \rangle|^2 }{ \sum_{\imath,\jmath} \lambda_{\imath}^{\alpha_0} \mu_{\jmath}^{1-\alpha_0}(\sigma_k)  |\langle e_{\imath}|f_{\jmath}^0 \rangle|^2 }   \log \frac{ \lambda_i }{ \mu_j (\sigma_0)}.
		\end{split}
		\end{align}
		
		In the following, we first exclude some trivial cases.	
		If $\lambda_i=0$, then the convergence is obvious (recalling that the power function is only acting on the support). We assume $\lambda_i>0$ onwards.
		If $P_k(x) > 0$ only for finite number of $k$, then the convergence in Eq.~\eqref{eq:UC1} is also trivial.
		We may assume $P_k(x) >0$ for all $k\in\mathbb{N}$ (switching to a subsequence if necessary).
		Further, Proposition~\ref{prop:saddle}-\ref{saddle-b} implies that $W_x \ll \sigma_k$ for all $k\in\mathbb{N}$.
		We have $\lambda_i |\langle e_i|f_j^k \rangle|^2=0$ whenever $\mu_j (\sigma_k) |\langle e_i|f_j^k \rangle|^2=0$ by the absolute continuity, which in turn implies the convergence of Eq.~\eqref{eq:UC1}.
		We may assume $\mu_j(\sigma_k)|\langle e_i|f_j^k \rangle|^2>0$ for all $k\in\mathbb{N}$.

		With the above assumptions, we study two cases: $P_0(x) = 0$ or not, separately.
		If $P_0(x) > 0$, then $W_x \ll \sigma_0$.
		The absolute continuity again implies that $\mu_j(\sigma_0)>0$ by the previous argument.
		Hence, we can deduce that $\mu_j (\sigma_k)$  is bounded away from zero. Using the continuity given in Eqs.~\eqref{eq:UC_4} and the continuity of logarithm,  $ \log \lambda_i/\mu_j (\sigma_k)$ tends to $\log \lambda_i/ \mu_j (\sigma_0)$, which shows the convergence in Eq.~\eqref{eq:UC1}.
		
		It remains to show the case of $P_0(x) = 0$. 
		Observe that the convergence in Eq.~\eqref{eq:UC1} holds when $\mu_j(\sigma_k) \not\rightarrow 0$.
		We thus consider the circumstance that $\mu_j(\sigma_k) \rightarrow 0$.
		To achieve our goal, we will show that the log-likelihood ratio $\log \frac{\lambda_i}{\mu_j(\sigma_k)}$ does not diverge too fast.
		
 		In what follows we inspect the eigenvalue $\mu_j(\sigma_k)$.
		The saddle-point property given in Proposition~\ref{prop:saddle}-\ref{saddle-a} and Proposition~\ref{prop:I2}-\ref{I2-Augustin_mean} indicate that $\sigma_k$ must satisfy
		\begin{align} \label{eq:UC2}
		\sigma_k = { \left( \sum_{\bar{x}\in\mathcal{X}} P_k( \bar{x} ) \frac{ W_{\bar{x}}^{\alpha_k} }{ \Tr\left[ W_{\bar{x}}^{\alpha_k} (\sigma_k)^{1-\alpha_k} \right] }  \right)^{\frac{1}{\alpha_k}} }.
		\end{align}
		Further, since $\alpha_{r,P}^\star \in (0,1]$ for all $(r,P)\in (C_{0,\mathscr{W}} , \infty]\times \mathscr{P}(\mathcal{X})$, the  continuity of $P\mapsto \alpha_{r,P}^\star$ given in Proposition~\ref{prop:saddle}-\ref{saddle-d} and the compactness of $[]\mathscr{P}(\mathcal{X})$  imply that 
		\begin{align}
		{\alpha}_{\underline{R}} := \min_{P\in\mathscr{P}(\mathcal{X})} \alpha_{\underline{R},P}^\star > 0 \label{eq:UC_5}
		\end{align}
		By the convexity of $r\mapsto E_\text{sp}^{(2)}(r,P)$ and  Proposition~\ref{prop:Esp}-\ref{Esp-c}, we have
		$
		\alpha_k \in [\alpha_{\underline{R}},1] 
		$ for all $k\in\{0\} \cup \mathbb{N}$.
		
		Therefore, 
		\begin{align}
		\mu_j (\sigma_k)
		& = \langle f_j^k |\sigma_k | f_j^k \rangle \\
		&\geq   \left(     \sum_{{\bar{x}}} P_k({\bar{x}}) \frac{  \langle f_j^k | W_{\bar{x}}^{\alpha_k}| f_j^k \rangle }{ \Tr\left[ W_{\bar{x}}^{\alpha_k} (\sigma_k)^{1-\alpha_k} \right] }    \right)^{\frac{1}{\alpha_k}} \label{eq:convex1} \\
		&\geq  \left(     P_k(x)   \frac{  \langle f_j^k | W_{{x}}^{\alpha_k}| f_j^k \rangle }{ \Tr\left[ W_{{x}}^{\alpha_k} (\sigma_k)^{1-\alpha_k} \right] }  \right)^{\frac{1}{\alpha_k}}  \\
		&= \left(     P_k(x) \frac{\sum_{\imath}  \lambda_{\imath}^{\alpha_k} |\langle e_{\imath}| f_j^k \rangle|^2}{  \sum_{\imath,\jmath}\lambda_{\imath}^{\alpha_k} \mu_{\jmath}^{1-\alpha_k}(\sigma_k)  |\langle e_{\imath}|f_{\jmath}^k \rangle|^2  }  \right)^{\frac{1}{\alpha_k}} \\
		&\geq \left(     P_k(x) \frac{  \lambda_{i}^{\alpha_k} |\langle e_{i}| f_j^k \rangle|^2}{  \sum_{\imath,\jmath}\lambda_{\imath}^{\alpha_k} \mu_{\jmath}^{1-\alpha_k}(\sigma_k)  |\langle e_{\imath}|f_{\jmath}^k \rangle|^2  }  \right)^{\frac{1}{\alpha_k}} \\
		&\geq \lambda_i  c_k^{\frac{1}{ {\alpha}_{\underline{R}} }  },
		\label{eq:UC5}
		\end{align}
		where inequality \eqref{eq:convex1} follows from Jensen's inequality\footnote{For every Hermitian matrix $B$ with eigenvalue decomposition $B = \sum_i \lambda_i |u_i\rangle\langle u_i|$, it follows that for every unit vector $|v\rangle$ and positive convex function $f$,
		$f(\langle v | B |v \rangle) \leq f( \sum_i \lambda_i |\langle v | u_i \rangle|^2 )
		\leq \sum_i f(\lambda_i) |\langle v|u_i	\rangle |^2 = \langle v | f(B) |v \rangle$.
}
		for the convex function $(\cdot)^{^{\frac{1}{\alpha_k}}}$ with $\frac{1}{\alpha_k} \in (0,1]$, and 
		we denote by $c_k:= P_k(x) \frac{  |\langle e_{i}| f_j^k \rangle|^2}{  \sum_{\imath,\jmath}\lambda_{\imath}^{\alpha_k} \mu_{\jmath}^{1-\alpha_k}(\sigma_k)  |\langle e_{\imath}|f_{\jmath}^k \rangle|^2 }$ in the last line.
		Note that $\mu_j (\sigma_k)\leq 1$. 
		Eq.~\eqref{eq:UC5} then implies
		\begin{align}
		\left| \log \frac{\lambda_i}{\mu_j (\sigma_k)} \right|
		&\leq \log \frac{1}{\lambda_i} - \log \left( \lambda_i  c_k^{\frac{1}{ {\alpha}_{\underline{R}} }  } \right) \\
		&= 2 \log \frac{1}{\lambda_i} - \frac{1}{ {\alpha}_{\underline{R}} } \log c_k  . \label{eq:UC6}
		\end{align}
		Since we assume $\mu_j(\sigma_k) \rightarrow 0$ and $\lambda_i>0$, Eq.~\eqref{eq:UC5} guarantees that 
		\begin{align}
		 c_k \to 0. \label{eq:UC_6}
		\end{align}

		Using Eqs.~\eqref{eq:UC_5},  \eqref{eq:UC6}, \eqref{eq:UC_6}, and the fact that $\lambda_i^{\alpha_k} \mu_j^{1-\alpha_k}(\sigma_k) \in [0,1]$ for all $k\in\mathbb{N}$, we are able show that the left-hand side of Eq.~\eqref{eq:UC1} converges to $0$:
		\begin{align}
		&P_k(x) \frac{ \lambda_i^{\alpha_k} \mu_j^{1-\alpha_k}(\sigma_k)  |\langle e_i|f_j^k \rangle|^2 }{ \sum_{\imath,\jmath}\lambda_{\imath}^{\alpha_k} \mu_{\jmath}^{1-\alpha_k}(\sigma_k)  |\langle e_{\imath}|f_{\jmath}^k \rangle|^2 } \left| \log \frac{ \lambda_i }{ \mu_j (\sigma_k) } \right| \\
		&\leq  c_k  \left| \log \frac{ \lambda_i }{ \mu_j (\sigma_k) } \right| \\
		&\leq 2 c_k \log \frac{1}{\lambda_i } - \frac{1}{\alpha_{\underline{R}}}  c_k \log  c_k  \\
		&\to 0,
		\end{align}
		which proves the joint continuity of $(r,P) \mapsto B_r(P,\mathscr{W})$.
		
		Next, we show the continuity of $V_r(P,\mathscr{W})$ and $T_r(P,\mathscr{W})$. 
		Denote by $B_r(W_x\|\sigma_k) := \mathbb{E}_{ v_{x, \alpha_k}} \left[  \log p_x/q_x \right] $ for convenience. For $P_0(x) >0$, $\mu_j (\sigma_k)$ is bounded away from zero. Then, $ \log \lambda_i/\mu_j (\sigma_k)$ tends to $\log \lambda_i/ \mu_j (\sigma_0)$, and it is not hard to see that $B_{R_k}(W_x\| \sigma_k) \to B_{R_k}(W_x\|\sigma_0)$.
		It suffices to prove the convergence when $P_k(x) \to 0$ and $\mu_j(\sigma_k) \to 0$ as mentioned before.
		Eq.~\eqref{eq:UC6} immediately implies that
		\begin{align}
		B_{R_k}(W_x\|\sigma_k) &= \sum_{i,j}
		\frac{ \lambda_i^{\alpha_k} \mu_j^{1-\alpha_k}(\sigma_k)  |\langle e_i|f_j^k \rangle|^2 }{ \sum_{\imath,\jmath}\lambda_{\imath}^{\alpha_k} \mu_{\jmath}^{1-\alpha_k}(\sigma_k)  |\langle e_{\imath}|f_{\jmath}^k \rangle|^2 } \log \frac{ \lambda_i }{ \mu_j (\sigma_k) } \\
		&\leq 	  2 \log \frac{1}{\lambda_i} - \frac{1}{\alpha_{\underline{R}}} \log c_k . \label{eq:UC9}
		\end{align}
		Using the inequality $|a+b|^2 \leq 2(|a|^2+|b|^2)$, we obtain
		\begin{align}
		&P_k(x)  \frac{ \lambda_i^{\alpha_k} \mu_j^{1-\alpha_k}(\sigma_k)  |\langle e_i|f_j^k \rangle|^2 }{ \sum_{\imath,\jmath}\lambda_{\imath}^{\alpha_k} \mu_{\jmath}^{1-\alpha_k}(\sigma_k)  |\langle e_{\imath}|f_{\jmath}^k \rangle|^2 } \left| \log \frac{ \lambda_i }{ \mu_j(\sigma_k)} - B_{R_k}(W_x\|\sigma_k)\right|^2  \notag \\
		&\leq 2 c_k \left| \log \frac{ \lambda_i }{ \mu_j(\sigma_k)} \right|^2  + 2 c_k B_{R_k}^2(W_x\|\sigma_k). 
		 \label{eq:UC10}
		\end{align}
		Combining Eqs.~\eqref{eq:UC6}, \eqref{eq:UC_6}, \eqref{eq:UC9}, and \eqref{eq:UC10}, we prove the continuity of $V_r(P,\mathscr{W})$.
		
		Similarly, using the inequality $|a+b|^3 \leq 4(|a|^3+|b|^3) $ gives
		\begin{align}
		&P_k(x) \frac{ \lambda_i^{\alpha_k} \mu_j^{1-\alpha_k}(\sigma_k)  |\langle e_i|f_j^k \rangle|^2 }{ \sum_{\imath,\jmath}\lambda_{\imath}^{\alpha_k} \mu_{\jmath}^{1-\alpha_k}(\sigma_k)  |\langle e_{\imath}|f_{\jmath}^k \rangle|^2 } \left| \log \frac{ \lambda_i }{ \mu_j(\sigma_k)} - B_{\alpha_k}(W_x\|\sigma_k)\right|^3  \notag \\
		&\leq 4 c_k\left| \log \frac{ \lambda_i }{ \mu_j(\sigma_k)} \right|^3 + 4 c_k B_{R_k}^3(W_x\|\sigma_k). \label{eq:UC11}
		\end{align}
		Further, Eq.~\eqref{eq:UC6} implies
		\begin{align}
		\left| \log \frac{ \lambda_i }{ \mu_j(\sigma_k)}\right|^3  
		\leq -4\log^3 \lambda_i - \frac{4}{\alpha_{\underline{R}}} \log^3 c_k . \label{eq:UC12}
		\end{align}
		Combining Eqs.~\eqref{eq:UC_6}, \eqref{eq:UC9}, \eqref{eq:UC11}, and \eqref{eq:UC12}, proves the continuity of $T_r(P,\mathscr{W})$.
	\end{proof}

\section{Proof of Proposition \ref{prop:I2}} \label{app:I2}
\begin{prop2}[Properties of order $\alpha$ Augustin Information and Radius]
	Given any classical-quantum channel $\mathscr{W}:\mathcal{X}\to \mathcal{S(H)}$ with $|\mathcal{X}|< \infty$, the following hold:
	\begin{enumerate}[(a)]
		\item\label{I2-mono_} For every $P\in\mathscr{P}(\mathcal{X})$, $\alpha \mapsto I_\alpha^{(2)}(P,\mathscr{W})$ is monotone increasing on $[0,1]$, and $I_\alpha^{(2)}(P,\mathscr{W}) \leq \log |\mathcal{X}|$ for all $\alpha\in[0,1]$.
		
		\item\label{I2-Augustin_mean_} For every $(\alpha,P)\in(0,1]\times \mathscr{P}(\mathcal{X})$, there exists a unique $\sigma_{\alpha,P} \in \mathcal{S(H)}$, termed Augustin mean, such that
		\begin{align}
		I_\alpha^{(2)}(P,\mathscr{W}) =  D_\alpha\left(\mathscr{W}\| \sigma_{\alpha,P} | P\right),
		\end{align}
		and
		\begin{align} \label{eq:fixed-point0}
		\mathsf{T}_{\alpha,P}(\sigma) = \sigma \text{ and }
		\sigma \gg P\mathscr{W}
		\text{ if and only if } \sigma = \sigma_{\alpha,P},
		\end{align}	
		where the map $\mathsf{T}_{\alpha,P}:\mathcal{S}_{P,\mathscr{W}}(\mathcal{H}) \to \mathcal{S(H)}$ is defined as
		\begin{align} 
		\mathsf{T}_{\alpha,P}(\sigma) = \sum_{x\in\mathcal{X}} P(x) \frac{\sigma^{\frac{1-\alpha}{2}} W_x^\alpha \sigma^{\frac{1-\alpha}{2}} }{\Tr\left[W_x^\alpha \sigma^{1-\alpha}\right]}.
		\end{align}		
		
		\item\label{I2-conc_P_} For every $\alpha\in[0,1]$, the map $P\mapsto I_\alpha^{(2)}(P,\mathscr{W})$ is concave on $\mathscr{P}(\mathcal{X})$.
		
		\item\label{I2-conc_alpha_} For every $P\in\mathscr{P}(\mathcal{X})$, $\alpha \mapsto \frac{1-\alpha}{\alpha} I_\alpha^{(2)}(P,\mathscr{W})$ is concave on $(0,1]$.
		
		\item\label{I2-cont_alpha_} For every $P\in\mathscr{P}(\mathcal{X})$, $\alpha \mapsto I_\alpha^{(2)}(P,\mathscr{W})$ is continuous on $[0,1]$.
		
		\item\label{I2-cont_equi_} The family of functions $\{ I_\alpha^{(2)}(P,\mathscr{W})  \}_{ \alpha\in[0,1]  }$ is uniformly equicontinuous in $P\in\mathscr{P}(\mathcal{X})$.
		Moreover,
		The map $(\alpha, P) \mapsto I_\alpha^{(2)}(P,\mathscr{W}) \text{ is jointly continuous on } [0,1]\times \mathscr{P}(\mathcal{X})$. 
		
		\item\label{I2-cont_mean_} The map $(\alpha, P) \mapsto \sigma_{\alpha,P} \text{ is jointly continuous on } (0,1]\times \mathscr{P}(\mathcal{X})$. 
		
		\item\label{I2-cont_C_}
		The map $\alpha \mapsto C_{\alpha,\mathscr{W}}$ is continuous and monotone increasing on $[0,1]$. 
	\end{enumerate}
\end{prop2}

\begin{proof}[Proof of Proposition \ref{prop:I2}-\ref{I2-mono_}]
Recalling the definition of $I_\alpha^{(2)}$ given in Eq.~\eqref{eq:I2}. The statement immediately follows from Lemma~\ref{lemma:chaotic}-\ref{Da_mono_cont} (see also \cite[Lemma 4.6]{MO14b}) because the minimization over $\sigma\in\mathcal{S(H)}$ preserves the monotonicity. Hence, we have $I_\alpha^{(2)}(P,\mathscr{W}) \leq I_1(P,\mathscr{W}) \leq \log |\mathcal{X}|$, where the last inequality follows from the well-known upper bound for the Holevo quantity (see e.g.~\cite[Chapter 12]{NC09}).
\end{proof}

\begin{proof}[Proof of Proposition \ref{prop:I2}-\ref{I2-Augustin_mean_}]	
We first note that the infimum in Eq.~\eqref{eq:I2} can be attained. This can be verified by the following argument.
Lemma~\ref{lemma:chaotic}-\ref{Da_second_convex} shows that $D_\alpha$ is lower semi-continuous in its second argument. Hence, the linear combination, i.e.~$D_\alpha\left(\mathscr{W}\|\sigma|P\right)$, is also lower semi-continuous on $\mathcal{S(H)}$. Further, $\mathcal{S(H)}$ is compact owing to the assumption of the finite-dimensional Hilbert space $\mathcal{H}$. Thus, the extreme value theorem \cite[Chapter 30\S12.2]{KF75} guarantees that the infimum can be attained

For $\alpha= 1$, it is well-known that (see e.g.~\cite{Hol73}) $\sigma_{1,P} = P\mathscr{W}$. Using the fact $P\mathscr{W} \gg W_x$ for all $x\in\texttt{supp}(P)$, the statements immediately follow.

We fix an arbitrary $(\alpha,P)\in(0,1)\times \mathscr{P}(\mathcal{X})$ subsequently. Without loss of generality, we may further assume
\begin{align}
\bigcup_{x\in\texttt{supp}(P)} \texttt{supp}(W_x) = \mathds{1}_{\mathcal{H}},
\end{align}
and hence $P\mathscr{W}$ has full support.
We first show that the minimizer $\sigma_{\alpha,P}$ has full support too. Second, we prove the fixed-point property Eq.~\eqref{eq:fixed-point}. Finally, we establish the uniqueness of $\sigma_{\alpha,P}$. We remark that the uniqueness has been proven by Dalai and Winter \cite[Appendix~D]{DW14b}. Here, we provide an alternative proof for the completeness.
Our approach follows closely from Hayashi and Tomamichel \cite[Appendix C]{HT14}.

Define
\begin{align} \label{eq:target}
\mathcal{M}_\alpha(\mathcal{H}):=\argmin_{\sigma\in\mathcal{S(H)}}  D_\alpha \left( \mathscr{W} \| \sigma |P \right)
= \argmax_{\sigma\in\mathcal{S(H)}} g_\alpha(\sigma)
= \argmax_{\sigma\in\mathcal{S}_{P,\mathscr{W}}(\mathcal{H})} g_\alpha(\sigma)
\end{align}	
where
\begin{align}
g_\alpha (\sigma) := \sum_{x\in\mathcal{X}} P(x) \log \Tr \left[ W_x^\alpha \sigma^{1-\alpha} \right].
\end{align}
To show that the optimizer of $g_\alpha(\cdot)$ has full support, we observe that the directional derivative on the boundary of $\mathcal{S(H)}$ where at least one eigenvalue is zero in a direction that increases its rank diverges to positive infinite. Namely, it suffices to show
\begin{align}
\lim_{t\to 0} \frac{g_\alpha((1-t)\sigma + t\sigma^\perp) - g_\alpha(\sigma)}{t} = +\infty, \label{eq:fix7}
\end{align} 
where $\sigma \in \mathcal{S}_{P,\mathscr{W}}(\mathcal{H})$ is some singular density operator, and $\sigma^\perp :=  \frac{(\mathds{1}_\mathcal{H} - \sigma)}{\Tr\left[ \mathds{1}_\mathcal{H} - \sigma	\right]}$.
For $x\in\texttt{supp}(P)$ with $W_x \ll \sigma$, we have $W_x \perp \sigma^\perp$. It is not hard to see that
\begin{align}
&\lim_{t\to0} P(x) \frac{  \log \Tr\left[ W_x^\alpha\left( (1-t)\sigma + t \sigma^\perp \right)^{1-\alpha} \right] - \log \Tr\left[ W_x^\alpha \sigma^{1-\alpha} \right] }{t} \\
&= \lim_{t\to0} P(x) \frac{  \log \Tr\left[ W_x^\alpha\left( (1-t)^{1-\alpha}\sigma^{1-\alpha} + t^{1-\alpha} (\sigma^\perp)^{1-\alpha} \right) \right] - \log \Tr\left[ W_x^\alpha \sigma^{1-\alpha} \right] }{t} \label{eq:fix1} \\
&= \lim_{t\to0} P(x)\frac{(1-\alpha)\log (1-t)}{t} \label{eq:fix2}	\\
&= \lim_{t\to0} P(x)\frac{-(1-\alpha)}{1-t} \label{eq:fix3} \\
&= -P(x)(1-\alpha) \\
&> -\infty \label{eq:fix4}
\end{align}
where Eq.~\eqref{eq:fix1} holds because $\sigma \perp \sigma^\perp$; Eq.~\eqref{eq:fix2} is due to $W_x \perp \sigma^\perp$; and Eq.~\eqref{eq:fix3} is owing to L'H\^{o}spital's rule.

On the other hand, since $\sigma$ is singular, there must be some  $x\in\texttt{supp}(P)$ such that $W_x \not\ll \sigma$. Hence, by denoting $c:= \frac{ \Tr\left[W_x^\alpha(\sigma^\perp)^{1-\alpha}\right] }{ \Tr\left[W_x^\alpha\sigma^{1-\alpha}\right] }>0$,
Eq.~\eqref{eq:fix1} leads to
\begin{align}
&\lim_{t\to0} P(x)\frac{\log \left\{(1-t)^{1-\alpha} + t^{1-\alpha} c \right\} }{t} \\
&= \lim_{t\to0} P(x) \frac{ -(1-\alpha)(1-t)^{-\alpha} + (1-\alpha)t^{-\alpha} c }{(1-t)^{1-\alpha} + t^{1-\alpha} c} \label{eq:fix5}\\
&= +\infty, \label{eq:fix6}
\end{align}
where Eq.~\eqref{eq:fix5} is by L'H\^{o}spital's rule again.
Combining Eqs.~\eqref{eq:fix4} and \eqref{eq:fix6} concludes Eq.~\eqref{eq:fix7}.

Next, we show the fixed-point property: $\mathcal{M}_\alpha(\mathcal{H}) = \mathcal{F}_\alpha(\mathcal{H})$, where $\mathcal{F}_\alpha(\mathcal{H}) := \{\sigma\in\mathcal{S}_{>0}(\mathcal{H}) \}$ denotes the fixed-points of the map: $\mathsf{T}_{\alpha,P}:\mathcal{S}_{P,\mathscr{W}}(\mathcal{H}) \to \mathcal{S(H)}$. A necessary and sufficient condition for $\sigma$ to be an optimizer is 
\begin{align} \label{eq:saddle30}
\partial_\omega g_\alpha (\sigma) :=  
\mathsf{D} g_\alpha(\sigma) [ \omega - \sigma] =	0,
\end{align}
for all $\omega \in\mathcal{S(H)}$, where $\mathsf{D}g_\alpha(\sigma)$ denotes the Fr\'echet derivative of the map $g_\alpha$ (see e.g.~\cite[Appendix C]{HT14}).
Using the chain rule of Fr\'echet derivatives, it follows
\begin{align} 
\partial_\omega g_\alpha ( {\sigma})
&= \Tr\left[  \sum_{x\in\mathcal{X}} P(x) \frac{W_x^\alpha}{\Tr\left[ W_x^\alpha {\sigma}^{1-\alpha} \right] } \partial_{\omega} {\sigma}^{1-\alpha}\right] \label{eq:saddle10} \\
&= \Tr\left[  \sum_{x\in\mathcal{X}} P(x) \frac{ \sigma^{\frac{-\alpha}{2}} W_x^\alpha\sigma^{\frac{-\alpha}{2}}}{\Tr\left[  W_x^\alpha {\sigma}^{1-\alpha} \right] }  \sigma^{\frac{\alpha}{2}} \partial_{\omega} {\sigma}^{1-\alpha}  \sigma^{\frac{\alpha}{2}}\right]. 
\end{align}
We claim that the operators
\begin{align}
\left\{ \Delta_\omega = \sigma^{\frac{\alpha}{2}} \partial_{\omega} {\sigma}^{1-\alpha}  \sigma^{\frac{\alpha}{2}}: \omega \in \mathcal{S(H)}  \right\}
\end{align}
span the space of traceless Hermitian operators on $\mathcal{S(H)}$.
Let $\sigma = \sum_{i} \lambda_i |i\rangle\langle i|$ with $\lambda_i>0$ be the eigenvalue decomposition. One can verify \cite[Theorem 3.25]{HP14} that
\begin{align}
\langle i | \Delta_\omega | j \rangle = \begin{dcases}
(\lambda_i \lambda_j)^{\frac{\alpha}{2}} \frac{\lambda_i^{1-\alpha} - \lambda_j^{1-\alpha} }{ \lambda_i - \lambda_j } \langle i| \omega - \sigma |j\rangle, & \text{if } \lambda_i \neq \lambda_j \\
(1-\alpha) \langle i| \omega - \sigma |j\rangle, & \text{if } \lambda_i = \lambda_j
\end{dcases}.
\end{align}
Therefore, $\Delta_\omega$ is Hermitian and $\Tr\left[\Delta_\omega\right] = 0$ for all $\omega \in \mathcal{S(H)}$.
Moreover, the basis of the traceless Hermitian operators is given by the operators
\begin{align}
\left\{ \Gamma_{ij} = |i\rangle\langle j | + |j\rangle\langle i|, \;
\Gamma_{ij}' = \mathrm{i}|i\rangle\langle j | - |j\rangle\langle i|, \;
\Gamma_{ij}'' = |i\rangle\langle i | - |j\rangle\langle j|
\right\}_{i\neq j} .
\end{align}
For every tuple $(i,j)$ with $i\neq j$ there exists an $\eps > 0$ such that the state $\omega = \sigma + \eps \Gamma_{ij}$ is still in $\mathcal{S(H)}$. For this state, we find that $\Delta_\omega = \eta \Gamma_{ij}$ for some real $\eta>0$. The similar argument applies to $\Gamma_{ij}'$ and $\Gamma_{ij}''$. Hence, we have verified that the operators $\{\Delta_\omega\}_{\omega\in\mathcal{S(H)}}$ span the space of traceless Hermitian operators.

Armed with the above discussion, the condition that $\partial_\omega g_\alpha(\sigma) = 0$ for all $\omega\in\mathcal{S(H)}$ is equivalent to the condition that the operators
\begin{align}
\sum_{x\in\mathcal{X}} P(x) \frac{ \sigma^{\frac{-\alpha}{2}} W_x^\alpha\sigma^{\frac{-\alpha}{2}}}{\Tr\left[  W_x^\alpha {\sigma}^{1-\alpha} \right] }
\end{align}
must be proportional to the identity. Thus, the optimum must be a fixed point of the map $\mathsf{T}_{\alpha,P}(\cdot)$.

Lastly, to prove the uniqueness of the optimizer, it remains to show 
$\partial_\omega^2 g_\alpha(\sigma) : \mathsf{D}^2 g_\alpha(\sigma)[ \omega- \sigma, \omega - \sigma] < 0$ for all $\omega \neq \sigma $ and $\sigma>0$.
Continuing on Eq.~\eqref{eq:saddle10}, we have
\begin{align}
\partial_\omega^2 g_\alpha(\sigma)
&= - \Tr\left[  \sum_{x\in\mathcal{X}} P(x) \frac{W_x^\alpha}{\Tr^2\left[ W_x^\alpha {\sigma}^{1-\alpha} \right] } \partial_{\omega} {\sigma}^{1-\alpha}\right] + \Tr\left[  \sum_{x\in\mathcal{X}} P(x) \frac{W_x^\alpha}{\Tr\left[ W_x^\alpha {\sigma}^{1-\alpha} \right] } \partial_{\omega}^2{\sigma}^{1-\alpha}\right] \\
&< \Tr\left[  \sum_{x\in\mathcal{X}} P(x) \frac{W_x^\alpha}{\Tr\left[ W_x^\alpha {\sigma}^{1-\alpha} \right] } \partial_{\omega}^2 {\sigma}^{1-\alpha}\right], \label{eq:fix8}
\end{align}
where Eq.~\eqref{eq:fix8} holds by noting that $\partial_{\omega}  {\sigma}^{1-\alpha} \neq 0$ for all $\omega\neq \sigma$.
Further, $\partial_{\omega}^2 {\sigma}^{1-\alpha} \leq 0$ since $u\mapsto u^{1-\alpha}$ is operator concave.
Thus, $\partial_{\omega}^2 g_\alpha(\sigma) < 0$, item \ref{I2-Augustin_mean_} is proved.
\end{proof}
	
\begin{proof}[Proof of Proposition \ref{prop:I2}-\ref{I2-conc_P_}.]
	Recall the definition given in Eq.~\eqref{eq:I2}. The assertion follows because the pointwise infimum of linear functions is concave.	
\end{proof}

\begin{proof}[Proof of Proposition \ref{prop:I2}-\ref{I2-conc_alpha_}.]
	This assertion was proved by Mosonyi and Ogawa \cite[Corollary B.2]{MO14b}.
%
\end{proof}

\begin{proof}[Proof of Proposition \ref{prop:I2}-\ref{I2-cont_alpha_}.]
	 The idea of the proof originate from Nakibo\u{g}lu \cite[Lemmas 16, 17]{Nak18a}.
	 
	 Recalling item~\ref{I2-conc_alpha_}, the map $\alpha \mapsto \frac{1-\alpha}{\alpha} I_\alpha^{(2)}(P,\mathscr{W})$ is concave on $(0,1]$. Since any concave function is continuous in its interior \cite[Corollary 6.3.3]{Dud02}, the map $\alpha \mapsto I_\alpha^{(2)}(P,\mathscr{W})$ is continuous on $(0,1)$.
	 The continuity at $\alpha = 0$ can be verified as follows. Let $\sigma_{0,P}\in\mathcal{S(H)}$ be any state such that $I_0^{(2)}(P,\mathscr{W}) = D_0\left(\mathscr{W}\|\sigma_{0,P}|P\right)$.
	 Then, the monotonicity in item~\ref{I2-mono_} and the definition of $I_\alpha^{(2)}$ given in Eq.~\eqref{eq:I2}  imply that
	 \begin{align}
	 I_0^{(2)}(P,\mathscr{W}) \leq I_\alpha^{(2)}(P,\mathscr{W}) \leq D_\alpha\left( \mathscr{W} \| \sigma_{0,P} |P \right), \quad \forall \alpha \in(0,1].
	 \end{align}
	 The continuity of $\alpha\mapsto D_\alpha$ on $[0,1]$, Lemma~\ref{lemma:chaotic}-\ref{Da_mono_cont}, thus implies $\lim_{\alpha\downarrow 0} I_\alpha\left( P,\mathscr{W} \right) = I_0(P,\mathscr{W})$.
	 
	 It remains to show the continuity at $\alpha = 1$.	 
	 We claim the following fact about the Augustin mean.
	 \begin{lemm} \label{lemm:Augustin_mean}
		Given $\alpha \in (0,1]$ and $P\in\mathscr{P}(\mathcal{X})$, we let $\sigma_{\alpha,P} \in \mathcal{S(H)}$ be the Augustin mean, i.e.~$I_\alpha^{(2)}(P,\mathscr{W}) = D_\alpha\left(\mathscr{W}\|\sigma_{\alpha,P}|P\right)$.
		The following hold.
		\begin{itemize}
			\item For any $\alpha \in [\sfrac12, 1]$, 
			\begin{align}
			D_\alpha\left( W_x \| \sigma_{\alpha,P} \right) \leq \log \frac{1}{P(x)}. \label{eq:Augustin_mean0}
			\end{align}
			
			\item For any $\alpha \in [\sfrac12, 1]$,
			\begin{align} \label{eq:Augustin_mean00}
			\sigma_{\alpha,P} \leq \left( \min_{x:P(x)>0} P(x) \right)^{\frac{\alpha-1}{\alpha}} \sigma_{1,P}.
			\end{align}
		\end{itemize}
	 \end{lemm}
	 Then, Eq.~\eqref{eq:Augustin_mean00} and Lemma~\ref{lemma:chaotic}-\ref{Da_second_mono} imply that
	 \begin{align}
	 D_\alpha\left(\mathscr{W}\|\sigma_{\alpha,P}|P\right) \geq  D_\alpha\left(\mathscr{W}\|\sigma_{1,P}|P\right) + \frac{1-\alpha}{\alpha} \log \left( \min_{x:P(x)>0} P(x) \right), \quad \forall \alpha\in\left[\sfrac12,1\right].
	 \end{align}	 
	 Moreover, the fact $D_\alpha\left(\mathscr{W}\|\sigma_{\alpha,P}|P\right) = I_\alpha^{(2)}(P,\mathscr{W})$ in item~\ref{I2-Augustin_mean_} and the monotonicity in item~\ref{I2-mono_} show that
	 \begin{align}
	 D_\alpha\left(\mathscr{W}\|\sigma_{1,P}|P\right) + \frac{1-\alpha}{\alpha} \log \left( \min_{x:P(x)>0} P(x) \right)
	 \leq I_\alpha^{(2)}(P,\mathscr{W}) \leq I_1^{(2)}(P,\mathscr{W}), \quad \forall \alpha\in\left[\sfrac12,1\right].
	 \end{align}
	 It is clearly that $D_\alpha\left(\mathscr{W}\|\sigma_{1,P}|P\right)<+\infty$ by the fact that $\sigma_{1,P} \gg W_x$ for all $x \in \texttt{supp}(P)$.
	 Then, the continuity of $\alpha\mapsto D_\alpha$ on $[0,1]$, Lemma~\ref{lemma:chaotic}-\ref{Da_mono_cont}, together with the fact $D_1\left(\mathscr{W}\|\sigma_{1,P}|P\right) = I_1^{(2)}(P,\mathscr{W})$ imply the continuity of $I_\alpha^{(2)}$ at $\alpha=1$.
	 
	 In the following, we prove Lemma~\ref{lemm:Augustin_mean} to complete the proof of item~\ref{I2-cont_alpha_}.
	 \begin{proof}[Proof of Lemma~\ref{lemm:Augustin_mean}]
	 Proposition~\ref{prop:I2}~\ref{I2-Augustin_mean_} implies that the Augustin mean satisfies
	 \begin{align} \label{eq:Augustin_mean1}
	 \sigma_{\alpha,P} = \left( \sum_{x\in\mathcal{X}} P(x)  W_x^\alpha \e^{(1-\alpha)D_\alpha(W_x\|\sigma_{\alpha,P})} \right)^{\frac{1}{\alpha}}.
	 \end{align}
	 Using the operator monotonicity of $(\cdot)^{\frac{1-\alpha}{\alpha}}$ for $\alpha \in [\sfrac12, 1]$ (see e.g.~\cite{Bha97, HP14}), we have
	 \begin{align} \label{eq:Augustin_mean3}
	 \sigma_{\alpha,P}^{1-\alpha } \geq P(x)^{\frac{1-\alpha}{\alpha}} W_x^{1-\alpha} \e^{\frac{(1-\alpha)^2}{\alpha} D_\alpha(W_x\|\sigma_{\alpha,P})}.
	 \end{align}
	 Then, Eq.~\eqref{eq:Augustin_mean3} and Lemma~\ref{lemma:chaotic}-\ref{Da_second_mono} imply that
	 \begin{align}
	 D_\alpha\left( W_x \| \sigma_{\alpha,P} \right) &= \frac{1}{\alpha-1} \log \Tr\left[ W_x^\alpha \sigma_{\alpha,P}^{1-\alpha} \right] \\
	 &\leq -\log P(x)^{\frac{1}{\alpha}} - \frac{1-\alpha}{\alpha} D_\alpha\left(W_x\|\sigma_{\alpha,P}\right),
	 \end{align}
	 which proves the first claim.
	 
	 For any $\alpha \in [\sfrac12, 1]$, the map $(\cdot )^{\frac{1}{\alpha}}$ is operator convex.
	 Then, Eq.~\eqref{eq:Augustin_mean1} yields
	 \begin{align} \label{eq:Augustin_mean2}
	 \sigma_{\alpha,P} \leq \sum_{x\in\mathcal{X}} P(x) W_x \e^{\frac{1-\alpha}{\alpha} D_\alpha(W_x\|\sigma_{\alpha,P})}.
	 \end{align}
	 Note that $P\mathscr{W} = \sigma_{1,P}$.
	 Applying Eq.~\eqref{eq:Augustin_mean1} on \eqref{eq:Augustin_mean2} gives the desired result in Eq.~\eqref{eq:Augustin_mean00}.
	 \end{proof}
\end{proof}
	
\begin{proof}[Proof of Proposition \ref{prop:I2}-\ref{I2-cont_equi_}]

To prove the equicontinuity, we need the following inequality:
\begin{align}
I_\alpha^{(2) }(\alpha,P_\beta) &\leq \beta I_\alpha^{(2) }(\alpha,P_1) + (1-\beta) I_\alpha^{(2) }(\alpha,P_0) + H(\beta) \label{eq:ec2}
\end{align}
for any $P_1, P_0 \in \mathscr{P}(\mathcal{X})$, $P_\beta = \beta P_1 + (1-\beta) P_0$, $\beta \in (0,1)$, $\alpha\in [0,1]$; and we shorthand $H(\beta):= -\beta \log \beta - (1-\beta)\log (1-\beta)$ the binary entropy function.

Let $\sigma_{\alpha,P} \in \mathcal{S(H)}$ be the Augustin mean\footnote{For $\alpha=1$, the Augustin mean is not unique. We note that the proof of item~\ref{I2-cont_equi_} does not require the uniqueness.} as in item~\ref{I2-Augustin_mean_} for $\alpha\in[0,1]$.
Lemma~\ref{lemma:chaotic}-\ref{Da_second_mono} implies that, for every $\alpha \in [0,1]$,
\begin{align}
&\sum_{x\in\mathcal{X}} P_\beta(x) D_{\alpha} \left(W_x\| \beta \sigma_{\alpha,P_1} + (1-\beta) \sigma_{\alpha,P_0} \right) \notag \\
&= \beta \sum_{ x \in \mathcal{X} } P_1(x) D_{\alpha}\left(W_x\| \beta \sigma_{\alpha,P_1}  + (1-\beta) \sigma_{\alpha,P_0} \right) + (1-\beta) \sum_{ x \in \mathcal{X} } P_0(x) D_{\alpha} \left( W_x\| \beta \sigma_{\alpha,P_1} + (1-\beta) \sigma_{\alpha,P_0} \right) \\
&\leq \beta \sum_{x \in \mathcal{X} } P_1(x) D_{\alpha}  \left( W_x\| \sigma_{\alpha, P_1} \right) - \beta\log\beta + (1-\beta) \sum_{ x\in\mathcal{X} } P_0(W_x) D_{\alpha} \left( W_x\| \sigma_{\alpha, P_0} \right) - (1-\beta)\log (1-\beta) \\
&= \beta I_\alpha^{(2)}(P_1,\mathscr{W}) + (1-\beta) I_\alpha(P_0, \mathscr{W}) + H(\beta).
\end{align}

Let $s_\wedge, s_1, s_0$ be
\begin{align}
s_\wedge &= \frac{ P_1\wedge P_0 }{ \left\| P_1\wedge P_0 \right\|_1}, \\
s_1 &= \frac{ P_1 - P_1\wedge P_0 }{ 1 - \left\| P_1\wedge P_0 \right\|_1 }, \\
s_0 &= \frac{ P_0 - P_1\wedge P_0 }{ 1 - \left\| P_1\wedge P_0 \right\|_1 }.
\end{align}
Here, $\|\cdot\|_1$ denotes the $\ell_1$-norm.
One can verify that (see e.g.~\cite{Nak18a})
\begin{align}
P_1 &= \left( 1 - \frac{\left\| P_1- P_0 \right\|_1}{2} \right) s_\wedge + \frac{\left\| P_1- P_0 \right\|_1}{2} s_1, \\
P_0 &= \left( 1 - \frac{\left\| P_1- P_0 \right\|_1}{2} \right) s_\wedge + \frac{\left\| P_1- P_0 \right\|_1}{2} s_0.
\end{align}

Then, the concavity of $P\mapsto I_\alpha^{(2)}(P,\mathscr{W})$ given in item~\ref{I2-conc_P_} together with Eq.~\eqref{eq:ec2} yield
\begin{align}
I_\alpha^{(2) }(P_0, \mathscr{W}) - I_\alpha^{(2) }(P_1, \mathscr{W}) &\leq H\left( \frac{\left\| P_1- P_0 \right\|_1}{2} \right) + \frac{\left\| P_1- P_0 \right\|_1}{2} \left( I_\alpha^{(2) }(s_0,\mathscr{W})  - I_\alpha^{(2) }(s_1, \mathscr{W}) \right) \\
&\leq H\left( \frac{\left\| P_1- P_0 \right\|_1}{2} \right) + \frac{\left\| P_1- P_0 \right\|_1}{2}  I_\alpha^{(2) }(s_0,\mathscr{W})
\end{align}
for all $\alpha \in [0,1]$.
Thus,
\begin{align}
\left| I_\alpha^{(2) }(P_0, \mathscr{W}) - I_\alpha^{(2) }(P_1, \mathscr{W}) \right| &\leq H\left( \frac{\left\| P_1- P_0 \right\|_1}{2} \right) + \frac{\left\| P_1 - P_0 \right\|_1}{2} \log |\mathcal{X}|
\end{align}
since $\alpha \mapsto I_\alpha^{(2) }$ is monotone increasing by item~\ref{I2-mono_}.
The above inequality implies the desired equicontinuity.

The joint continuity of $(\alpha,P) \mapsto I_\alpha^{(2)}(P,\mathscr{W})$ follows from the continuity of $\alpha\mapsto I_\alpha^{(2)}(P,\mathscr{W})$ given in \ref{I2-cont_alpha_} and uniform equicontinuity.

\end{proof}

\begin{proof}[Proof of Proposition \ref{prop:I2}-\ref{I2-cont_mean_}.]
Let $(\alpha_k, P_k)_{k\in\mathbb{N}}$ be an arbitrary sequence such that $\alpha_k\in (0,1]$, $P_k\in\mathscr{P}(\mathcal{X})$, and $\lim_{k\to+\infty} (\alpha_k,P_k) = (\alpha_0, P_0) \in (0,1]\times\mathscr{P}(\mathcal{X})$.
Further, let $(\sigma_{\alpha_k, P_k})_{k\in\mathbb{N}}$ be the sequence of the Augustin mean corresponding to $(\alpha_k, P_k)$. Since $\mathcal{S(H)}$ is compact, there exists a convergent subsequence $\{k_{l}\}_{l\in\mathbb{N}}$ such that
$\lim_{l\to+\infty} \sigma_{\alpha_{k_{l}}, P_{k_l}} = \sigma_0$
for some $\sigma_0\in\mathcal{S(H)}$.

The joint continuity of $(\alpha,P)\mapsto I_\alpha^{(2)}(P,\mathscr{W})$ in item~\ref{I2-cont_equi} thus implies
\begin{align}
\lim_{k\to+\infty} I_{\alpha_k}^{(2)}(P_k,\mathscr{W}) = 
D_{\alpha_0}(\mathscr{W}\|\sigma_0|P_0) = I_{\alpha_0}^{(2)}(P_0,\mathscr{W}) = 	D_{\alpha_0}(\mathscr{W}\|\sigma_{\alpha_0,P_0}|P_0).
\end{align}
Then, the uniqueness of the minimizer $\sigma_{\alpha,P}$ in item~\ref{I2-Augustin_mean} guarantees that $\sigma_0 = \sigma_{\alpha_0,P_0}$.
Hence, 
\begin{align}
\lim_{k\to+\infty} \sigma_{\alpha_k,\sigma_k} =\sigma_0 = \sigma_{\alpha_0,\sigma_0},
\end{align}
which proves item~\ref{I2-cont_mean_}.
\end{proof}

\begin{proof}[Proof of Proposition \ref{prop:I2}-\ref{I2-cont_C}.]
Berge's maximum theorem \cite[Section IV.3]{Ber63}, \cite[Lemma 3.1]{Psh71} shows that the continuous map $(\alpha,P) \mapsto I_\alpha^{(2)} (P,\mathscr{W})$ maximized over the compact set $P\in\mathscr{P}(\mathcal{X})$ is still continuous for $\alpha \in [0,1]$.	
\end{proof}

	\section{Proof of Proposition \ref{prop:saddle}} \label{app:saddle}
\begin{prop3}[Saddle-Point]
	Consider a classical-quantum channel $\mathscr{W} : \mathcal{X}\to \mathcal{S(H)}$, any $R\in( C_{0,\mathscr{W}}, C_\mathscr{W}   )$, and $P\in \mathscr{P}(\mathcal{X})$. Let
	\begin{align} 
	\mathcal{S}_{P,\mathscr{W}}(\mathcal{H}) &:= \left\{ \sigma \in \mathcal{S(H)}: \forall x \in \textnormal{\texttt{supp}}(P), \, W_x \not\perp  \sigma   \right\}.
	\end{align}
	Define
	\begin{align} \label{eq:FF}
	F_{R,P} (\alpha, \sigma) := 
	\begin{dcases}
	\frac{1-\alpha}{\alpha} \left(   D_\alpha\left( \mathscr{W} \| \sigma | P  \right) -R \right), & \alpha \in (0,1) \\
	0, & \alpha = 1
	\end{dcases},
	\end{align}
	on $(0,1]\times \mathcal{S}(\mathcal{H})$, and denote by
	\begin{align} \label{eq:PPR}
	\mathscr{P}_R(\mathcal{X}) := \left\{ P\in\mathscr{P}(X) : \sup_{0<\alpha\leq 1} \inf_{\sigma \in \mathcal{S(H)}}  F_{R,P} (\alpha, \sigma) \in \mathbb{R}_{>0}    \right\}.
	\end{align}
	The following holds
	\begin{enumerate}[(a)]
		\item\label{saddle-aa} For any $P\in\mathscr{P}(\mathcal{X})$, $F_{R,P}(\cdot,\cdot)$ has a saddle-point on $(0,1]\times \mathcal{S}_{P,\mathscr{W}}(\mathcal{H})$ with  the saddle-value:
		\begin{align}
		\min_{\sigma \in \mathcal{S(H)}}\sup_{0<\alpha\leq 1}   F_{R,P} (\alpha,\sigma) = \sup_{0<\alpha\leq 1} \min_{\sigma \in \mathcal{S(H)}}  F_{R,P} (\alpha,\sigma)= E_\textnormal{sp}^{(2)}(R,P).
		\end{align}
		
		\item\label{saddle-bb} 
		Fix $P\in\mathscr{P}_R(\mathcal{X})$. Any saddle-point $(\alpha_{R,P}^\star, \sigma_{R,P}^\star)$ of $F_{R,P}(\cdot,\cdot)$ satisfies $\alpha_{R,P}^\star\in(0,1)$ and 
		\begin{align}
		\sigma_{R,P}^\star 
		\gg W_x, \quad \forall x \in \textnormal{\texttt{supp}}(P).
		\end{align}
		
		\item\label{saddle-cc} For $P\in\mathscr{P}_R(\mathcal{X})$, the saddle-point is unique.
		
		\item\label{saddle-dd} For any $\underline{R} \in (C_{0,\mathscr{W}}, R]$, both $\alpha_{r,P}^\star $ and $\sigma_{r,P}^\star$ are jointly continuous functions of $(r,P)$ on $[\underline{R},R]\times \mathscr{P}(\mathcal{X})$.
		
	\end{enumerate}
\end{prop3}	
\begin{proof}[Proof of Proposition \ref{prop:saddle}-\ref{saddle-aa}]
		Fix arbitrary $R>C_{0,\mathscr{W}}$ and $P\in\mathscr{P}(\mathcal{X})$.
		In the following, we prove the existence of a saddle-point of $F_{R,P}(\cdot,\cdot)$ on $(0,1	] \times \mathcal{S}_{P,\mathscr{W}}(\mathcal{H})$. Ref.~\cite[Lemma 36.2]{Roc70} states that $(\alpha^\star,\sigma^\star)$ is a saddle point of $F_{R,P}(\cdot,\cdot)$ if and only if the supremum in 
		\begin{align}
		\sup_{\alpha\in (0,1]} \inf_{\sigma \in \mathcal{S}_{P,\mathscr{W}}(\mathcal{H})} F_{R,P} (\alpha,\sigma) \label{eq:saddle15}
		\end{align}
		is attained at $\alpha^\star\in(0,1]$, the infimum in 
		\begin{align}
		\inf_{\sigma \in \mathcal{S}_{P,\mathscr{W}}(\mathcal{H})} \sup_{\alpha\in (0,1]}  F_{R,P} (\alpha,\sigma) \label{eq:saddle16}
		\end{align}
		is attained at $\sigma^\star\in \mathcal{S}_{P,\mathscr{W}}(\mathcal{H})$, and the two extrema in Eqs.~\eqref{eq:saddle15}, \eqref{eq:saddle16} are equal and finite.
		We first claim that, $\forall \alpha\in(0,1],$
		\begin{align}
		\inf_{\sigma \in \mathcal{S}_{P,\mathscr{W}}(\mathcal{H})} F_{R,P}(\alpha,\sigma)
		= \inf_{\sigma \in \mathcal{S}(\mathcal{H})} F_{R,P}(\alpha,\sigma). \label{eq:saddle21}
		\end{align}
		To see this, observe that for any $\alpha \in (0,1)$, Eqs.~\eqref{eq:Petz} and \eqref{eq:Petz_P} yield
		\begin{align}
		\forall \sigma \in \mathcal{S(H)}\backslash \mathcal{S}_{P,\mathscr{W}}(\mathcal{H}), \quad 
		D_{\alpha}\left( \mathscr{W}\|\sigma | P \right) = +\infty, 
		\label{eq:saddle23}
		\end{align}
		which, in turn, implies
		\begin{align} \label{eq:saddle22}
		\forall \sigma \in \mathcal{S(H)}\backslash \mathcal{S}_{P,\mathscr{W}}(\mathcal{H}), \quad 
		F_{R,P}(\alpha,\sigma ) = +\infty.
		\end{align}
		Further, Eq.~\eqref{eq:saddle21} holds trivially when $\alpha=1$.
		Hence, Eq.~\eqref{eq:saddle21} yields
		\begin{align}
		\begin{split}
		\sup_{\alpha\in (0,1]} \inf_{\sigma \in \mathcal{S}_{P,\mathscr{W}}(\mathcal{H})} F_{R,P} (\alpha,\sigma)
		&= \sup_{\alpha\in (0,1]} \inf_{\sigma \in \mathcal{S}(\mathcal{H})} F_{R,P} (\alpha,\sigma) 
		\end{split}
		\end{align}
		Owing to the fact $R> C_{0,\mathscr{W}}$ and Eq.~\eqref{eq:Esp2}, we have
		\begin{align}
		E_\text{sp}^{(2)} (R,P) = \sup_{\alpha\in (0,1]} \inf_{\sigma \in \mathcal{S}(\mathcal{H})} F_{R,P} (\alpha,\sigma) < +\infty, \label{eq:saddle28}
		\end{align}
		which guarantees the supremum in the right-hand side of Eq.~\eqref{eq:saddle28} is attained at some $\alpha\in(0,1]$. Namely, there exists some $\bar{\alpha}_{R,P} \in (0,1]$ such that
		\begin{align} \label{eq:fact1}
		\sup_{\alpha\in (0,1]} \inf_{\sigma \in \mathcal{S}_{P,\mathscr{W}}(\mathcal{H})} F_{R,P} (\alpha,\sigma)
		= \max_{\alpha\in [\bar{\alpha}_{R,P},1]} \inf_{\sigma \in \mathcal{S}(\mathcal{H})} F_{R,P} (\alpha,\sigma)  < +\infty.
		\end{align}
		Thus, we complete our claim in Eq.~\eqref{eq:saddle15}. It remains to show that the infimum in Eq.\eqref{eq:saddle16} is attained at some $\sigma^\star \in \mathcal{S}_{P,\mathscr{W}}(\mathcal{H})$ and the supremum and infimum are exchangeable. To achieve this, we will show that $\left( [\bar{\alpha}_{R,P},1], \mathcal{S}_{P,\mathscr{W}}(\mathcal{H}), F_{R,P} \right)$ is a closed saddle-element (see Definition \ref{defn:saddle} below) and employ the boundness of $[\bar{\alpha}_{R,P},1]\times \mathcal{S}_{P,\mathscr{W}}(\mathcal{H})$ to conclude our claim.

\begin{defn} [Closed Saddle-Element {\cite{Roc64}}] \label{defn:saddle}
	We denote by $\texttt{ri}$ and $\texttt{cl}$ the relative interior and the closure of a set, respectively.	Let $\mathcal{A},\mathcal{B}$ be subsets of a real vector space, and $F:\mathcal{A}\times \mathcal{B}\to\mathbb{R}\cup\{\pm \infty  \}$.	
	The triple $\left(\mathcal{A},\mathcal{B}, F\right)$ is called a closed saddle-element if for any $x\in \texttt{ri}\left(\mathcal{A}\right)$ (resp.~$y\in \texttt{ri}\left(\mathcal{B}\right)$), 
	\begin{itemize}
		\item[(i)] $\mathcal{B}$ (resp.~$\mathcal{A}$) is convex.
		\item[(ii)] $F(x,\cdot)$ (resp.~$F(\cdot, y)$) is convex (resp.~concave) and lower (resp.~upper) semi-continuous.
		\item[(iii)] Any accumulation point of $\mathcal{B}$ (resp.~$\mathcal{A}$) that does not belong to $\mathcal{B}$ (resp.~$\mathcal{A}$), say $y_o$ (resp.~$x_o$) satisfies $\lim_{y\to y_o} F(x,y) = +\infty$ (resp.~$\lim_{x\to x_o} F(x, y) = -\infty$).
	\end{itemize}
\end{defn}
		
		Fix an arbitrary $\alpha \in \texttt{ri}\left( [\bar{\alpha}_{R,P},1] \right) = (\bar{\alpha}_{R,P},1)$.
		We check that $\left( \mathcal{S}_{P,\mathscr{W}}(\mathcal{H}), F_{R,P}(\alpha, \cdot)\right)$ fulfills the three items in Definition \ref{defn:saddle}.
		(i) The set $\mathcal{S}_{P,\mathscr{W}}(\mathcal{H})$ is clearly convex.
		(ii) Recall Lemma~\ref{lemma:chaotic}-\ref{Da_second_convex} that $\sigma\mapsto D_{\alpha}(W_x\|\sigma)$ is convex and lower semi-continuous. Since convex combination preservers the convexity and the  lower semi-continuity, Eq.~\eqref{eq:FF} yields that  $\sigma \mapsto F_{R,P}(\alpha,\sigma)$ is convex and lower semi-continuous on $\mathcal{S}_{P,\mathscr{W}}(\mathcal{H})$.
		(iii) Due to the compactness of $\mathcal{S(H)}$, any accumulation point of $\mathcal{S}_{P,\mathscr{W}}(\mathcal{H})$ that does not belong to $\mathcal{S}_{P,\mathscr{W}}(\mathcal{H})$, say $\sigma_o$, satisfies $\sigma_o \in \mathcal{S(H)} \backslash \mathcal{S}_{P,\mathscr{W}}(\mathcal{H})$. Eqs.~\eqref{eq:saddle23} and \eqref{eq:saddle22} then show that $F_{R,P}(\alpha, \sigma_o) = +\infty$.
		
		Next, fix an arbitrary $\sigma \in \texttt{ri}\left(  \mathcal{S}_{P,\mathscr{W}}(\mathcal{H}) \right)$. Owing to the convexity of $\mathcal{S}_{P,\mathscr{W}}(\mathcal{H})$, it follows that $\texttt{ri}\left(  \mathcal{S}_{P,\mathscr{W}}(\mathcal{H}) \right) $ $= \texttt{ri}\left(\texttt{cl}\left( \mathcal{S}_{P,\mathscr{W}}(\mathcal{H})\right)\right)$ (see e.g.~\cite[Theorem 6.3]{Roc70}).
		We first claim $\texttt{cl}\left( \mathcal{S}_{P,\mathscr{W}}(\mathcal{H})\right) = \mathcal{S(H)}$. To see this, observe that $\mathcal{S}_{>0}(\mathcal{H}) \subseteq  \mathcal{S}_{P,\mathscr{W}}(\mathcal{H})$ since a full-rank density operator is not orthogonal with every $W_x$, $x\in\mathcal{X}$.
		Hence, 
		\begin{align}
		\mathcal{S(H)}=
		\texttt{cl}\left( \mathcal{S}_{>0}(\mathcal{H}) \right) 
		\subseteq \texttt{cl} \left(  \mathcal{S}_{P,\mathscr{W}}(\mathcal{H}) \right). \label{eq:saddle25}
		\end{align}
		On the other hand, the fact $\mathcal{S}_{P,\mathscr{W}}(\mathcal{H}) \subseteq  \mathcal{S(H)}$ leads to
		\begin{align}
		\texttt{cl}\left( \mathcal{S}_{P,\mathscr{W}}(\mathcal{H}) \right) \subseteq  
		\texttt{cl}\left(\mathcal{S(H)}\right) 
		= \mathcal{S(H)}. \label{eq:saddle26}
		\end{align}
		By Eqs.~\eqref{eq:saddle25} and \eqref{eq:saddle26}, we deduce that 
		\begin{align}
		\texttt{ri}\left(  \mathcal{S}_{P,\mathscr{W}}(\mathcal{H}) \right)
		= \texttt{ri}\left( \texttt{cl}\left( \mathcal{S}_{P,\mathscr{W}}(\mathcal{H}) \right) \right)
		= \texttt{ri}\left(  \mathcal{S}(\mathcal{H}) \right)
		=  \mathcal{S}_{>0}(\mathcal{H}), \label{eq:saddle24}
		\end{align}
		where the last equality in Eq.~\eqref{eq:saddle24} follows from \cite[Proposition 2.9]{Wei11}.
		Hence, we obtain
		\begin{align} \label{eq:saddle19}
		\forall \sigma \in \texttt{ri}\left(  \mathcal{S}_{P,\mathscr{W}}(\mathcal{H}) \right) \quad \text{and} \quad \forall x\in \mathcal{X},
		\quad \sigma \gg W_x.
		\end{align}
		Now we verify that $\left( [\bar{\alpha}_{R,P},1], F_{R,P}(\cdot,\sigma)\right)$ satisfies the three items in Definition \ref{defn:saddle}.
		Fix an arbitrary $\sigma \in \texttt{ri}\left(  \mathcal{S}_{P,\mathscr{W}}(\mathcal{H}) \right)$.
		(i) The set $(0,1]$ is obviously convex.
		(ii) From Lemma~\ref{lemma:chaotic}-\ref{Da_mono_cont}, the map $\alpha \mapsto F_{R,P}(\alpha,\sigma)$ is continuous on $(0,1)$. Further, it is not hard to verify that $F_{R,P}(1,\sigma) = 0 = \lim_{\alpha\uparrow 1} F_{R,P}(\alpha,\sigma)$ from Eqs.~\eqref{eq:saddle19}, \eqref{eq:FF}, and \eqref{eq:Petz}.
		Item \ref{I2-conc_alpha} in Proposition~\ref{prop:I2} and \cite[Collorary B.2]{MO14b}
		implies that $\alpha \mapsto F_{R,P}(\alpha,\sigma)$ on $[\bar{\alpha}_R,1)$ is concave.
		Moreover, the continuity of $\alpha \mapsto F_{R,P}(\alpha,\sigma)$ on $[\bar{\alpha}_{R,P},1)$ guarantees the concavity of $\alpha \mapsto F_{R,P}(\alpha,\sigma)$ on $[\bar{\alpha}_{R,P},1]$.
		(iii) Since $[\bar{\alpha}_{R,P},1]$ is closed, there is no accumulation point of $[\bar{\alpha}_{R,P},1]$ that does not belong to $[\bar{\alpha}_{R,P},1]$.
		
		We are at the position to prove item \ref{saddle-aa} of Proposition \ref{prop:saddle}.
The closed saddle-element, along with the boundness of $\mathcal{S}_{P,\mathscr{W}}(\mathcal{H})$ and Rockafellar's saddle-point result \cite[Theorem 8]{Roc64}, \cite[Theorem 37.3]{Roc70} imply that
		\begin{align}
		- \infty < \sup_{ \alpha\in [\bar{\alpha}_{R,P},1] } \inf_{\sigma \in \mathcal{S}_{P,\mathscr{W}}(\mathcal{H})} F_{R,P} (s,\sigma)
		= \min_{\sigma \in \mathcal{S}_{P,\mathscr{W}}(\mathcal{H})} \sup_{ \alpha\in[\bar{\alpha}_{R,P},1] }  F_{R,P} (s,\sigma). \label{eq:fact2}
		\end{align}
		Then Eqs.~\eqref{eq:fact1} and \eqref{eq:fact2} lead to the existence of a saddle-point of $F_{R,P}(\cdot,\cdot)$ on $(0,1]\times \mathcal{S}_{P,\mathscr{W}}(\mathcal{H})$.
		Hence, item \ref{saddle-aa} is proved.

\end{proof}		

\begin{proof}[Proof of Proposition \ref{prop:saddle}-\ref{saddle-bb}]
	Fix arbitrary $R\in(C_{0,\mathscr{W}},C_\mathscr{W})$ and $P\in \mathscr{P}_R(\mathcal{X})$.
	We have 
	\begin{align} \
	 &\sup_{0<\alpha\leq 1} \min_{\sigma \in \mathcal{S(H)}}  F_{R,P} (\alpha, \sigma) 
	 =\min_{\sigma \in \mathcal{S(H)}} \sup_{0<\alpha\leq 1}  F_{R,P} (\alpha, \sigma)   \in \mathbb{R}_{>0}  \label{eq:saddle1}
	\end{align}
	by the saddle-point property in item~\ref{saddle-aa} and the definition of $\mathscr{P}_R(\mathcal{X})$ given in Eq.~\eqref{eq:PPR}.
	First note that $(1,\sigma)$ for any $\sigma\in\mathcal{S(H)}$ will not be  a saddle point of $F_{R,P}(\cdot,\cdot)$ because $F_{R,P} (1, \sigma) = 0$, $\forall\sigma \in \mathcal{S(H)}$,  contradicting Eq.~\eqref{eq:saddle1}.
	
	Next, we assume $(\alpha^\star, \sigma^\star)$ is a saddle-point of $F_{R,P}(\cdot,\cdot)$ with $\alpha^\star \in (0,1)$, it holds that
	\begin{align} \label{eq:saddle3}
	F_{R,P} (\alpha^\star, \sigma^\star) = \min_{\sigma\in\mathcal{S(H)}} F_{R,P} (\alpha^\star, \sigma)
	= \frac{\alpha^\star - 1}{\alpha^\star} R + \frac{1-\alpha^\star}{\alpha^\star}  \min_{\sigma\in\mathcal{S(H)}}  D_{\alpha^\star} (\mathscr{W}\|\sigma|P).
	\end{align}
	Since $\sigma^\star$ is the minimizer of $\min_{\sigma\in\mathcal{S(H)}} D_{\alpha^\star} (\mathscr{W}\|\sigma|P)$, it is clear from Proposition~\ref{prop:I2}-\ref{I2-Augustin_mean}  that
	\begin{align} \label{eq:saddle13}
	\sigma^\star \gg W_x, \quad \forall x \in \texttt{supp}(P),
	\end{align}
	and  thus item \ref{saddle-bb} is proved.
	
\end{proof}

\begin{proof}[Proof of Proposition \ref{prop:saddle}-\ref{saddle-cc}]
	
Continuing from item~\ref{saddle-bb}, we show the uniqueness of the saddle-point.
Since $(1,\sigma)$ for $\sigma\in\mathcal{S(H)}$ will not be a saddle-point of $F_{R,P}(\cdot,\cdot)$ as shown in item~\ref{saddle-bb},
we let $\alpha^\star\in(0,1)$ attain the supremum in the left-hand side of Eq.~\eqref{eq:saddle1}.
Proposition~\ref{prop:I2}-\ref{I2-Augustin_mean} implies that the minimizer to the map $\sigma\mapsto D_{\alpha^\star}(\mathscr{W}\|\sigma |P)$ is unique, and thus it follows that the minimizer of Eq.~\eqref{eq:saddle1} is unique as well.	

Next, 
we will invoke Lemma~\ref{lemm:regularity} in Appendix~\ref{app:LF} to show the uniqueness of the maximizer.
Let $\sigma^\star \in\mathcal{S(H)}$ be the minimizer of right-hand side of the equality in Eq.~\eqref{eq:saddle1}, and
let $\mathbf{x}^n \in \mathcal{X}^n$ be an arbitrary sequence with an empirical distribution $P$.
Denote by $p^n, q^n$ be two distributions with $(p_i,q_i)$ being the Nussbaum-Szko{\l}a mapping of $(W_{x_i}, \sigma^\star)$, where $x_i$ is the $i$-th symbol of $\mathbf{x}^n$ for $i\in[n]$.
Further, item~\ref{saddle-bb} guarantees that $p^n \ll q^n$.

Now, we let $p^n$ and $q^n$ to be the hypotheses described in Eq.~\eqref{eq:pn_qn}.
It is not hard to observe that $\sup_{0<\alpha\leq 1} F_{R,P}(\alpha,\sigma^\star) = \phi_n(R) $ given in Eq.~\eqref{eq:phi_n}.
Items~\ref{regularity-b} and \ref{regularity-d} in Lemma~\ref{lemm:regularity} then show that the optimizer $\alpha^\star \in(0,1)$ of $\sup_{0<\alpha\leq 1} F_{R,P}(\alpha,\sigma^\star)$ is unique, which completes the proof of item~\ref{saddle-cc}.

		\end{proof}

\begin{proof}[Proof of Proposition \ref{prop:saddle}-\ref{saddle-dd}]
	We first prove the joint continuity of $(r,P)\mapsto \alpha_{r,P}^\star$ on $[\underline{R}, R] \times \mathscr{P}(\mathcal{X})$.
	To that end, it suffices to show that $P\mapsto \alpha_{r,P}^\star$ is continuous on $\mathscr{P}(\mathcal{X})$ for every $r\in[\underline{R},R]$, and the family $\{ \alpha_{r,P}^\star  \}_{ P\in\mathscr{P}(\mathcal{X}) }$ is uniformly equicontinuous in $r$ on $[\underline{R},R]$. 
	Moreover, it is equivalent to prove the joint continuity of $(r,P)\mapsto s_{r,P}^\star$ on $[\underline{R}, R] \times \mathscr{P}(\mathcal{X})$ by using the substitution $s^\star_{r,P} := (1-\alpha^\star_{r,P} )/\alpha^\star_{r,P}$. This will ease the burden of notation.
	
	In the following, we show the continuity of  $P\mapsto s_{r,P}^\star$.
	The proof idea of such continuity is similar to \cite[Proposition 3.4]{AW14}. 
	Fix $r\in[\underline{R},R]$, any $P_0 \in \mathscr{P}(\mathcal{X})$ and consider arbitrary $\left\{ P_k \right\}_{k\in\mathbb{N}}$ such that $P_k \in \mathscr{P}(\mathcal{X})$ for all $ k\in\mathbb{N}$, and $\lim_{n\to+\infty} P_k = P_0$.	
	Following from Proposition~\ref{prop:Esp}-\ref{Esp-c} that will be proved later in Appendix~\ref{app:Esp}, we have\footnote{Here, for $E_\text{sp}^{(2)}(r,P) = 0$, we adopt $(1,P\mathscr{W})$ as the saddle-point in Eq.~\eqref{eq:F}, which means $s^\star_{r,P} = 0$.} 
	\begin{align} \label{eq:saddle-d1}
	s^\star_{r,P_k} = - \frac{\partial E_\text{sp}^{(2)}(r,P_k)}{\partial r} \in \mathbb{R}_{\geq 0}.
	\end{align}
	Since $r \geq \underline{R} > C_{0,\mathscr{W}}$, the continuity of $E_\text{sp}^{(2)}(r,\cdot)$ given in Proposition~\ref{prop:Esp}-\ref{Esp-a} that will be proved later shows that
	\begin{align}
	\lim_{k\to+\infty } E_\text{sp}^{(2)} (r, P_k) = E_\text{sp}^{(2)} (r, P_0).
	\end{align}
	Viewing $ ( E_\text{sp}^{(2)} (r, P_k) )_{k\in\mathbb{N}}$ as a sequence of functions that converges to $E_\text{sp}^{(2)} (r, P_0)$,
	Ref.~\cite[Corollary VI.6.2.8]{Hir01} proved that the sequence of first-order derivatives of differentiable convex functions converses to the first-order derivative of the limit. Indeed, Proposition~\ref{prop:Esp}-\ref{Esp-a} guarantees that $E_\text{sp}^{(2)}(\cdot, P)$ is convex. Therefore,
	\begin{align}
	\lim_{k\to+\infty} s_{R,P_k}^\star = \lim_{k\to+\infty} - \left.\frac{\partial E_\text{sp}^{(2)}(r,P_k)}{\partial r}\right|_{r=R} = -\left.\frac{\partial E_\text{sp}^{(2)}(r,P_0)}{\partial r}\right|_{r=R} = s_{R,P_0}^\star,
	\end{align}
	which shows the continuity of $P\mapsto s_{r,P}^\star$ for every $r\in[\underline{R},R]$.

	Next, we prove the equicontinuity.
	Let $R_1, R_2 \in [\underline{R}, R]$ be arbitrary. 
	As will be shown later in Proposition~\ref{prop:Esp}-\ref{Esp-a}, for every $P\in\mathscr{P}(\mathcal{X})$, $E_\text{sp}^{(2)}(\cdot, P)$ is convex and non-increasing on $[0,+\infty]$.
	Using Eq.~\eqref{eq:saddle-d1}, the absolute value of the difference between the first-order derivative of $E_\text{sp}^{(2)}(\cdot, P)$ at $R_1$ and $R_2$ can be calculated as follows
	\begin{align}
	\left| s^*_{R_1,P} - s^*_{R_2,P} \right| \leq s^*_{R_1,P} \vee  s^*_{R_2,P} = s^*_{R_1 \wedge R_1, \mathscr{W}}
	\leq s^*_{\underline{R},\mathscr{W}},
	\label{eq:difference}
	\end{align}
	where $s^*_{\underline{R},P} = (1-\alpha^\star_{\underline{R},P})/\alpha^\star_{\underline{R},P}$ and $\alpha^\star_{\underline{R},P}$ is the optimizer of $E_\text{sp}^{(2)}(\underline{R}, P)$ given in Eq.~\eqref{eq:Esp2}.
	
	For all $P \in \mathscr{P}(\mathcal{X})$ such that $\underline{R} \geq I_1(P,\mathscr{W})$, the right-hand side of Eq.~\eqref{eq:difference} is zero since $E_\text{sp}^{(2)}(\underline{R},P) = 0$ (see Proposition~\ref{prop:Esp}-\ref{Esp-a} again).
	On the other hand, for all $P \in \mathscr{P}(\mathcal{X})$ such that $\underline{R} < I_1(P,\mathscr{W})$, Proposition~\ref{prop:Esp}-\ref{Esp-c} shows that
	\begin{align}	\label{eq:saddle-d3}
	I_{\alpha_{\underline{R},P}^\star }^{(2)}(P,\mathscr{W}) > \underline{R}.
	\end{align}
	Further, since $\underline{R} \in (C_{0,\mathscr{W}}, C_{1,\mathscr{W}})$, the continuous monotone increase of the map $\alpha\mapsto C_{\alpha,\mathscr{W}}$ proved in Proposition~\ref{prop:I2}-\ref{I2-cont_C} guarantees that there exists a $\alpha_{\underline{R}} \in (0,1)$ such that
	\begin{align}
	C_{ \alpha_{\underline{R}} } = \underline{R}. \label{eq:saddle-d4}
	\end{align}
	Then, from Eqs.~\eqref{eq:saddle-d3}, \eqref{eq:saddle-d4}, and the definition of the R\'enyi information radius given in Eq.~\eqref{eq:radius}, we have
	\begin{align}
	I_{\alpha_{\underline{R}} }^{(2)}(P,\mathscr{W}) 
	\leq C_{\underline{R}, \mathscr{W}} = \underline{R} < I_{\alpha_{\underline{R},P}^\star }^{(2)}(P,\mathscr{W}), \label{eq:saddle-d5}
	\end{align}
	The above inequality \eqref{eq:saddle-d5} and the monotone increases of the map $\alpha\mapsto I_{\alpha}^{(2)} (P,\mathscr{W})$ further imply that
	\begin{align}
	\alpha_{\underline{R}} < \alpha_{\underline{R},P}^\star. \label{eq:saddle-d2}
	\end{align}
	Both Eqs.~\eqref{eq:difference} and \eqref{eq:saddle-d2} then yield
	\begin{align}
	\left| s^*_{R_1,P} - s^*_{R_2,P} \right| \leq \frac{1-\alpha_{\underline{R}}}{\alpha_{\underline{R}}} < \infty
	\end{align}
	for all $P \in \mathscr{P}(\mathcal{X})$ such that $\underline{R} < I_1(P,\mathscr{W})$.
	This shows the equicontinuity of the family $\{ \alpha_{r,P}^\star  \}_{ P\in\mathscr{P}(\mathcal{X}) }$ on $[\underline{R},R]$. 	
	Together with the continuity of $P\mapsto s^*_{r,P}$ for all $r \in [\underline{R},R]$, the joint continuity of $(r,P) \mapsto s^*_{r,P}$ on $[\underline{R},R] \times \mathscr{P}(\mathcal{X})$ is proved.
	
	Lastly, we move on to prove the continuity of $(r,P)\mapsto \sigma_{r,P}^\star$ on $[\underline{R},R] \times \mathscr{P}(\mathcal{X})$.
	Let $(R_k,P_k) \in [\underline{R},R] \times \mathscr{P}(\mathcal{X})$ for all $k\in\mathbb{N}$ be arbitrary such that $\lim_{k\in\mathbb{N}} (R_k,P_k) = (R_0,P_0) \in [\underline{R},R] \times \mathscr{P}(\mathcal{X})$.
	From Eq.~\eqref{eq:saddle3}, the saddle-point property yields that
	\begin{align} \label{eq:saddle100}
	\sigma_{R_k,P_k}^\star = \sigma_{\alpha_{R_k,P_k}^\star, P_k},
	\end{align}
	where in the right-hand side of the above equality we denote the Augustin mean by $\sigma_{\alpha,P} := \min_{\sigma \in \mathcal{S(H)}} D_\alpha\left(\mathscr{W}\|\sigma|P\right)$.
	Moreover, Proposition~\ref{prop:I2}-\ref{I2-cont_mean} states that $(\alpha,P)\mapsto \sigma_{\alpha,P}$ is jointly continuous on $(0,1]\times \mathscr{P}(\mathcal{X})$.
	Hence, the joint continuity of $(r,P)\mapsto \alpha_{r,P}^\star$ proved above together with Eq.~\eqref{eq:saddle100} show that
	\begin{align}
		\lim_{k\to+\infty} \sigma_{R_k,P_k}^\star &= 	\lim_{k\to+\infty} \sigma_{\alpha_{R_k,P_k}^\star, P_k} \\
		&= \sigma_{\lim_{k\to+\infty} \alpha_{R_k,P_k}^\star, \lim_{k\to+\infty} P_k} \\
		&= \sigma_{ \alpha_{\lim_{k\to+\infty}(R_k, P_k)}^\star, \lim_{k\to+\infty} P_k} \\
		&= \sigma_{ \alpha_{R_0, P_0 }^\star, P_0} \\
		&= \sigma_{R_0,P_0}^\star,
	\end{align}
	which completes the proof of item~\ref{saddle-dd}.
\end{proof}
		

	\section{Proof of Proposition \ref{prop:Esp}} \label{app:Esp}
\begin{prop4}[Properties of Error-Exponent Functions] 
	Consider a classical-quantum channel $\mathscr{W} : \mathcal{X} \to \mathcal{S(H)}$ with $C_{0,\mathscr{W}} < C_\mathscr{W}$.  We have
	\begin{enumerate}[(a)]
		
		\item\label{Esp-aa}  
		Given every $P\in\mathscr{P}(\mathcal{X})$, $E_\textnormal{sp}^{(2)}(\cdot,P)$ is convex and non-increasing on $[0,+\infty]$, and continuous  on $\left[ I_0^{(2)}(P,\mathscr{W}) ,  +\infty \right]$. For every $R>C_{0,\mathscr{W}}$, $E_\textnormal{sp}^{(2)}(R,\cdot)$ is continuous on $\mathscr{P}(\mathcal{X})$. Further,
		\begin{align}
		&E_\textnormal{sp}^{(2)} (R,P) =
		\begin{dcases}
		+\infty, & R< I_0^{(2)}(P,\mathscr{W}) \\
		0, & R\geq I_1^{(2)}(P,\mathscr{W})	
		\end{dcases}.
		\end{align}
		
		\item\label{Esp-bb}   $E_\textnormal{sp}(\cdot)$ is convex and non-increasing  on $ [0,+\infty]$,  and continuous on $ [ C_{0,\mathscr{W}} ,  +\infty ]$. Further,
		\begin{align}
		&E_\textnormal{sp} (R) =
		\begin{dcases}
		+\infty, & R< C_{0,\mathscr{W}} \\
		0, & R\geq C_{1,\mathscr{W}}
		\end{dcases}.
		\end{align}
		
		\item\label{Esp-cc} Consider any $R\in(C_{0,\mathscr{W}}, C_\mathscr{W})$ and $P\in\mathscr{P}_R(\mathcal{X})$ (see Eq.~\eqref{eq:PR}). The function $E_\textnormal{sp}^{(2)}(\cdot, P)$ is differentiable with
		\begin{align} \label{eq:ss2}
		s_{R,P}^{\star} := \frac{1-\alpha_{R,P}^\star}{\alpha_{R,P}^\star} = - \left.\frac{\partial E_\textnormal{sp}^{(2)}(r,P)}{\partial r}\right|_{r=R} \in \mathbb{R}_{>0},
		\end{align}
		where $ \alpha_{R,P}^{\star} $ is the optimizer in Eq.~\eqref{eq:Esp2}.
		Moreover, 
		\begin{align}
		I_{\alpha_{R,P}^\star} (P,\mathcal{W}) > R. \label{eq:Esp_100}
		\end{align}
		
	\end{enumerate}
\end{prop4}	
\begin{proof}[Proof of Proposition \ref{prop:Esp}-\ref{Esp-aa}]
		Fix any arbitrary $P\in\mathscr{P}(\mathcal{X})$. Item \ref{I2-mono} in Proposition \ref{prop:I2} shows that the map $\alpha \mapsto I_\alpha^{(2)}(P,\mathscr{W})$ is monotone increasing on $[0,1]$. Hence, from the definition in Eq.~\eqref{eq:Esp2}, it is not hard to verify that $E_\text{sp}^{(2)}(R,P) = +\infty$ for all $ R\in(0, 
		I_0^{(2)}(P,\mathscr{W}) )$; finite for all $  R>I_0^{(2)}(P,\mathscr{W})$; and $E_\text{sp}^{(2)}(R,P) = 0$, for all $ R\geq  I_1^{(2)}(P,\mathscr{W})$.	
		
		For every $\alpha\in(0,1]$, the function $\frac{1-\alpha}{\alpha} ( I_{\alpha}^{(2)}(P,\mathscr{W}) - R )$ in Eq.~\eqref{eq:Esp2} is an non-increasing, convex, and continuous function in $R\in\mathbb{R}_{>0}$. Since $E_\text{sp}^{(2)}(R,P)$ is the pointwise supremum of the above function, $E_\text{sp}^{(2)}(R,P)$ is non-increasing, convex, and lower semi-continuous function for all $R\geq 0$. Furthermore, since a convex function is continuous on the interior of the interval if it is finite \cite[Corollary 6.3.3]{Dud02}, thus $E_\text{sp}^{(2)}(R,P)$ is continuous for all $R > I_0^{(2)}(P,\mathscr{W})$, and continuous from the right at $R = I_0^{(2)}(P,\mathscr{W})$.
		
		To establish the continuity of $E_\text{sp}^{(2)}(R,P)$ in $P\in\mathscr{P}(\mathcal{X})$, we first claim that there exists some $\bar{\alpha}_R\in(0,1]$ such that for every $P\in\mathscr{P}(\mathcal{X})$,
		\begin{align}
		\sup_{\alpha\in(0,1]} \frac{1-\alpha}{\alpha}\left( I_\alpha^{(2)}(P,\mathscr{W}) - R\right)
		= \sup_{\alpha\in [\bar{\alpha}_R,1]} \frac{1-\alpha}{\alpha}\left( I_\alpha^{(2)}(P,\mathscr{W}) -R \right). \label{eq:Esp-a2}
		\end{align}
		Recall that $R> C_{0,\mathscr{W}} = \max_{P\in\mathscr{P}(\mathcal{X})} I_0^{(2)}(P,\mathscr{W})$. The continuity, item \ref{I2-cont_C} in Proposition \ref{prop:I2}, implies that there exists an $\bar{\alpha}_R>0$ such that 
		\begin{align}
		R\geq  I_{ \bar{\alpha}_R }^{(2)}(P,\mathscr{W}), \quad \forall P \in \mathscr{P}(\mathcal{X}). \label{eq:Esp-a1}
		\end{align}
		Then, Eq.~\eqref{eq:Esp-a1} and the monotone increases of the map $\alpha \mapsto I_\alpha^{(2)}(P,\mathscr{W})$ yield that,
		\begin{align}
		\frac{1-\alpha}{\alpha}\left( I_\alpha^{(2)}(P,\mathscr{W}) - R\right) < 0, \quad 
		\forall P\in\mathscr{P}(\mathcal{X}),\; \text{and } \alpha \in(0, \bar{\alpha}_R).
		\end{align}
		The non-negativity of $E_\text{sp}^{(2)}(R,P)\geq 0$ ensures that the maximizer $\alpha^\star$ will not happen in the region $(0, \bar{\alpha}_R)$, and thus Eq.~\eqref{eq:Esp-a2} is evident.
		Finally, Berge's maximum theorem \cite[Section IV.3]{Ber63}, \cite[Lemma 3.1]{Psh71}, the compactness of $[\bar{\alpha}_R, 1]$, and item \ref{I2-cont_equi} in Proposition \ref{prop:I2} complete our claim:
		\begin{align}
		P \mapsto E_\text{sp}^{(2)}(R,P) = \sup_{\alpha\in [\bar{\alpha}_R,1]} \frac{1-\alpha}{\alpha}\left( I_\alpha^{(2)}(P,\mathscr{W}) -R \right) \text{ is continuous on } \mathscr{P}(\mathcal{X}). 
		\end{align}
\end{proof}

\begin{proof}[Proof of Proposition \ref{prop:Esp}-\ref{Esp-bb}]		
		The statement follows since item \ref{Esp-aa} holds for any $P\in\mathscr{P}(\mathcal{X})$ and we invoke the definition of $C_{\alpha,\mathscr{W}}$ in Eq.~\eqref{eq:radius}.
\end{proof}

\begin{proof}[Proof of Proposition \ref{prop:Esp}-\ref{Esp-cc}]
For any $R\in(C_{0,\mathscr{W}}, C_\mathscr{W})$ and $P\in\mathscr{P}_R(\mathcal{X})$, item \ref{saddle-c} in Proposition \ref{prop:saddle} shows that the optimizer $\alpha_{R,P}^\star$ is unique. 
Eq.~\eqref{eq:ss2} directly follows from item \ref{regularity-d} in Lemma~\ref{lemm:regularity}.

The saddle-point property in Proposition~\ref{prop:saddle}-\ref{saddle-a} shows that
\begin{align}
E_\text{sp}^{(2)}(R,P) = \frac{1-\alpha_{R,P}^\star}{ \alpha_{R,P}^\star } \left( I_{\alpha_{R,P}^\star}^{(2)}(P,\mathscr{W}) - R \right).
\end{align}
Further, since $E_\text{sp}^{(2)} (R,P) > 0 $ and $\alpha_{R,P}^\star \in (0,1)$ for $P\in\mathscr{P}_R(\mathcal{X})$, the above equality implies Eq.~\eqref{eq:Esp_100}.

\end{proof}


\begin{thebibliography}{10}
	\providecommand{\url}[1]{#1}
	\csname url@samestyle\endcsname
	\providecommand{\newblock}{\relax}
	\providecommand{\bibinfo}[2]{#2}
	\providecommand{\BIBentrySTDinterwordspacing}{\spaceskip=0pt\relax}
	\providecommand{\BIBentryALTinterwordstretchfactor}{4}
	\providecommand{\BIBentryALTinterwordspacing}{\spaceskip=\fontdimen2\font plus
		\BIBentryALTinterwordstretchfactor\fontdimen3\font minus
		\fontdimen4\font\relax}
	\providecommand{\BIBforeignlanguage}[2]{{%
			\expandafter\ifx\csname l@#1\endcsname\relax
			\typeout{** WARNING: IEEEtran.bst: No hyphenation pattern has been}%
			\typeout{** loaded for the language `#1'. Using the pattern for}%
			\typeout{** the default language instead.}%
			\else
			\language=\csname l@#1\endcsname
			\fi
			#2}}
	\providecommand{\BIBdecl}{\relax}
	\BIBdecl
	\normalsize
	
	\bibitem{Sha48}
	C.~E. Shannon, ``A mathematical theory of communication,'' \emph{The Bell
		System Technical Journal}, vol.~27, pp. 379--423, 1948.
	
	\bibitem{Sha59}
	------, ``Probability of error for optimal codes in a {Gaussian} channel,''
	\href{http://dx.doi.org/10.1002/j.1538-7305.1959.tb03905.x}{\emph{Bell System
			Technical Journal}},
	\href{http://dx.doi.org/10.1002/j.1538-7305.1959.tb03905.x}{vol.~38},
	\href{http://dx.doi.org/10.1002/j.1538-7305.1959.tb03905.x}{no.~3},
	\href{http://dx.doi.org/10.1002/j.1538-7305.1959.tb03905.x}{pp. 611--656},
	\href{http://dx.doi.org/10.1002/j.1538-7305.1959.tb03905.x}{may 1959}.
	
	\bibitem{Fei55}
	A.~Feinstein, ``Error bounds in noisy channels without memory,''
	\href{http://dx.doi.org/10.1109/tit.1955.1055131}{\emph{IEEE Transactions on
			Information Theory}},
	\href{http://dx.doi.org/10.1109/tit.1955.1055131}{vol.~1},
	\href{http://dx.doi.org/10.1109/tit.1955.1055131}{no.~2},
	\href{http://dx.doi.org/10.1109/tit.1955.1055131}{pp. 13--14},
	\href{http://dx.doi.org/10.1109/tit.1955.1055131}{sep 1955}.
	
	\bibitem{Fan61}
	R.~M. Fano, \emph{Transmission of Information, A Statistical Theory of
		Communications}.\hskip 1em plus 0.5em minus 0.4em\relax The MIT Press, 1961.
	
	\bibitem{Gal65}
	R.~Gallager, ``A simple derivation of the coding theorem and some
	applications,''
	\href{http://dx.doi.org/10.1109/tit.1965.1053730}{\emph{{IEEE} Transaction on
			Information Theory}},
	\href{http://dx.doi.org/10.1109/tit.1965.1053730}{vol.~11},
	\href{http://dx.doi.org/10.1109/tit.1965.1053730}{no.~1},
	\href{http://dx.doi.org/10.1109/tit.1965.1053730}{pp. 3--18},
	\href{http://dx.doi.org/10.1109/tit.1965.1053730}{jan 1965}.
	
	\bibitem{Gal68}
	\BIBentryALTinterwordspacing
	------, \emph{Information Theory and Reliable Communication}.\hskip 1em plus
	0.5em minus 0.4em\relax Wiley, 1968. [Online]. Available:
	\url{http://as.wiley.com/WileyCDA/WileyTitle/productCd-0471290483.html}
	\BIBentrySTDinterwordspacing
	
	\bibitem{SGB67}
	C.~Shannon, R.~Gallager, and E.~Berlekamp, ``Lower bounds to error probability
	for coding on discrete memoryless channels. {I},''
	\href{http://dx.doi.org/10.1016/s0019-9958(67)90052-6}{\emph{Information and
			Control}}, \href{http://dx.doi.org/10.1016/s0019-9958(67)90052-6}{vol.~10},
	\href{http://dx.doi.org/10.1016/s0019-9958(67)90052-6}{no.~1},
	\href{http://dx.doi.org/10.1016/s0019-9958(67)90052-6}{pp. 65--103},
	\href{http://dx.doi.org/10.1016/s0019-9958(67)90052-6}{jan 1967}.
	
	\bibitem{Har68}
	\BIBentryALTinterwordspacing
	E.~A. Haroutunian, ``Estimates of the error exponents for the semicontinuous
	memoryless channel,'' \emph{Problemy Peredachi Informatsii}, vol.~4, no.~4,
	pp. 37--48, 1968, (in Russian). [Online]. Available:
	\url{http://mi.mathnet.ru/eng/ppi1871}
	\BIBentrySTDinterwordspacing
	
	\bibitem{HHH07}
	E.~A. Haroutunian, M.~E. Haroutunian, and A.~N. Harutyunyan, ``Reliability
	criteria in information theory and in statistical hypothesis testing,''
	\href{http://dx.doi.org/10.1561/0100000008}{\emph{Foundations and
			Trends{\textregistered} in Communications and Information Theory}},
	\href{http://dx.doi.org/10.1561/0100000008}{vol.~4},
	\href{http://dx.doi.org/10.1561/0100000008}{no. 2--3},
	\href{http://dx.doi.org/10.1561/0100000008}{pp. 97--263},
	\href{http://dx.doi.org/10.1561/0100000008}{2007}.
	
	\bibitem{Bla74}
	R.~E. Blahut, ``Hypothesis testing and information theory,''
	\href{http://dx.doi.org/10.1109/tit.1974.1055254}{\emph{{IEEE} Transaction on
			Information Theory}},
	\href{http://dx.doi.org/10.1109/tit.1974.1055254}{vol.~20},
	\href{http://dx.doi.org/10.1109/tit.1974.1055254}{no.~4},
	\href{http://dx.doi.org/10.1109/tit.1974.1055254}{pp. 405--417},
	\href{http://dx.doi.org/10.1109/tit.1974.1055254}{jul 1974}.
	
	\bibitem{BTC88}
	A.~Ben-Tal, M.~Teboulle, and A.~Charnes, ``The role of duality in optimization
	problems involving entropy functionals with applications to information
	theory,'' \href{http://dx.doi.org/10.1007/bf00939682}{\emph{Journal of
			Optimization Theory and Applications}},
	\href{http://dx.doi.org/10.1007/bf00939682}{vol.~58},
	\href{http://dx.doi.org/10.1007/bf00939682}{no.~2},
	\href{http://dx.doi.org/10.1007/bf00939682}{pp. 209--223},
	\href{http://dx.doi.org/10.1007/bf00939682}{aug 1988}.
	
	\bibitem{CK11}
	I.~Csisz{\'a}r and J.~K{\"o}rner, \emph{Information Theory: Coding Theorems for
		Discrete Memoryless Systems}.\hskip 1em plus 0.5em minus 0.4em\relax
	Cambridge University Press ({CUP}), 2011.
	
	\bibitem{BH98}
	M.~V. Burnashev and A.~S. Holevo, ``On the reliability function for a quantum
	communication channel,'' 
	\href{http://arxiv.org/abs/quant-ph/9703013}{\emph{Problems of information transmission},
	vol.~34, no.~2, pp. 97--107, 1998}.
	
	\bibitem{Hol00}
	A.~Holevo, ``Reliability function of general classical-quantum channel,''
	\href{http://dx.doi.org/10.1109/18.868501}{\emph{{IEEE} Transaction on
			Information Theory}}, \href{http://dx.doi.org/10.1109/18.868501}{vol.~46},
	\href{http://dx.doi.org/10.1109/18.868501}{no.~6},
	\href{http://dx.doi.org/10.1109/18.868501}{pp. 2256--2261},
	\href{http://dx.doi.org/10.1109/18.868501}{2000}.
	
	\bibitem{HM16}
	H.-C. Cheng and M.-H. Hsieh, ``Concavity of the auxiliary function for
	classical-quantum channels,''
	\href{http://dx.doi.org/10.1109/TIT.2016.2598835}{\emph{{IEEE} Transactions
			on Information Theory}},
	\href{http://dx.doi.org/10.1109/TIT.2016.2598835}{vol.~62},
	\href{http://dx.doi.org/10.1109/TIT.2016.2598835}{no.~10},
	\href{http://dx.doi.org/10.1109/TIT.2016.2598835}{pp. 5960 -- 5965},
	\href{http://dx.doi.org/10.1109/TIT.2016.2598835}{2016}.
	
	\bibitem{Win99}
	A.~Winter, ``Coding theorems of quantum information theory,''  \href{http://arxiv.org/abs/quant-ph/9907077}{\textit{PhD
	Thesis, Universit{\"{a}}t Bielefeld}, 1999}. 	

	\bibitem{Dal13}
	M.~Dalai, ``Lower bounds on the probability of error for classical and
	classical-quantum channels,''
	\href{http://dx.doi.org/10.1109/tit.2013.2283794}{\emph{{IEEE} Transactions
			on Information Theory}},
	\href{http://dx.doi.org/10.1109/tit.2013.2283794}{vol.~59},
	\href{http://dx.doi.org/10.1109/tit.2013.2283794}{no.~12},
	\href{http://dx.doi.org/10.1109/tit.2013.2283794}{pp. 8027--8056},
	\href{http://dx.doi.org/10.1109/tit.2013.2283794}{dec 2013}.
	
		\bibitem{DW14b}
		M.~Dalai and A.~Winter, ``Constant Conpositions in the Sphere Packing Bound for Classical-Quantum Channels,"
		\href{https://doi.org/10.1109/tit.2017.2726555}{\textit{IEEE Transactions on Information Theory}, vol.~63, no.~9, Sept 2017}.
	
	\bibitem{Pet86}
	D.~Petz, ``Quasi-entropies for finite quantum systems,''
	\href{http://dx.doi.org/10.1016/0034-4877(86)90067-4}{\emph{Reports on
			Mathematical Physics}},
	\href{http://dx.doi.org/10.1016/0034-4877(86)90067-4}{vol.~23},
	\href{http://dx.doi.org/10.1016/0034-4877(86)90067-4}{no.~1},
	\href{http://dx.doi.org/10.1016/0034-4877(86)90067-4}{pp. 57--65},
	\href{http://dx.doi.org/10.1016/0034-4877(86)90067-4}{feb 1986}.
	
	\bibitem{Gol65}
	S.~Golden, ``Lower bounds for the {Helmholtz} function,''
	\href{http://dx.doi.org/10.1103/physrev.137.b1127}{\emph{Physical Review}},
	\href{http://dx.doi.org/10.1103/physrev.137.b1127}{vol. 137},
	\href{http://dx.doi.org/10.1103/physrev.137.b1127}{no.~4B},
	\href{http://dx.doi.org/10.1103/physrev.137.b1127}{pp. B1127--B1128},
	\href{http://dx.doi.org/10.1103/physrev.137.b1127}{feb 1965}.
	
	\bibitem{Tho65}
	C.~J. Thompson, ``Inequality with applications in statistical mechanics,''
	\href{http://dx.doi.org/10.1063/1.1704727}{\emph{Journal of Mathematical
			Physics}}, \href{http://dx.doi.org/10.1063/1.1704727}{vol.~6},
	\href{http://dx.doi.org/10.1063/1.1704727}{no.~11},
	\href{http://dx.doi.org/10.1063/1.1704727}{p. 1812},
	\href{http://dx.doi.org/10.1063/1.1704727}{1965}.
	
	\bibitem{AW12}
	Y.~Altu{\u{g}} and A.~B. Wagner, ``A refinement of the random coding bound,''
	in \href{http://dx.doi.org/10.1109/allerton.2012.6483281}{\emph{2012 50th Annual Allerton Conference on Communication, Control, and
		Computing (Allerton)}.\hskip 1em plus 0.5em minus 0.4em\relax Institute of
	Electrical and Electronics Engineers ({IEEE}), oct 2012}.
	
	\bibitem{SMF14}
	J.~Scarlett, A.~Martinez, and A.~{Guill{\'e}n i F{`a}bregas}, ``Mismatched
	decoding: Error exponents, second-order rates and saddlepoint
	approximations,''
	\href{http://dx.doi.org/10.1109/tit.2014.2310453}{\emph{{IEEE} Transactions
			on Information Theory}},
	\href{http://dx.doi.org/10.1109/tit.2014.2310453}{vol.~60},
	\href{http://dx.doi.org/10.1109/tit.2014.2310453}{no.~5},
	\href{http://dx.doi.org/10.1109/tit.2014.2310453}{pp. 2647--2666},
	\href{http://dx.doi.org/10.1109/tit.2014.2310453}{may 2014}.
	
	\bibitem{Sca14}
	J.~Scarlett, ``Reliable communication under mismatched decoding,'' \href{http://itc.upf.edu/system/files/biblio-pdf/Thesis_ONLINE.pdf}{\textit{PhD Thesis (University of Cambridge)}, 2014}.
	
	\bibitem{Hon15}
	J.~Honda, ``Exact asymptotics for the random coding error probability,'' \href{http://arxiv.org/abs/1506.03355}{\texttt{arXiv:1506.03355 [cs.IT]}}.
	
	\bibitem{BR60}
	R.~R. Bahadur and R.~R. Rao, ``On deviations of the sample mean,''
	\href{http://dx.doi.org/10.1214/aoms/1177705674}{\emph{The Annals of
			Mathematical Statistics}},
	\href{http://dx.doi.org/10.1214/aoms/1177705674}{vol.~31},
	\href{http://dx.doi.org/10.1214/aoms/1177705674}{no.~4},
	\href{http://dx.doi.org/10.1214/aoms/1177705674}{pp. 1015--1027},
	\href{http://dx.doi.org/10.1214/aoms/1177705674}{dec 1960}.
	
	\bibitem{DZ98}
	A.~Dembo and O.~Zeitouni, \emph{Large Deviations Techniques and
		Applications}.\hskip 1em plus 0.5em minus 0.4em\relax Springer, 1998.
	
	\bibitem{AW14b}
	Y.~Altu{\u{g}} and A.~B. Wagner, ``Moderate deviations in channel coding,''
	\href{http://dx.doi.org/10.1109/tit.2014.2323418}{\emph{IEEE Transactions on
			Information Theory}},
	\href{http://dx.doi.org/10.1109/tit.2014.2323418}{vol.~60},
	\href{http://dx.doi.org/10.1109/tit.2014.2323418}{no.~8},
	\href{http://dx.doi.org/10.1109/tit.2014.2323418}{pp. 4417--4426},
	\href{http://dx.doi.org/10.1109/tit.2014.2323418}{aug 2014}.
	
	\bibitem{CH17}
	H.-C. Cheng and M.-H. Hsieh, ``Moderate Deviation Analysis for Classical-Quantum Channels and Quantum Hypothesis Testing,'' 
	\href{http://dx.doi.org/10.1109/tit.2014.2299275}{\emph{{IEEE} Transactions
			on Information Theory}},
	\href{http://dx.doi.org/10.1109/tit.2014.2299275}{vol.~64},
	\href{http://dx.doi.org/10.1109/tit.2014.2299275}{no.~2},
	\href{http://dx.doi.org/10.1109/tit.2014.2299275}{pp. 1385--1403},
	\href{http://dx.doi.org/10.1109/tit.2014.2299275}{feb 2018}.

	\bibitem{CCT+16b}
	C.~T. Chubb, V.~Y.~F. Tan, and M.~Tomamichel, ``Moderate Deviation Analysis for Classical Communication over Quantum Channels,'' 
	\href{http://dx.doi.org/10.1109/10.1007/s00220-017-2971-1}{Communications in Mathematical Physics},
	\href{http://dx.doi.org/10.1109/10.1007/s00220-017-2971-1}{vol.~355},
	\href{http://dx.doi.org/10.1109/10.1007/s00220-017-2971-1}{no.~3},
	\href{http://dx.doi.org/10.1109/10.1007/s00220-017-2971-1}{pp. 1283--1315},
	\href{http://dx.doi.org/10.1109/10.1007/s00220-017-2971-1}{nov 2017}.
	\bibitem{CHDH-2018}
	H.-C.~Cheng, E.~P. Hanson, N.~Datta, and M.-H.~Hsieh, ``Non-Asymptotic Classical Data Compression with Quantum Side Information,"
	\href{http://arxiv.org/abs/1803.07505}{\texttt{arXiv:1803.07505 [quant-ph]}}.

	\bibitem{CHDH2-2018}
	H.-C.~Cheng, E.~P. Hanson, N.~Datta, and M.-H.~Hsieh, ``Duality between source coding with quantum side information and c-q channel coding,"
	\href{http://arxiv.org/abs/1809.11143}{\texttt{arXiv:1809.11143 [quant-ph]}}.

	\bibitem{CHDH3-2018}
	H.-C.~Cheng, E.~P. Hanson, N.~Datta, and M.-H.~Hsieh, ``Non-Asymptotic Joint Source-Channel Coding with Quantum Side Information,"
	(in preparation).
	
	\bibitem{Omu75}
	J.~K. Omura, ``A lower bounding method for channel and source coding
	probabilities,''
	\href{http://dx.doi.org/10.1016/s0019-9958(75)90120-5}{\emph{Information and
			Control}}, \href{http://dx.doi.org/10.1016/s0019-9958(75)90120-5}{vol.~27},
	\href{http://dx.doi.org/10.1016/s0019-9958(75)90120-5}{no.~2},
	\href{http://dx.doi.org/10.1016/s0019-9958(75)90120-5}{pp. 148--177},
	\href{http://dx.doi.org/10.1016/s0019-9958(75)90120-5}{feb 1975}.
	
	\bibitem{AW14}
	Y.~Altu{\u{g}} and A.~B. Wagner, ``Refinement of the sphere-packing bound:
	Asymmetric channels,''
	\href{http://dx.doi.org/10.1109/tit.2014.2299275}{\emph{{IEEE} Transactions
			on Information Theory}},
	\href{http://dx.doi.org/10.1109/tit.2014.2299275}{vol.~60},
	\href{http://dx.doi.org/10.1109/tit.2014.2299275}{no.~3},
	\href{http://dx.doi.org/10.1109/tit.2014.2299275}{pp. 1592--1614},
	\href{http://dx.doi.org/10.1109/tit.2014.2299275}{mar 2014}.
	
	\bibitem{EF16}
	\BIBentryALTinterwordspacing
	N.~Elkayam and M.~Feder, ``Sphere packing bound for constant composition,''
	2016, \href{http://www.eng.tau.ac.il/~elkayam/SPB_Abstract.pdf}{(in preparation)}.
	
	\bibitem{Aug69}
	U.~Augustin, ``Error estimates for low rate codes,'' \emph{Zeitschrift f{\"u}r
		Wahrscheinlichkeitstheorie und Verwandte Gebiete}, vol.~14, no.~1, pp.
	61--88, 1969.
	
	\bibitem{Aug78}
	------, ``Noisy channels,'' 1978, habilitation thesis, Universitat Erlangen (\href{http://bit.ly/2ID8h7m}{http://bit.ly/2ID8h7m}).
	
	\bibitem{Nak16a}
	B.~Nakibo\u{g}lu, ``The {Renyi} Capacity and Center,'' 
	\href{http://dx.doi.org/10.1109/TIT.2018.2861002}{\emph{{IEEE} Transactions on Information Theory}, vol.~65, no.~2, pp.	841--860, feb 2019}.
	
	\bibitem{Nak16b}
	B.~Nakibo\u{g}lu, ``The Sphere Packing Bound via {Augustin's} Method,'' 
	\href{http://10.1109/tit.2018.2882547}{\emph{{IEEE} Transactions on Information Theory}, vol.~65, no.~2, pp.	816--840, feb 2019}. 

	\bibitem{Nak18a}
	B.~Nakibo\u{g}lu, ``The The {Augustin} Capacity and Center,'' \href{http://arxiv.org/abs/1803.07937}{\texttt{arXiv:1803.07937 [cs.IT]}}.
	
	\bibitem{Nak18b}
	B.~Nakibo\u{g}lu, ``The Sphere Packing Bound For Memoryless Channels,'' \href{http://arxiv.org/abs/1804.06372}{\texttt{arXiv:1804.06372 [cs.IT]}}.

	\bibitem{CLH18}
	H.-C.~Cheng, L.~Gao, and M.-H.~Hsieh, ``Properties of Noncommutative {R\'enyi} and {Augustin} Information,"
	\href{http://arxiv.org/abs/1811.04218}{\texttt{arXiv:1811.04218 [quant-ph]}}.
	
	\bibitem{Ume62}
	H.~Umegaki, ``Conditional expectation in an operator algebra. {IV}. entropy and
	information,'' \href{http://dx.doi.org/10.2996/kmj/1138844604}{\emph{Kodai
			Mathematical Seminar Reports}},
	\href{http://dx.doi.org/10.2996/kmj/1138844604}{vol.~14},
	\href{http://dx.doi.org/10.2996/kmj/1138844604}{no.~2},
	\href{http://dx.doi.org/10.2996/kmj/1138844604}{pp. 59--85},
	\href{http://dx.doi.org/10.2996/kmj/1138844604}{1962}.
	
	\bibitem{HP91}
	F.~Hiai and D.~Petz, ``The proper formula for relative entropy and its
	asymptotics in quantum probability,''
	\href{http://dx.doi.org/10.1007/bf02100287}{\emph{Communications in
			Mathematical Physics}}, \href{http://dx.doi.org/10.1007/bf02100287}{vol.
		143}, \href{http://dx.doi.org/10.1007/bf02100287}{no.~1},
	\href{http://dx.doi.org/10.1007/bf02100287}{pp. 99--114},
	\href{http://dx.doi.org/10.1007/bf02100287}{dec 1991}.
	
	\bibitem{TH13}
	M.~Tomamichel and M.~Hayashi, ``A hierarchy of information quantities for
	finite block length analysis of quantum tasks,''
	\href{http://dx.doi.org/10.1109/tit.2013.2276628}{\emph{{IEEE} Transactions
			on Information Theory}},
	\href{http://dx.doi.org/10.1109/tit.2013.2276628}{vol.~59},
	\href{http://dx.doi.org/10.1109/tit.2013.2276628}{no.~11},
	\href{http://dx.doi.org/10.1109/tit.2013.2276628}{pp. 7693--7710},
	\href{http://dx.doi.org/10.1109/tit.2013.2276628}{nov 2013}.
	
	\bibitem{Li14}
	K.~Li, ``Second-order asymptotics for quantum hypothesis testing,''
	\href{http://dx.doi.org/10.1214/13-aos1185}{\emph{The Annals of Statistics}},
	\href{http://dx.doi.org/10.1214/13-aos1185}{vol.~42},
	\href{http://dx.doi.org/10.1214/13-aos1185}{no.~1},
	\href{http://dx.doi.org/10.1214/13-aos1185}{pp. 171--189},
	\href{http://dx.doi.org/10.1214/13-aos1185}{feb 2014}.
	
	\bibitem{TV15}
	M.~Tomamichel and V.~Y.~F. Tan, ``Second-order asymptotics for the classical
	capacity of image-additive quantum channels,''
	\href{http://dx.doi.org/10.1007/s00220-015-2382-0}{\emph{Communications in
			Mathematical Physics}},
	\href{http://dx.doi.org/10.1007/s00220-015-2382-0}{vol. 338},
	\href{http://dx.doi.org/10.1007/s00220-015-2382-0}{no.~1},
	\href{http://dx.doi.org/10.1007/s00220-015-2382-0}{pp. 103--137},
	\href{http://dx.doi.org/10.1007/s00220-015-2382-0}{may 2015}.
	
	\bibitem{ON00}
	T.~Ogawa and H.~Nagaoka, ``Strong converse and {Stein's} lemma in quantum
	hypothesis testing,'' \href{http://dx.doi.org/10.1109/18.887855}{\emph{{IEEE}
			Transaction on Information Theory}},
	\href{http://dx.doi.org/10.1109/18.887855}{vol.~46},
	\href{http://dx.doi.org/10.1109/18.887855}{no.~7},
	\href{http://dx.doi.org/10.1109/18.887855}{pp. 2428--2433},
	\href{http://dx.doi.org/10.1109/18.887855}{2000}.
	
	\bibitem{MO14b}
	M.~Mosonyi and T.~Ogawa, ``Strong Converse Exponent for Classical-Quantum Channel Coding,'' \href{https://doi.org/10.1007/s00220-017-2928-4}{\textit{Communications in Mathematical Physics}, vol.~355, no.~1, pp.~373--426, Jun 2017}.

	\bibitem{Ando79}
	T.~Ando, ``Concavity of certain maps on positive definite matrices and applications to {Hadamard} products," \href{https://doi.org/10.1016/0024-3795(79)90179-4 }{\emph{Linear Algebra and its Applications}, vol.~26, pp.~203--241, 1979}.
	
	\bibitem{Lie73}
	E.~H. Lieb, ``Convex trace functions and the {Wigner-Yanase-Dyson}
	conjecture,''
	\href{http://dx.doi.org/10.1016/0001-8708(73)90011-x}{\emph{Advances in
			Mathematics}},
	\href{http://dx.doi.org/10.1016/0001-8708(73)90011-x}{vol.~11},
	\href{http://dx.doi.org/10.1016/0001-8708(73)90011-x}{no.~3},
	\href{http://dx.doi.org/10.1016/0001-8708(73)90011-x}{pp. 267--288},
	\href{http://dx.doi.org/10.1016/0001-8708(73)90011-x}{dec 1973}.
	


	\bibitem{SW01}
	B.~Schumacher and M.~D. Westmoreland, ``Optimal signal ensembles,''
	\href{http://dx.doi.org/10.1103/physreva.63.022308}{\emph{Physical Review
			A}}, \href{http://dx.doi.org/10.1103/physreva.63.022308}{vol.~63},
	\href{http://dx.doi.org/10.1103/physreva.63.022308}{no.~2},
	\href{http://dx.doi.org/10.1103/physreva.63.022308}{Jan 2001}.
		
	\bibitem{SW97}
	B.~Schumacher and M.~D. Westmoreland, ``Sending classical information via noisy
	quantum channels,''
	\href{http://dx.doi.org/10.1103/physreva.56.131}{\emph{Physical Review A}},
	\href{http://dx.doi.org/10.1103/physreva.56.131}{vol.~56},
	\href{http://dx.doi.org/10.1103/physreva.56.131}{no.~1},
	\href{http://dx.doi.org/10.1103/physreva.56.131}{pp. 131--138},
	\href{http://dx.doi.org/10.1103/physreva.56.131}{jul 1997}.
	
	\bibitem{Hol98}
	A.~Holevo, ``The capacity of the quantum channel with general signal states,''
	\href{http://dx.doi.org/10.1109/18.651037}{\emph{{IEEE} Transaction on
			Information Theory}}, \href{http://dx.doi.org/10.1109/18.651037}{vol.~44},
	\href{http://dx.doi.org/10.1109/18.651037}{no.~1},
	\href{http://dx.doi.org/10.1109/18.651037}{pp. 269--273},
	\href{http://dx.doi.org/10.1109/18.651037}{1998}.
	
	\bibitem{HT14}
	M.~Hayashi and M.~Tomamichel, ``Correlation detection and an operational
	interpretation of the r{\'{e}}nyi mutual information,''
	\href{http://dx.doi.org/10.1063/1.4964755}{\emph{Journal of Mathematical
			Physics}}, \href{http://dx.doi.org/10.1063/1.4964755}{vol.~57},
	\href{http://dx.doi.org/10.1063/1.4964755}{no.~10},
	\href{http://dx.doi.org/10.1063/1.4964755}{p. 102201},
	\href{http://dx.doi.org/10.1063/1.4964755}{oct 2016}.
	
	\bibitem{WWY14}
	M.~M. Wilde, A.~Winter, and D.~Yang, ``Strong converse for the classical
	capacity of entanglement-breaking and {Hadamard} channels via a sandwiched
	{R{\'{e}}nyi} relative entropy,''
	\href{http://dx.doi.org/10.1007/s00220-014-2122-x}{\emph{Communications in
			Mathematical Physics}},
	\href{http://dx.doi.org/10.1007/s00220-014-2122-x}{vol. 331},
	\href{http://dx.doi.org/10.1007/s00220-014-2122-x}{no.~2},
	\href{http://dx.doi.org/10.1007/s00220-014-2122-x}{pp. 593--622},
	\href{http://dx.doi.org/10.1007/s00220-014-2122-x}{jul 2014}.

%

	
	\bibitem{SW12}
	N.~Sharma and N.~A. Warsi, ``Fundamental bound on the reliability of quantum
	information transmission,''
	\href{http://dx.doi.org/10.1103/physrevlett.110.080501}{\emph{Physical Review
			Letters}}, \href{http://dx.doi.org/10.1103/physrevlett.110.080501}{vol. 110},
	\href{http://dx.doi.org/10.1103/physrevlett.110.080501}{no.~8},
	\href{http://dx.doi.org/10.1103/physrevlett.110.080501}{feb 2013}.
	
	\bibitem{Hay06}
	M.~Hayashi, \emph{Quantum Information: An Introduction}.\hskip 1em plus 0.5em
	minus 0.4em\relax Springer, 2006.
	
	\bibitem{NH07}
	H.~Nagaoka and M.~Hayashi, ``An information-spectrum approach to classical and
	quantum hypothesis testing for simple hypotheses,''
	\href{http://dx.doi.org/10.1109/tit.2006.889463}{\emph{{IEEE} Transactions on
			Information Theory}},
	\href{http://dx.doi.org/10.1109/tit.2006.889463}{vol.~53},
	\href{http://dx.doi.org/10.1109/tit.2006.889463}{no.~2},
	\href{http://dx.doi.org/10.1109/tit.2006.889463}{pp. 534--549},
	\href{http://dx.doi.org/10.1109/tit.2006.889463}{feb 2007}.
	
	\bibitem{ANS+08}
	K.~M.~R. Audenaert, M.~Nussbaum, A.~Szko{\l}a, and F.~Verstraete, ``Asymptotic
	error rates in quantum hypothesis testing,''
	\href{http://dx.doi.org/10.1007/s00220-008-0417-5}{\emph{Communications in
			Mathematical Physics}},
	\href{http://dx.doi.org/10.1007/s00220-008-0417-5}{vol. 279},
	\href{http://dx.doi.org/10.1007/s00220-008-0417-5}{no.~1},
	\href{http://dx.doi.org/10.1007/s00220-008-0417-5}{pp. 251--283},
	\href{http://dx.doi.org/10.1007/s00220-008-0417-5}{feb 2008}.
	
	\bibitem{NS09}
	M.~Nussbaum and A.~Szko{\l}a, ``The {Chernoff} lower bound for symmetric
	quantum hypothesis testing,''
	\href{http://dx.doi.org/10.1214/08-aos593}{\emph{Annals of Statistics}},
	\href{http://dx.doi.org/10.1214/08-aos593}{vol.~37},
	\href{http://dx.doi.org/10.1214/08-aos593}{no.~2},
	\href{http://dx.doi.org/10.1214/08-aos593}{pp. 1040--1057},
	\href{http://dx.doi.org/10.1214/08-aos593}{apr 2009}.
	
	\bibitem{Nag06}
	H.~Nagaoka, ``The converse part of the theorem for quantum {Hoeffding} bound,'' \href{http://arxiv.org/quant-ph/0611289}{\texttt{arXiv:quant-ph/0611289}}.
	
	\bibitem{Bla87}
	R.~E. Blahut, \emph{Principles and practice of information theory}.\hskip 1em
	plus 0.5em minus 0.4em\relax Addison-Wesley, 1987.
	
	\bibitem{Tan14}
	V.~Y.~F. Tan, ``Asymptotic estimates in information theory with non-vanishing
	error probabilities,''
	\href{http://dx.doi.org/10.1561/0100000086}{\emph{Foundations and
			Trends{\textregistered} in Communications and Information Theory}},
	\href{http://dx.doi.org/10.1561/0100000086}{vol.~10},
	\href{http://dx.doi.org/10.1561/0100000086}{no.~4},
	\href{http://dx.doi.org/10.1561/0100000086}{pp. 1--184},
	\href{http://dx.doi.org/10.1561/0100000086}{2014}.
	
	\bibitem{AW14c}
	Y.~Altu{\u{g}} and A.~B. Wagner, ``The third-order term in the normal
	approximation for singular channels,'' in \href{http://dx.doi.org/10.1109/isit.2014.6875163}{\emph{2014 {IEEE} International
		Symposium on Information Theory}.\hskip 1em plus 0.5em minus 0.4em\relax
	Institute of Electrical and Electronics Engineers ({IEEE}), jun 2014}.
	
	\bibitem{PPV10}
	Y.~Polyanskiy, H.~V. Poor, and S.~Verdu, ``Channel coding rate in the finite
	blocklength regime,''
	\href{http://dx.doi.org/10.1109/tit.2010.2043769}{\emph{{IEEE} Trans. Inform.
			Theory}}, \href{http://dx.doi.org/10.1109/tit.2010.2043769}{vol.~56},
	\href{http://dx.doi.org/10.1109/tit.2010.2043769}{no.~5},
	\href{http://dx.doi.org/10.1109/tit.2010.2043769}{pp. 2307--2359},
	\href{http://dx.doi.org/10.1109/tit.2010.2043769}{may 2010}.
	
	\bibitem{ACM+07}
	K.~M.~R. Audenaert, J.~Calsamiglia, R.~Mu{\~{n}}oz-Tapia, E.~Bagan, L.~Masanes,
	A.~Acin, and F.~Verstraete, ``Discriminating states: The quantum {Chernoff}
	bound,''
	\href{http://dx.doi.org/10.1103/physrevlett.98.160501}{\emph{Physical Review
			Letters}}, \href{http://dx.doi.org/10.1103/physrevlett.98.160501}{vol.~98},
	\href{http://dx.doi.org/10.1103/physrevlett.98.160501}{p. 160501},
	\href{http://dx.doi.org/10.1103/physrevlett.98.160501}{apr 2007}.
	
	\bibitem{NC09}
	M.~A.~Nielsen and Issac L. Chuang, ``Quantum Computation and Quantum Information" \href{https://dx.doi.org/10.1017/cbo9780511976667}{Cambridge University Press, 2009}.
	
	\bibitem{Tom16}
	M.~Tomamichel, \emph{Quantum Information Processing with Finite
		Resources}.\hskip 1em plus 0.5em minus 0.4em\relax Springer International
	Publishing, 2016.
	
	\bibitem{Bha97}
		R.~Bhatia, \emph{Matrix Analysis}, \href{https://dx.doi.org/10.1007/978-1-4612-0653-8}{Springer New York}, 1997.

	
	\bibitem{Ber63}
	C.~Berge, \emph{Topological Spaces}.\hskip 1em plus 0.5em minus 0.4em\relax
	Oliver \& Boyd, 1963.
	
	\bibitem{Psh71}
	B.~Pshenichnyi, \emph{Necessary Conditions for an Extremum}.\hskip
	1em plus 0.5em minus 0.4em\relax CRC Press, 1971.
	
	\bibitem{Roc64}
	R.~T. Rockafellar, ``Minimax theorems and conjugate saddle-functions.''
	\href{http://dx.doi.org/10.7146/math.scand.a-10714}{\emph{Mathematica
			Scandinavica}}, \href{http://dx.doi.org/10.7146/math.scand.a-10714}{vol.~14},
	\href{http://dx.doi.org/10.7146/math.scand.a-10714}{p. 151},
	\href{http://dx.doi.org/10.7146/math.scand.a-10714}{jun 1964}.
	
	\bibitem{Roc70}
	------, \emph{Convex Analysis}.\hskip 1em plus 0.5em minus 0.4em\relax Walter
	de Gruyter {GmbH}, jan 1970.
	
	\bibitem{Wei11}
	S.~Weis, ``Quantum convex support,''
	\href{http://dx.doi.org/10.1016/j.laa.2011.06.004}{\emph{Linear Algebra and
			its Applications}}, \href{http://dx.doi.org/10.1016/j.laa.2011.06.004}{vol.
		435}, \href{http://dx.doi.org/10.1016/j.laa.2011.06.004}{no.~12},
	\href{http://dx.doi.org/10.1016/j.laa.2011.06.004}{pp. 3168--3188},
	\href{http://dx.doi.org/10.1016/j.laa.2011.06.004}{dec 2011}.
	
	\bibitem{HP14}
	F.~Hiai and D.~Petz, \emph{Introduction to Matrix Analysis and
		Applications}.\hskip 1em plus 0.5em minus 0.4em\relax Springer International
	Publishing, 2014.
	
	\bibitem{CH1}
	H.-C. Cheng and M.-H. Hsieh, ``New characterizations of matrix
	{$\Phi$}-entropies, {Poincar\'e} and {Sobolev} inequalities and an upper
	bound to {Holevo} quantity,'' \href{http://arxiv.org/abs/1506.06801}{\texttt{arXiv:1506.06801 [quant-ph]}}.
	
	\bibitem{CH2}
	H.-C. Cheng, M.-H. Hsieh, and Tomamichel, ``Exponential decay of matrix
	{$\Phi$}-entropies on {Markov} semigroups with applications to dynamical
	evolutions of quantum ensembles,'' 
	\href{http://dx.doi.org/10.1063/1.5000846}{\emph{Journal of Mathematical Physics}},
	\href{http://dx.doi.org/10.1063/1.5000846}{vol.~58},
	\href{http://dx.doi.org/10.1063/1.5000846}{no.~9},
	\href{http://dx.doi.org/10.1063/1.5000846}{p.~092202},
	\href{http://dx.doi.org/10.1063/1.5000846}{sep 2017}.
	
	\bibitem{CH16RSPA}
	H.-C. Cheng and M.-H. Hsieh, ``Characterizations of matrix and operator-valued
	{$\Phi$}-entropies, and operator {Efron–Stein} inequalities,''
	\href{http://dx.doi.org/10.1098/rspa.2015.0563}{\emph{Proceedings of the
			Royal Society of London A}},
	\href{http://dx.doi.org/10.1098/rspa.2015.0563}{p. 20150563},
	\href{http://dx.doi.org/10.1098/rspa.2015.0563}{2016}.
	
	\bibitem{Dud02}
	R.~M. Dudley, \emph{Real Analysis and Probability}.\hskip 1em plus 0.5em minus
	0.4em\relax Cambridge University Press ({CUP}), 2002.
	
	\bibitem{Hir01}
	J.-B. Hiriart-Urruty and C.~Lemar{\'{e}}chal, \emph{Fundamentals of Convex
		Analysis}.\hskip 1em plus 0.5em minus 0.4em\relax Springer Nature, 2001.
	
	\bibitem{Hol73}
	A.~S.~Holevo, ``Bounds for the quantity of information transmitted by a quantum communication channel," \href{http://mi.mathnet.ru/ppi903}{\textit{Problems of Information Transmission}, vol~9, no.~3, pp.~177--183, 1973}.
	
	\bibitem{KF75}
	A.~N.~Kolmogorov and S.~V.~Fomin", \emph{Introductory Real Analysis}. Dover Publications, New York, 1975.
	
	\bibitem{CLH18}
	Hao-Chung~Cheng, Li Gao, and Min-Hsiu~Hsieh, ``Properties of Scaled Noncommutative R\'enyi and Augustin Information", 	\href{http://arxiv.org/abs/1811.0421}{\texttt{arXiv:1811.0421 [quant-ph]}}.	
	
	\bibitem{HC}
	Hao-Chung Cheng, \emph{Error Exponent Analysis in Quantum Information Theory}, PhD Thesis (University of Technology Sydney), 2018.
	
		
\end{thebibliography}
\end{document}